\documentclass[onecolumn, 11pt]{IEEEtran}

\usepackage[utf8]{inputenc} 
\usepackage[T1]{fontenc}
\usepackage{url}
\usepackage{ifthen}
\usepackage{cite}
\usepackage[cmex10]{amsmath} 

\usepackage{graphicx} 
\usepackage{amssymb}
\usepackage{amsthm}
\usepackage{mathtools}
\usepackage{hyperref}
\usepackage{bbm}
\usepackage{tikz,lipsum,lmodern}
\usepackage[most]{tcolorbox}
\usepackage{algorithm}
\usepackage{algpseudocode}

\newtheorem{theorem}{Theorem}
\theoremstyle{definition}
\newtheorem{remark}{Remark}
\newtheorem{lemma}{Lemma}
\newtheorem{example}{Example}
\newtheorem{definition}{Definition}
\newtheorem{claim}{Claim}

\newcommand{\pare}{\eta}

\newcommand{\gdot}{g(.)}

\newcommand{\acce}{\mathcal{A}_{\pare}}

\newcommand{\bu}{\mathbf{u}}
\newcommand{\by}{\mathbf{y}}
\newcommand{\bm}{\mathbf{m}}
\newcommand{\bn}{\mathbf{n}}
\newcommand{\bc}{\mathbf{c}}
\newcommand{\be}{\mathbf{e}}
\newcommand{\br}{\mathbf{r}}
\newcommand{\mT}{\mathcal{T}}
\newcommand{\mH}{\mathcal{H}}

\newcommand{\est}{\mathsf{est}}
\newcommand{\DC}{\mathsf{DC}}
\newcommand{\AD}{\mathsf{AD}}

\newcommand{\MSE}{\mathsf{MSE}_N}
\newcommand{\PA}{\mathsf{PA}_N}

\newcommand{\advsize}{t}
\newcommand{\upperbound}{\frac{h^*_{\eta, N}(\alpha)}{4\alpha}}
\newtheorem{corollary}{Corollary}

\begin{document}
\title{Game of Coding: Sybil Resistant Decentralized Machine Learning with Minimal Trust Assumption 
\thanks{
  The work of Mohammad Ali Maddah-Ali has been partially supported by the National Science Foundation under Grant CIF-1908291.
  The work of Mohammad Ali Maddah-Ali and Hanzaleh Akbari Nodehi has been partially supported by the National Science Foundation under Grant CCF-2348638. The work of Viveck Cadambe is supported partially by the National Science Foundation under grant CCF-2506573.
  }
} 


\author{%
  \IEEEauthorblockN{Hanzaleh Akbari Nodehi$^*$,}
  \and
  \IEEEauthorblockN{Viveck R. Cadambe$^\dagger$,}
  \and
  \IEEEauthorblockN{Mohammad Ali Maddah-Ali$^*$,\\}
$^*$University of Minnesota Twin                    Cities, $^\dagger$The Georgia Institute of Technology
\thanks{
  This work has been partially presented in the 2025 IEEE International Symposium on Information Theory (ISIT 2025), Ann Arbor, MI, USA, 22-27 June, 2025.  
  }}


\maketitle

\begin{abstract}
Coding theory plays a crucial role in ensuring data integrity and reliability across various domains, from communication to computation and storage systems. However, its reliance on trust assumptions for data recovery, which requires the number of honest nodes to exceed adversarial nodes by a certain margin,  poses significant challenges, particularly in emerging decentralized systems where trust is a scarce resource. To address this, the \emph{game of coding} framework was introduced, offering insights into strategies for data recovery within incentive-oriented environments. In such environments, participant nodes are rewarded as long as the system remains functional (\emph{live}). This incentivizes adversaries to maximize their rewards (utility) by ensuring that the  decoder, as the data collector (DC), successfully recovers the data, preferably with a high estimation error.  This rational behavior is leveraged in a game-theoretic framework, where the equilibrium leads to a robust and resilient system, referred to as the \emph{game of coding}. The focus of the earliest version of the game of coding was limited to scenarios involving only two nodes. In this paper, we generalize the game of coding framework to scenarios with $N \geq 2$ nodes, exploring critical aspects of system behavior. Specifically, we (i) demonstrate that the adversary’s utility at equilibrium is non-increasing with additional adversarial nodes, ensuring \emph{no gain} for the adversary and \emph{no pain} for the DC, thus establishing the game of coding framework’s \emph{Sybil resistance}; (ii) show that increasing the number of honest nodes does not always enhance the DC’s utility, providing examples and proposing an algorithm to identify and mitigate this counterintuitive effect; and (iii) outline the optimal strategies for both the DC and the adversary, demonstrating that the system achieves enhanced liveness at equilibrium, in contrast to conventional coding theory, which results in zero liveness in trust-minimized settings.
\end{abstract}

\section{Introduction}
In communication, computing, and storage systems, coding theory finds widespread application in preserving data integrity and reliability. Moreover, it has been expanded from discrete to analog domains, facilitating approximate recovery \cite{SudanBook}.  However, whether in discrete or analog domains, coding theory heavily relies on fundamental trust assumption for data recovery. For example, consider coded computing, where a typically resource-intensive computation is distributed among $N$ nodes. Some of these nodes are classified as honest, adhering to the protocol, while others are adversarial, deviating from the protocol arbitrarily. We denote the set of honest nodes as $\mathcal{H}$ and the set of adversarial nodes as $\mathcal{T}$.
With repetition coding, error-free recovery requires that $|\mathcal{H}| \geq |\mathcal{T}|+1$. Similarly, when computation represents a Reed-Solomon $(K,N)$ code \cite{SudanBook}, error-free recovery necessitates $|\mathcal{H}| \geq |\mathcal{T}|+K$. Likewise, employing Lagrange code \cite{yu2019lagrange} for a polynomial function of degree $d$ requires $|\mathcal{H}| > |\mathcal{T}|+(K-1)d$. In all of these scenarios, the number of honest nodes surpass the number of adversarial ones by a certain margin. Similar constraints apply in analog scenarios~\cite{roth2020analog}, imposing a significant and fundamental trust assumption.
 
Decentralized machine learning (DeML) is an emerging paradigm where machine learning services are offered to the user without at a trusted central entity. DeML platforms are enabled by blockchain consensus as a state machine replication to ensure data integrity and security without an explicit trusted entity. Ethereum \cite{buterin2013ethereum} has instantiated this architecture to create a robust ecosystem of decentralized applications \cite{ruoti2019sok}, including DeML recently. 
However, the main challenge  is that the blockchain-based computing platforms suffer from inherent computational limitations \cite{croman2016scaling}.  To mitigate this, a popular solution is to utilize verifiable computing schemes \cite{zhao2021veriml}, where massive computations are outsourced to external entities which execute the computations. These entities, along with providing the computation outputs, are also required to provide a cryptographic \emph{validity proof} using verifiable computing schemes. This proof can be verified with minimal computational overhead over the blockchain \cite{thaler2022proofs}. However, the verifiable computing approach has limitations that make it challenging for DeML, particularly for trending applications like large language models. Specifically, the drawbacks are:  (i) the resource-intensive process of generating proofs \cite{liu2021zkcnn, xing2023zero, mohassel2017secureml, lee2024vcnn, garg2023experimenting, weng2021mystique} and (ii) the challenge that the most efficient verifiable computing schemes often represent computations through arithmetic circuits defined over a finite field \cite{weng2021mystique, chen2022interactive, garg2022succinct, setty2012taking},   and require exact computations, rather than approximate ones. This may not be generally applicable, thereby restricting their versatility to machine learning-types of computations.

Coding techniques offer an alternative to verifiable computing to enable DeML, by for example, performing redundant computations across a distributed set of nodes. Because of the redundancy, adversarial nodes can be detected and their computations can be corrected to ensure integrity of the computation. However, this solution requires significant trust assumptions dictated by the error correction threshold prescribed by the coding scheme used. In particular, under the classical coding-theoretic paradigm, it is required that the number of honest nodes exceed the number of adversaries by some threshold. To address these challenges and transcend the fundamental limitations of coding theory, a novel game-theoretic framework called \emph{game of coding}, has been proposed in \cite{nodehi2024game}, leveraging unique aspects of decentralized systems. Specifically, decentralized systems introduce an incentive-oriented environment where contributors are rewarded for accepted submissions. This motivates adversaries to prioritize data recoverability by the decoder, as the data collector (DC), over disrupting the system's functionality, i.e., \emph{liveness}. Based on this observation, \cite{nodehi2024game} models a game where the DC and the adversary are players with utility functions to maximize. These utility functions are based on two metrics: (1) the probability of the DC accepting the result of computations and (2) the error in estimation if the results are accepted. Focusing initially on repetition coding with two nodes, \cite{nodehi2024game} analyzes the Stackelberg equilibrium of the game. It also derives the optimal strategy for the DC to accept or reject the reported results and identifies the optimal noise distribution for the adversary to achieve equilibrium. Importantly, \cite{nodehi2024game} shows that error correction is indeed possible for the DC in a system with two nodes, even if one node is adversarial, for a broad class of utility functions.

However, the analysis in \cite{nodehi2024game} is restricted to the special case of 
$N =2$ nodes.
In this paper, we expand the game of coding framework by analyzing the cases of $N \geq 2$. Specifically, we  explore the following cases:
\begin{enumerate}
    \item {\bf The Effect of Additional Adversarial Nodes}: Intuitively, an increase in adversarial nodes might seem to enhance the adversary's flexibility to improve its utility while reducing the DC's utility. For example, assume that the DC accepts inputs if they are sufficiently close to each other and estimates the value using either the median or the mean of the accepted inputs. In both cases, the adversary can exploit the system by increasing the number of adversarial nodes, thereby amplifying the estimation error and maximizing its own utility at the expense of the DC's utility. This raises a fundamental question: Can the game of coding framework be designed to be \emph{Sybil-resistant}, effectively neutralizing the adversary's potential advantage? 
Sybil resistance refers to the system's ability to remain secure and functional even when the adversary creates multiple fake or duplicate identities (nodes) to manipulate the system. In a Sybil attack, the adversary introduces numerous nodes to overwhelm honest participants, distort decision-making processes, and degrade system performance. For the game of coding framework, ensuring Sybil resistance is critical, otherwise, it would imply that the framework cannot be effectively deployed at scale. In this paper, we show that the game of coding offers the following, strong sybil resistance property: the adversary's utility at equilibrium is non-increasing as the number of adversarial nodes grows. This result ensures \emph{no gain} for the adversary and \emph{no pain} for the DC, highlighting the robustness of the game of coding framework. This surprising finding underscores its effectiveness in mitigating adversarial influence and ensuring Sybil resistance.

 \item {\bf {The Best Strategy for the Players:}}  Leveraging our sybil resistance property, we characterize the Stackelberg equilibrium for the repetition code assuming that honest nodes have uniformly distributed additive noise, significantly generalizing our previous result that studied just one honest node and one adversarial node. Our technical approach outlines the optimal strategy for the DC regarding the acceptance or rejection of reported results, and then characterizes the optimal noise distribution for the adversary.

\item {\bf The Effect of Additional Honest Nodes:} To select the optimal strategy, the DC typically assumes a worst-case scenario with a minimum number of honest nodes. However, the actual scenario may be more favorable, with more honest nodes. The key question is: Does the DC's utility always increase in these more favorable cases? Surprisingly, in this paper, we show that this is not always true. We present examples where, despite the presence of more honest nodes, the DC's utility decreases. To address this issue, we propose an algorithm that can be effectively used  to determine whether the utility pair exhibits this counterintuitive behavior.

\end{enumerate}

The rest of the paper is structured as follows: In Section \ref{Formal Problem Settin}, we introduce the problem setting. Additionally, Section \ref{Main results} presents the main results and contributions of this paper. Finally, Section \ref{Proofs of main theorems} provides the detailed proofs of the theorems.

\subsection{Main Contributions of This Paper}
  In summary, the main contributions of this paper are as follows:
\begin{itemize}
    \item This paper explores the game of coding framework, specifically repetition coding, expanding on previous analyses to scenarios with more than two nodes. By studying these scenarios, we broaden the applicability and versatility of the game of coding framework, making it more relevant to practical cases and thereby enabling a wider range of applications.
    
    \item Surprisingly, we show that the adversary's utility is non-increasing as the number of adversarial nodes grows, yielding no gain for the adversary and no pain for the DC, despite increased flexibility for the adversary. In particular, anything achievable with multiple nodes by the adversary can also be accomplished with just a single node.  This surprising result underscores the efficacy of the game of coding in mitigating the influence of adversaries. 
    \item Interestingly, we show that increasing the number of honest nodes does not always enhance the DC’s utility. We provide examples and propose an algorithm to identify and mitigate this counterintuitive effect. Specifically, while the DC typically selects the optimal strategy under a worst-case scenario with a minimum number of honest nodes, the actual scenario may involve more honest nodes. We present examples where, despite the presence of more honest nodes in reality, the DC's utility decreases. To address this issue, we propose an algorithm to determine whether the utility pair exhibits this counterintuitive behavior.

    \item  We outline the optimal strategy for the DC regarding the acceptance or rejection of reported results and characterize the optimal noise distribution for the adversary. Notably, our analysis relies on minimal and natural assumptions about utility functions and honest players' noise, enabling broad applicability across various scenarios. Moreover, we show that at the equilibrium, the system operates with an enhanced probability of liveness, whereas conventional coding theory yields a zero likelihood of liveness in trust-minimized settings.
     
\end{itemize}

\subsection{Notation}
We use bold notation, for example, $\mathbf{y}$, to represent a random variable. We don't use any bold symbol for a specific value of this random variable. For a case that a random variable $\mathbf{y}$ follows a uniform distribution on the interval $[-u, u]$, we use the notation $\mathbf{y} \sim \text{unif}[-u,u]$, where $u \in \mathbb{R}$. 
 The notation $[a]$ is the set of $\{1,\dots,a\}$, for $a \in \mathbb{N}$. 
 For any countable set $\mathcal{S}$, we denote its size as $|\mathcal{S}|$. 
 Each variable marked with an underline, such as $\underline{z}$, represents a vector. For any set $\mathcal{S}$ and an arbitrary function $f: \mathbb{R}^* \to \mathbb{R}$, the output of $\underset{x \in \mathcal{S}}{\arg\max} ~f (x)$ is a set comprising all elements $x$ in $\mathcal{S}$ that maximize $f (x)$. Similarly we define $\underset{x \in \mathcal{S}}{\arg\min} ~f (x)$.

\section{Problem formulation}\label{Formal Problem Settin}
In this section, we describe the problem formulation, which is an adaptation of the problem formulation presented in \cite{nodehi2024game}, extending it to accommodate more than two nodes. We consider a system consisting of $N \in \mathbb{N}$ nodes and a DC. There is a random variable $\bu$,  uniformly distributed in $[-M,M]$, where $M \in \mathbb{R}$. The DC aims to estimate $\bu$; however, it does not have direct access to $\bu$. Instead, it relies on the set of nodes to receive information about $\bu$.

The set of nodes $[N]$ is divided into two groups: honest nodes, denoted by $\mH$, and adversarial nodes, denoted by $\mT$, where $\mH, \mT \subseteq [N]$, $\mathcal{H} \cap \mT = \emptyset$, and $|\mT| \leq \advsize$, for some $\advsize \in [N]$.  The subset $\mathcal{T}$ is chosen uniformly at random from the set of all subsets of size $\advsize$ from $[N]$. Neither the DC nor the nodes in $\mathcal{H}$ have knowledge of which nodes belong to $\mT$.

An honest node $h \in \mathcal{H}$ sends $\by_h$ to the DC, where
\begin{align}
    \by_h = \bu + \bn_h, \quad h \in \mathcal{H},
\end{align}
and the noise $\bn_h$ has a symmetric probability density function (PDF), $f_{\bn_h}$, over the bounded range $[-\Delta, \Delta]$ for some $\Delta \in \mathbb{R}$. Specifically, $f_{\bn_h}(-z) = f_{\bn_h}(z)$ for $z \in \mathbb{R}$, and $\Pr(|\bn_h| > \Delta) = 0$. This noise models the inaccuracy of approximate computation (e.g., sketching, randomized quantization, random projection, random sampling, etc.), where exact computation is costly. 

Each adversarial node $a \in \mT$ sends some function, possibly randomized, of $\bu$ to the DC, 
\begin{align}
    \by_a = \bu + \bn_a, \quad a \in \mT,
\end{align}
where $\bn_a$ is arbitrary noise independent of $\bu$, and $\{\bn_a\}_{a \in \mT} \sim g(\{n_a\}_{a \in \mT})$, for some joint probability density function $g(.)$. The adversary selects $g(.)$, and the DC is unaware of this choice. We assume that the adversary has access to the exact value\footnote{Note that this is a strong assumption about the adversary. In practice, the adversary may be weaker, as they could also have a noisy estimate of $\bu$. Nonetheless, our proposed scheme and main insights would still apply in such cases.} of $\bu$. Additionally, we assume that the PDF of  $\bn_h$ is public and known to all.

\begin{figure}[t]
    \centering
    \includegraphics[width=0.65\linewidth]{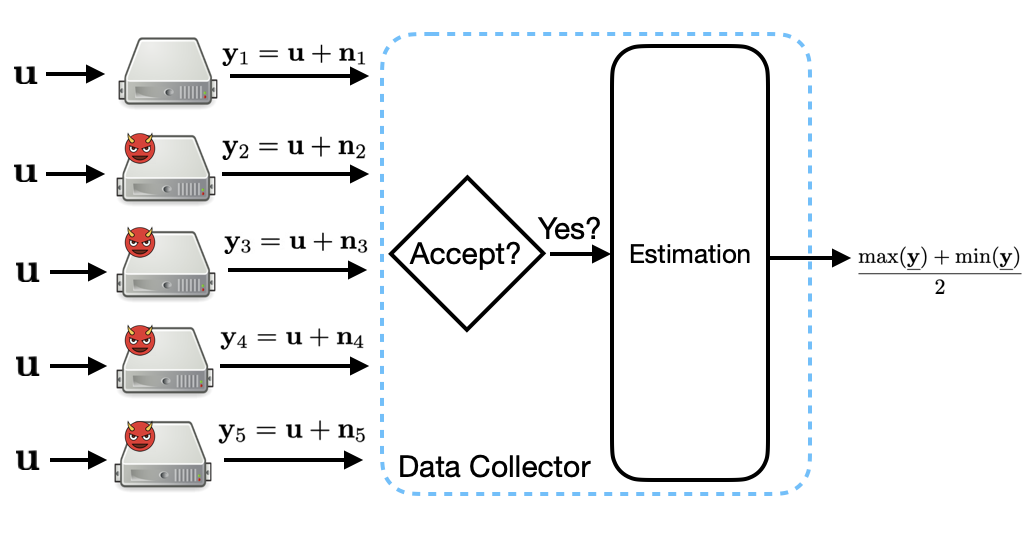}
    \caption{This figure illustrates a system with $N=5$ nodes, where $4$ of them are adversarial, shown in red. Each node's task is to output $\mathbf{u}$, but this process is subject to noise. Honest nodes experience noise given by $\bn_h$, while adversarial nodes have noise $[\mathbf{n}_a]_{a \in \mathcal{T}} \sim \gdot$, with $\gdot$ representing an arbitrary distribution independent of $\mathbf{u}$, and $\mathcal{T} = \{2,3,4,5\}$. Upon receiving the data, i.e., $\underline{\by} \triangleq (\by_1,\dots,\by_5)$, the DC checks whether $\max(\underline{\by}) - \min(\underline{\by}) \leq \eta \Delta$. If this condition is not met, the DC rejects the input; otherwise, it accepts the input and outputs $\frac{\max(\underline{\by}) + \min(\underline{\by})}{2}$ as its estimate.}
    \label{fig:General Model}
\end{figure}

The DC receives $\underline{\by} \triangleq \{\by_1, \dots, \by_N\}$ and evaluates the results according to the following steps (See Fig.~\ref{fig:General Model}):

\begin{enumerate}
    \item \textbf{Accept or Reject}: Define $\max(\underline{\by}) \triangleq \max \by_i$ and $\min(\underline{\by}) \triangleq \min \by_i$ for $i \in [N]$. The DC \emph{accepts} $\underline{\by}$ if and only if $\max(\underline{\by}) - \min(\underline{\by}) \leq \eta \Delta$ for some $\eta \geq 2$. We denote by $\acce$ the event that the inputs are accepted. Additionally, we define
    \begin{align}
        \PA \left( \gdot, \pare \right) = \Pr(\acce; \gdot).
    \end{align}
    
    \item \textbf{Estimation}: If the inputs are accepted, the DC outputs
    \begin{align}
        \est (\underline{\by}) \triangleq \frac{\max(\underline{\by}) + \min(\underline{\by})}{2}
    \end{align}
     as its estimate of $\bu$. The cost of estimation is defined as
    \begin{align}
        \MSE\left(\gdot, \pare\right) \triangleq \mathbb{E}\left[\left(\mathbf{u} - \frac{\max(\underline{\by}) + \min(\underline{\by})}{2}\right)^2 \bigg|~ \acce; \gdot\right],
    \end{align}
    which is the mean squared error (MSE) of the estimation, given that the inputs have been accepted.
\end{enumerate}

The parameter $\eta$ governs a critical trade-off between liveness and estimation accuracy. When $\eta = \infty$, the system achieves perfect liveness, accepting all inputs, yet the adversary is free to distort the estimate of $\bu$ arbitrarily, resulting in an unbounded error. In contrast, by setting $\eta = 2$, the DC enforce a strict consistency requirement, yielding accuracy comparable to a system with only  honest nodes. However, such a restrictive threshold may invite denial-of-service (DoS) attacks; a rational adversary could deliberately submit inconsistent inputs, causing the DC to reject inputs even when they might enable an acceptably accurate estimation of $\bu$.

 In particular, even if the adversaries are unwilling to adhere to the $\pm \Delta$ accuracy constraint, their inputs could still contribute to a useful estimate of $\bu$, especially if their noise magnitude is sufficiently low. In fact, in decentralized applications, such as DeML platforms and oracle networks \cite{eskandari2021sok, breidenbach2021chainlink, benligiray2020decentralized}, liveness is inherently valuable even for the adversary because rewards are contingent on the acceptance of inputs, and adversaries can only influence when the system remains live. In a non-live system, adversaries lose the opportunity to affect the outcome, aligning their incentives with the DC's objective of maintaining system liveness.

To capture these dynamics, we model the interaction as a two-player game between the DC and the adversary. We consider a two-player game between the DC and the adversary. Each player has a utility function they aim to maximize. The utility function of the DC is denoted as
\begin{align}\label{dc-utility}
    \mathsf{U}_{\mathsf{DC}}^N\left( \gdot, \pare \right) \triangleq Q_{\mathsf{DC}} \left( \MSE\left(\gdot, \pare\right), \PA \left( \gdot, \pare \right)\right),
\end{align}
where $Q_{\mathsf{DC}}: \mathbb{R}^2 \to \mathbb{R}$ is a non-increasing function with respect to its first argument and a non-decreasing function with respect to its second. More specifically, for a fixed value of the second argument, $Q_{\mathsf{DC}}(.,.)$ is monotonically non-increasing in the first argument. Similarly, for a fixed value of the first argument, $Q_{\mathsf{DC}}(.,.)$ is monotonically non-decreasing in the second argument. The DC selects its strategy from the action set $\Lambda_{\mathsf{DC}} = \left\{ \pare \mid \pare \geq 2 \right\}$.

The utility function of the adversary is denoted as
\begin{align}\label{adv-utility}
    \mathsf{U}_{\mathsf{AD}}^N\left( \gdot, \pare \right) \triangleq Q_{\mathsf{AD}} \left( \MSE\left(\gdot, \pare\right), \PA \left( \gdot, \pare \right)\right),
\end{align}
where $Q_{\mathsf{AD}}: \mathbb{R}^2 \to \mathbb{R}$ is a strictly increasing function with respect to both of its arguments. More specifically, for a fixed value of the second argument, $Q_{\mathsf{AD}}(.,.)$ is strictly increasing in the first argument. Similarly, for a fixed value of the first argument, $Q_{\mathsf{AD}}(.,.)$ is strictly increasing in the second argument.
The adversary selects its strategy from its action set, which includes all possible noise distributions\footnote{Our model also allows discrete noise distributions by using impulse (Dirac delta) functions, as is standard in signals and systems.}, i.e., $\Lambda_{\mathsf{AD}}^\advsize = \left\{ \gdot \mid ~ g: \mathbb{R}^{\advsize} \to \mathbb{R} ~ \text{is a valid PDF} \right\}$.

For this game, we aim to determine the \textbf{Stackelberg equilibrium}. Specifically, we assume that the DC acts as the leader, while the adversary plays the role of the follower in this Stackelberg game. For each action $\pare \in \Lambda_{\mathsf{DC}}$ to which the DC commits, we define the set of the adversary’s best responses as
\begin{align}
    \mathcal{B}^{\pare}_{N,t} \triangleq \underset{\gdot \in \Lambda_{\mathsf{AD}}^t}{\arg\max} ~ {\mathsf{U}}_\mathsf{AD}^N\big(\gdot, \pare \big).
\end{align}

For each $\pare \in \Lambda_{\mathsf{DC}}$ chosen by the DC, the adversary may select any element $g^*(.) \in \mathcal{B}^{\pare}_{N,t}$. Note that while each element in the set $\mathcal{B}^{\pare}_{N,t}$ provides the same utility for the adversary, different elements may result in varying utilities for the DC. The utility of the DC is at least $\underset{\gdot \in \mathcal{B}^{\pare}_{N,t}}{\min} ~ {\mathsf{U}}_\mathsf{DC}^N\left(\gdot, \pare \right)$. We define
\begin{align}\label{chaiveble_best_response}
    \Bar{\mathcal{B}}^{\pare}_{N,t} \triangleq  \underset{\gdot \in \mathcal{B}^{\pare}_{N,t}}{\arg \min} ~ {\mathsf{U}}_\mathsf{DC}^N\left(\gdot, \pare \right).
\end{align}

 \begin{definition}
      For any $\gdot \in \Lambda_{\AD}^t$ and $\pare \in \Lambda_{\DC}$, a pair $(\gdot, \pare)$ is an achievable pair, if and only if we have  $\gdot \in \Bar{\mathcal{B}}^{\pare}_{N,t}$.  
 \end{definition}
Note that if we consider a $\pare \in \Lambda_{\DC}$ and fix it, then for all $\gdot \in \Bar{\mathcal{B}}^{\pare}_{N,t}$, the value of $\mathsf{U}_\mathsf{DC}^N\left(\gdot, \pare \right)$ remains the same. 
We define
\begin{align}\label{stackleberg-eqili}
    \eta^*_{N,t}= \underset{\pare \in \Lambda_{\mathsf{DC}}}{\arg\max}  ~ {\mathsf{U}}_\mathsf{DC}^N\left(\gdot, \pare \right),
\end{align}
where in the above $\gdot$ is an arbitrary element in $\Bar{\mathcal{B}}^{\pare}_{N,t}$.
We define  Stackelberg equilibrium as follows.
\begin{definition}\label{def:stack_equil}
    For any $g^*(.) \in \Bar{\mathcal{B}}^{\pare^*_{N,t}}_{N,t}$, we call the pair of $\left(
{\mathsf{U}}_\mathsf{DC}^N\left(g^*(.), \pare^*_{N,t} \right), {\mathsf{U}}_\mathsf{AD}^N\left(g^*(.), \pare^*_{N,t} \right)
    \right)$ as an Stackelberg equilibrium. 
\end{definition}
The objective is to characterize  Stackelberg equilibrium of the game.


\begin{definition}\label{proper_paier_def}
    We call a pair of utility functions $({\mathsf{U}}_\mathsf{DC}^N(.,.), {\mathsf{U}}_\mathsf{AD}^N(.,.))$ a \emph{proper pair} is for any $t \in [N]$, $r < t$, $g^*_t(.) \in \Bar{\mathcal{B}}^{\pare^*_{N,t}}_{N,t}$, and $\tilde{g}^*(.) \in \Bar{\mathcal{B}}^{\pare^*_{N,t}}_{N,r}$, we have
    \begin{align}\label{proper_definition_utility_pair}
        {\mathsf{U}}_\mathsf{DC}^N\left(g^*_t(.), \pare^*_{N,t} \right) \leq {\mathsf{U}}_\mathsf{DC}^N\left(\tilde{g}^*(.), \pare^*_{N,t} \right).
    \end{align}
\end{definition}

Assume the maximum number of adversaries is \(t\), and the DC selects \(\pare^*_{N,t}\) based on this worst-case bound. In reality, however, the number of adversaries may be \(r<t\). The DC is unaware of this and still commits to \(\pare^*_{N,t}\). The adversary, on the other hand, knows the realized value \(r\) and also knows that the DC has committed to \(\pare^*_{N,t}\); it therefore chooses its noise distribution accordingly. In this setting, even though fewer adversaries are present than in the worst case, the DC’s realized utility may still be lower than expected (see Example~\ref{first_example_equilibrium} and the subsequent discussion). This situation puts the DC in a paradoxical position: designing the system for the worst case does not necessarily guarantee a minimum utility for the DC. Assuming the pair of utility functions is proper rules out such counterintuitive cases. 
In this work, we assume that utility pairs are proper unless stated otherwise. Additionally, in Section~\ref{Main results}, we provide an algorithm to determine whether a given pair of utilities is proper.

\section{Main Results}\label{Main results}
In this section, we present two types of results, fundamental and analytical.
\subsection{Fundamental Results}
Note that based on \eqref{adv-utility} and \eqref{stackleberg-eqili}, we have
\begin{align}
    \eta^*_{N,t} 
     =\underset{\pare \in \Lambda_{\mathsf{DC}}}{\arg\max} ~ \underset{\gdot \in \mathcal{B}^{\eta}_{N,t}}{\min} ~ Q_{\mathsf{DC}} \left( \MSE\left(\gdot, \pare\right), \PA \left( \gdot, \pare \right)\right). \label{eq:etastar}
\end{align}

At first glance, (\ref{eq:etastar}) presents a complex optimization problem due to two main factors: (i) It relies on the utility functions of both the adversary and the DC, with the only assumptions being that the adversary’s utility is strictly increasing in both arguments, while the DC’s utility is non-increasing in the first argument and non-decreasing in the second, and (ii) it involves an infinite-dimensional optimization over the entire space of possible noise distributions $g(.)$ chosen by the adversary.

Following the same approach used in \cite{nodehi2024game}, we address issue (i) by formulating an intermediate optimization problem that does not depend on the players' utility functions. For each $0 < \alpha \leq 1$, we introduce the following optimization problem:
\begin{align}\label{C_definition}
    c^{\eta}_{N,t} (\alpha) \triangleq 
    \underset{\gdot \in \Lambda_{\mathsf{AD}}^t}{\max} ~ \underset{\PA \left( \gdot, \pare \right) \geq \alpha}{\MSE\left(\gdot, \pare \right)}.
\end{align}

We demonstrate that the equilibrium of the Stackelberg game can be easily obtained using the function $c^{\eta}_{N,t}(\alpha)$, as defined in (\ref{C_definition}), through a two-dimensional optimization problem outlined in Algorithm \ref{Alg:finding_eta}. Algorithm \ref{Alg:finding_eta} takes as inputs the utility functions $Q_{\mathsf{AD}}(., .)$, $Q_{\mathsf{DC}}(., .)$, and $c^{\eta}_{N,t}(.)$ — the result of OPT. 1 — and outputs $\hat{\eta}_{N,t}$. The following theorem confirms the correctness of Algorithm \ref{Alg:finding_eta}.

\begin{theorem}\label{theorem: equivalence_two_problem}
    Let $\hat{\eta}_{N,t}$ be the output of Algorithm \ref{Alg:finding_eta}. We have
    $\eta^*_{N,t} = \hat{\eta}_{N,t}$.
\end{theorem}
The proof of this theorem can be found in Section \ref{proof:theorem: equivalence_two_problem}.

\begin{algorithm}[t]
\caption{Finding the Optimal Decision Region}
\label{Alg:finding_eta}
\begin{algorithmic}[1]
\State \textbf{Input:} Functions $Q_{\mathsf{AD}}(., .), Q_{\mathsf{DC}}(., .)$, and $c^{\eta}_{N,t}(.)$
\State \textbf{Output:} $\hat{\eta}_{N,t}$

\vspace{1em} 

\State \textbf{Step 1:}
\State Calculate the set $\mathcal{L}^{\eta}_{N,t} = \underset{0 < \alpha \leq 1 }{\arg\max} ~Q_{\mathsf{AD}}(c^{\eta}_{N,t} (\alpha), \alpha)$

\State \textbf{Step 2:}
\State Calculate $\hat{\eta}_{N,t} = \underset{\pare \in \Lambda_{\mathsf{DC}}}{\arg\max} ~ \underset{\alpha \in \mathcal{L}^{\eta}_{N,t}}{\min} ~ Q_{\mathsf{DC}} \left(c^{\eta}_{N,t} (\alpha), \alpha\right)$
\end{algorithmic}
\end{algorithm}

Our main result, stated next in Theorems \ref{lemma:C_is_constant_for_N} and \ref{theorem:one_honest_is_enough}, show that we can remove the effect of the number of adversarial nodes complexity in \eqref{C_definition}. 
\begin{theorem}\label{lemma:C_is_constant_for_N}
    For any $0 < \alpha \leq 1$, $N\geq 2$, and $t < N$ we have 
\begin{align}\label{C_is_constant_for_N}
    c^{\eta}_{N,t} (\alpha) = c^{\eta}_{N-t+1,1} (\alpha).
\end{align}
\end{theorem}
The proof of this theorem can be found in Section \ref{proof:lemma:C_is_constant_for_N}.

Based on Theorem \ref{lemma:C_is_constant_for_N},
we prove that in the presence of at least one honest node, i.e., in the minimal trust assumption scenario, the utility of the DC at equilibrium remains constant despite an increase in the number of adversarial nodes. Formally: 
\begin{theorem}\label{theorem:one_honest_is_enough}
    For any $N \geq 2$, $t < N$, $\pare^*_{N,t}= \pare^*_{N-t+1,1}$. Moreover, for any
    $g(.) \in \mathcal{B}^{\pare_{N,t}^*}_{N,t}$, and $\hat{g}(.) \in \mathcal{B}^{\pare_{N-t+1,1}^*}_{N-t+1,1}$, we have
    \begin{align}
    \mathsf{U}_{\mathsf{AD}}^N\left(g(.), \pare_{N,t}^* \right) &= \mathsf{U}_{\mathsf{AD}}^{N-t+1}\left(\hat{g}(.), \pare_{N-t+1,1}^* \right) \label{adversary_util_is_same}.
    \end{align}
    In addition,  for any
    $g(.) \in \Bar{\mathcal{B}}^{\pare_{N,t}^*}_{N,t}$, and $\hat{g}(.) \in \Bar{\mathcal{B}}^{\pare_{N-t+1,1}^*}_{N-t+1,1}$
    we have
    \begin{align}\label{util_dc_same}
        \mathsf{U}_{\mathsf{DC}}^N\left(g(.), \pare_{N,t}^* \right) &= \mathsf{U}_{\mathsf{DC}}^{N-t+1}\left(\hat{g}(.), \pare_{N-t+1,1}^* \right)
    \end{align}
\end{theorem}
The proof of this theorem can be found in Section \ref{Proofs of main theorems}.
\begin{remark}
    Theorem \ref{theorem:one_honest_is_enough} states that the adversary does not gain from increasing the number of adversarial nodes. In other words, the equilibrium does not depend on the number of adversarial nodes, as long as it is more than one. According to \eqref{adversary_util_is_same}, what the adversary can attain with multiple nodes, it can achieve with just a single node.
 Similarly, the attainable utility of the DC at the equilibrium does not decrease with an increase in the number of adversarial nodes. This  surprising result demonstrates the effectiveness of the  game   
    of coding framework in restraining the adversarial power. In addition, it expands the boundary of coding theory to  transcend the fundamental limits of trust assumption.
\end{remark}

\begin{remark}
    Note that in Theorem \ref{theorem:one_honest_is_enough}, we make no assumptions on the utility functions of the adversary and the DC, apart from the requirement that the adversary's utility is strictly increasing in both arguments and the DC's utility is non-increasing in the first and non-decreasing in the second. 
\end{remark}

\subsection{Analytical Results}

Algorithm \ref{Alg:finding_eta} is a two-dimensional optimization problem which is computationally feasible for
utility functions.  Thus,  much of the technical challenges of finding $\eta^*_{N,t}$ lies in obtaining a characterization of the function $c^{\eta}_{N,t}(.)$. 
Theorem \ref{theorem:CofJ_N_ is same} resolves this task. 
We denote the cumulative distribution function (CDF) of $\bn_h$ by $F_{\bn_h}$ and the PDF of $\bn_h$ by $f_{\bn_h}$. Recall that $\Pr (|\bn_h| > \Delta) = 0$, for some $\Delta \in \mathbb{R}$. This implies that  $F_{\bn_h}(-\Delta) = 0$, and $F_{\bn_h}(\Delta) = 1$. We assume that $F_{\bn_h}(.)$ is a strictly increasing function in $[-\Delta, \Delta]$, i.e., for all $-\Delta \leq a < b \leq \Delta$, we have 
$F_{\bn_h}(a) < F_{\bn_h}(b)$. For any $t$, we classify a noise distribution $\gdot \in \Lambda_{\mathsf{AD}}^{t}$ as a \emph{strong noise}, if $\Pr (|\bn_a| < \Delta; \gdot) = 0$, for all $a \in \mT$.

\begin{theorem}\label{theorem:CofJ_N_ is same}
     If the noise distribution of the adversary is a strong noise, then, for any $\pare \in \Lambda_{DC}$, $N \geq 2$, $0 < \alpha \leq 1$, and $t < N$, we have
    \begin{align}\label{bound_of_mean}
        c^{\eta}_{N,t} (\alpha) =  \frac{h^*_{\eta, N-t}(\alpha)}{4\alpha},
\end{align}
where for any $\ell \geq 1$, $h^*_{\eta, \ell}(q)$ is the concave envelop of the function\footnote{Note that since we assumed $F_{\bn_h}(.)$ is a strictly increasing function in $[-\Delta, \Delta]$, the inverse function of $k_{\eta, \ell}(.)$ exists.} $ h_{\eta, \ell}(q) \triangleq \nu_{\eta, \ell}(k_{\eta, \ell}^{-1} (q))$, $0 \leq q \leq 1$,  for $\nu_{\eta, \ell}(z) \triangleq \int_{z-\eta\Delta}^{\Delta} (x+z)^2w(x)\,dx$ and $k_{\eta, \ell}(z) \triangleq \int_{z-\eta\Delta}^{\Delta} w(x)  \,dx$, $(\eta-1)\Delta \leq z \leq  (\eta+1)\Delta$, and $w(x) = \ell f_{\bn_{h}}(x) (1-F_{\bn_h}(x))^{(\ell-1)}$.
\end{theorem}
    The proof of this lemma is in Section \ref{proof:theorem:CofJ_N_ is same}.

\begin{remark}
    In the process of proving Theorem \ref{theorem:CofJ_N_ is same}, we in fact characterize an optimal noise distribution $g^*(.)$ as stated in  Algorithm \ref{Alg:finding_noise}. This algorithm takes $Q_{\mathsf{AD}}(., .)$, $f_{\bn_h}(.)$, and $\eta^*_{N,t}$ as inputs and outputs $g^*(\{n_a\}_{a \in \mT})$, which is the best noise distribution of the adversary. 
\end{remark}

\begin{remark}\label{remark:special_case_uniform}
    For the specific case where \(\bn_h \sim \text{Unif}[-\Delta,\Delta]\), Theorem~\ref{theorem:CofJ_N_ is same} holds without requiring the adversary's noise to be restricted to the class of strong noise distributions. The proof is available in Appendix~\ref{proof:remark:special_case_uniform}.
\end{remark}

\begin{algorithm}[t]
\caption{Characterizing the Optimal Distribution for Adversary}
\label{Alg:finding_noise}
\begin{algorithmic}[1]
\State \textbf{Input:} The functions $Q_{\mathsf{AD}}(., .)$, $f_{\bn_h}(n_h)$, $\eta^*_{N,t}$, and $N,t$. For simplicity here, in the description of algorithm we use $\eta^*$ instead of $\eta^*_{N,t}$.
\State \textbf{Output:} $g^*(\{n_a\}_{a \in \mT})$
\State Let $\ell \triangleq N-t$
\State Let $w (x) \triangleq \ell f_{\bn_{h}}(x) (1-F_{\bn_h}(x))^{(\ell-1)}$
\State Let  $k_{\eta^*, \ell}(z) \triangleq \int_{z-\eta^*\Delta}^{\Delta} w(x)  \,dx$ and $\nu_{\eta^*, \ell}(z) \triangleq \int_{z-\eta^*\Delta}^{\Delta} (x+z)^2w(x)\,dx$, for $z \in [(\eta^*-1)\Delta, (\eta^*+1)\Delta]$.

\State Let $h_{\eta^*, \ell}(q) \triangleq \nu_{\eta^*, \ell}(k_{\eta}^{-1} (q))$ and $h^*_{\eta^*, \ell}(q)$ be the concave envelop of $h_{\eta^*, \ell}(q)$, for  $q \in [0,1]$.

\vspace{1em} 
\State \textbf{Step 1:} 
\State Calculate $\mathcal{L}^{\eta^*}_{N,t} = \underset{0 < \alpha \leq 1 }{\arg\max} ~Q_{\mathsf{AD}}(\frac{h^*_{\eta^*, \ell}(\alpha)}{4\alpha}, \alpha)$, and choose $\alpha$ as an arbitrary element of  $ \mathcal{L}^{\eta^*}_{N,t}$.

\vspace{1em} 
\State \textbf{Step 2:} 
\If {$h^*_{\eta^*, \ell}(\alpha) = h_{\eta^*, \ell}(\alpha)$}
    \State Let $z_1 \triangleq k^{-1}_{\eta^*, \ell}(\alpha)$
    \State Output $f^*(z) = \frac{1}{2}\delta(z+z_1) + \frac{1}{2}\delta(z-z_1)$
\Else
    \State Find $q_1 < \alpha < q_2$, such that $h^*_{\eta^*, \ell}(q_1) = h_{\eta^*, \ell}(q_1)$ and $h^*_{\eta^*, \ell}(q_2) = h_{\eta^*, \ell}(q_2)$,
    and for all $q_1 \leq q \leq q_2$, we have
    \begin{align*}
        h^*_{\eta^*, \ell}(q) = \frac{h_{\eta^*, \ell}(q_2) - h_{\eta^*, \ell}(q_1)}{q_2 - q_1} (q - q_1) + h_{\eta^*, \ell}(q_1).
    \end{align*}
    \State Let $z_1 \triangleq k^{-1}_{\eta^*, \ell}(q_1)$, $z_2 \triangleq  k^{-1}_{\eta^*, \ell}(q_2)$, $\beta_1 \triangleq  \frac{q_2 -\alpha }{2(q_2 - q_1)}$, and $\beta_2 \triangleq  \frac{\alpha - q_1}{2(q_2 - q_1)}$.
    \State Output $f^*(z) = \beta_1 \delta(z+z_1) +\beta_2 \delta(z+z_2) +\beta_1 \delta(z-z_1) +\beta_2 \delta(z-z_2)$
\EndIf
\vspace{1em} 
\State \textbf{Step 3:} 
\State $g^*(\{n_a\}_{a \in \mT}) = \underset{{a \in \mT}}{\Pi}\delta(z-n_a)f^*(z)$.
\end{algorithmic}
\end{algorithm}

To clarify the results, let's consider an illustrative example. 
\begin{example}\label{first_example_equilibrium}
Consider \( N=20 \), \( t=19 \), \( \Delta=1 \), and let \( \mathbf{n}_h \) follow a uniform distribution over \( [-\Delta, \Delta] \). We analyze two cases to determine the Stackelberg equilibrium, including the optimal strategy \( \eta^*_{20,19} \) for the DC and \( g^*(\cdot) \) for the adversary.

In {\bf Case 1}, the utility functions are defined as 
\begin{align}
\text{{\bf Case 1}}:  
 \mathsf{U}_{\mathsf{AD}}( \gdot, \pare ) &= \log \mathsf{MMSE} + 0.7 \log \mathsf{PA}, \\
  \mathsf{U}_{\mathsf{DC}}( \gdot, \pare ) &= -\mathsf{MMSE} + 15 \log \mathsf{PA} . 
\end{align}
Using Theorem~\ref{theorem:CofJ_N_ is same}, we compute \( c^{\eta}_{20,19}(\cdot) \) for \( \eta \in \{2, 2.25, 2.5, \dots, 8\} \), as shown in Fig.~\ref{fig:Finding_equilibrium}. Notably, this computation is independent of the specific utility functions. Next, employing Algorithm~\ref{Alg:finding_noise}, we determine the adversary's best response. For each \( \eta \), the set $ \mathcal{L}^{\eta}_{20,19} = \underset{0 < \alpha \leq 1}{\arg\max} ~Q_{\mathsf{AD}}(c^{\eta}_{20,19}(\alpha), \alpha) $ is obtained, where each member is indicated by a green circle in Fig.~\ref{fig:Finding_equilibrium}. Finally, applying Theorem~\ref{theorem: equivalence_two_problem}, the DC’s optimal strategy is determined as $ \eta^*_{20,19} = \underset{\pare \in \Lambda_{\mathsf{DC}}}{\arg\max} \underset{\alpha \in \mathcal{L}^{\eta}_{20,19}}{\min} (-c^{\eta}_{20,19}(\alpha) + 15 \log\alpha ) $. For Case 1, the equilibrium is \( \eta^*_{20,19} = 5.5 \), with \( \mathsf{PA}=0.657 \) and \( \mathsf{MMSE}=7.68 \). The optimal adversarial noise \( g^*(\cdot) \) can be computed using Algorithm \ref{Alg:finding_noise}, and the equilibrium is represented by the black circle in Fig.~\ref{fig:Finding_equilibrium}.

In {\bf Case 2}, let 
\begin{align}
\text{{\bf Case 2}}:  
 \mathsf{U}_{\mathsf{AD}}( \gdot, \pare ) &= \log \mathsf{MMSE} + 0.25 \log \mathsf{PA}, \\
  \mathsf{U}_{\mathsf{DC}}( \gdot, \pare ) &= -\mathsf{MMSE} + 50 \mathsf{PA}. 
\end{align}
Following the same procedure as in Case 1, we find \( \eta^*_{20,19} = 2.75 \), resulting in the equilibrium \( (\mathsf{PA}, \mathsf{MMSE}) = (0.177, 4.454) \), depicted as a yellow circle in Fig.~\ref{fig:Finding_equilibrium}. For each \( \eta \), the corresponding \( \mathcal{L}^{\eta}_{20,19} \) is marked by a red circle on the respective curve of \( c^{\eta}_{20,19}(\cdot) \).

\begin{figure}[t]
  \centering
  \includegraphics[width=0.65\linewidth]{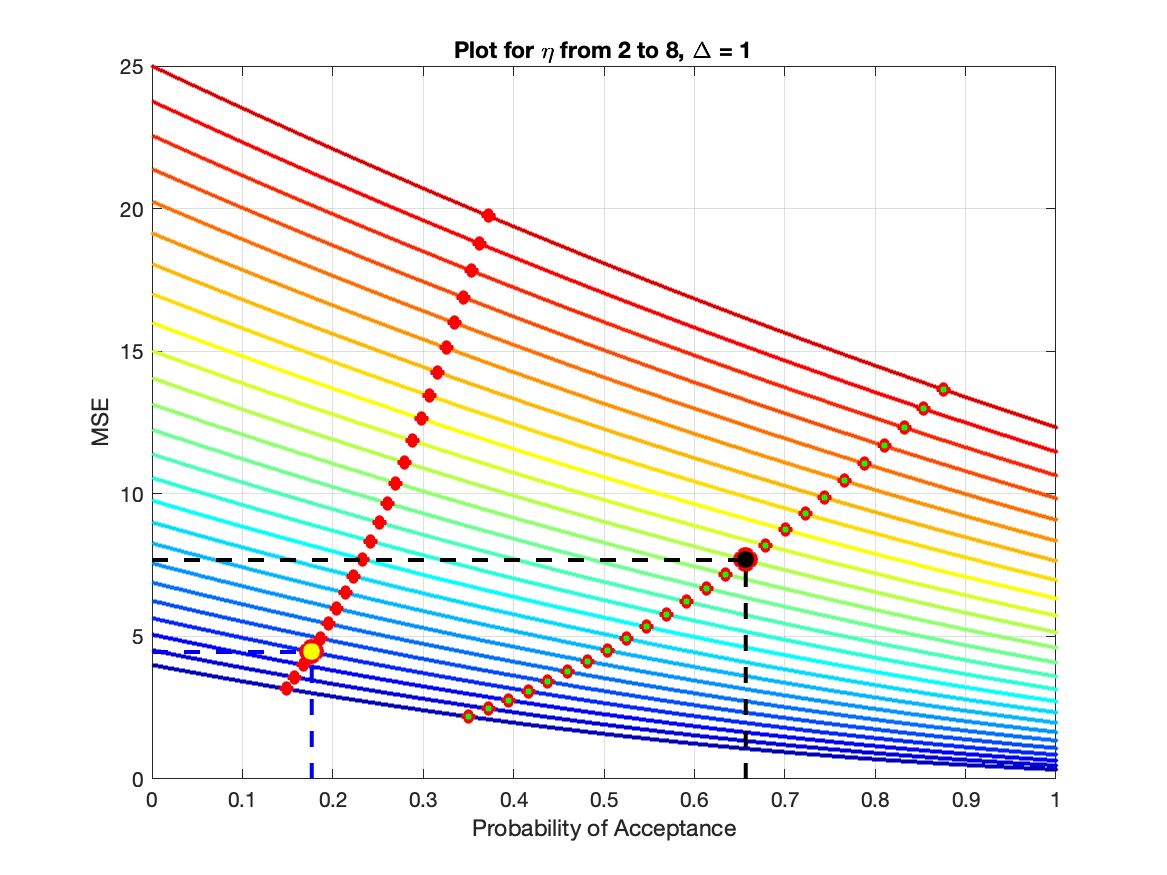}
  \caption{The curves of \( c^{\eta}_{20,19}(\cdot) \) for \( \eta \in \{2, 2.25, 2.5, \dots, 8\} \). The lower-blue curve represents \( c^{2}_{20,19}(\cdot) \), and the upper-red curve represents \( c^{8}_{20,19}(\cdot) \). The green and red circles correspond to \( \mathcal{L}^{\eta}_{20,19} \) for Case 1 and Case 2, respectively. The black and yellow circles indicate the equilibria for Case 1 and Case 2.}
  \label{fig:Finding_equilibrium}
\end{figure}
\end{example}

\begin{figure}[t]
  \centering
  \includegraphics[width=0.65\linewidth]{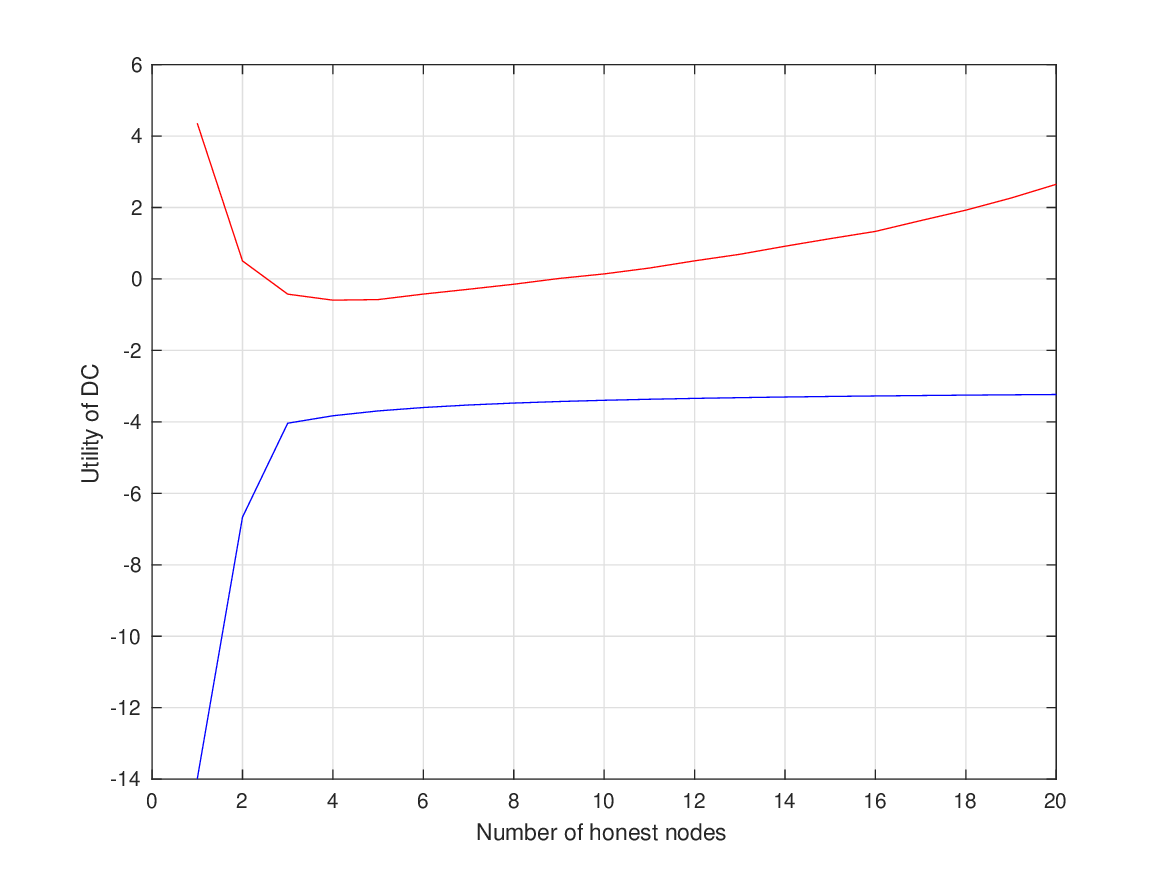}
  \caption{Utility of the DC for different numbers of honest nodes in Example \ref{first_example_equilibrium}. Here, the DC commits to \(\eta^*_{20,19}\), computed under the worst-case assumption of one honest node and \(t=19\) adversarial nodes. In practice, however, there may be more honest nodes (and hence fewer adversaries), even though the DC is unaware of the realized composition. Each curve therefore starts at the worst-case point (one honest node and \(t=19\)). The blue curve corresponds to Case~1, and the red curve corresponds to Case~2. The important observation is that the DC’s utility does not necessarily increase when the number of honest nodes exceeds what the DC perceives. In Case~1 (blue curve), the utility functions are proper, whereas in Case~2 (red curve), they are not.
 }

  \label{fig:non_proper_pair}
\end{figure}

In Example \ref{first_example_equilibrium}, for both cases, the DC considers the worst-case scenario, where the number of adversarial nodes is \( t=19 \). Based on this scenario, the DC selects the corresponding optimal strategy, \( \eta^*_{N,t} \). However, a fundamental question is: What happens if the system is in a more favorable situation, such as having 10 honest nodes instead of just one in a trust-minimized scenario? Intuitively, one might expect the utility of the DC at the equilibrium to always increase as the number of honest nodes grows. Surprisingly, this is not always the case.

Let us revisit the two cases in Example \ref{first_example_equilibrium}. In Case 1, the DC calculates \( \eta^*_{N,t} = 5.5 \) by assuming the worst-case scenario with the maximum number of adversarial nodes. However, in practice, the number of adversarial nodes may be less, with \( 0 \leq t \leq 19 \). While the DC is unaware of the exact number of adversarial nodes, the adversary knows this information and can adjust its strategy based on the value of \( \eta \) committed by the DC. This adjustment leads to a different outcome for \( \mathsf{PA} \) (probability of acceptance) and \( \mathsf{MSE} \) (mean squared error) compared to the values calculated in Example \ref{first_example_equilibrium}. Based on Lemma \ref{lemmaJSC}, the new values of \( \mathsf{PA} \) and \( \mathsf{MMSE} \) depend on the adversary’s utility function and the curvature of \( c^{\eta}_{N,t} \), as defined in \eqref{C_definition}. These changes in \( \mathsf{PA} \) and \( \mathsf{MMSE} \) may result in either an increase or a decrease in the DC’s utility.

For Case 1 of Example \ref{first_example_equilibrium}, as shown by the blue curve in Figure \ref{fig:non_proper_pair}, the utility of the DC increases when the number of honest nodes exceeds what the DC perceives. However, in Case 2, as shown by the red curve in Figure \ref{fig:non_proper_pair}, it is not the case.

 To address this counterintuitive behavior, the DC must verify whether the utility pair (its utility and the adversary's utility) is \emph{proper}, as defined in Definition \ref{proper_paier_def}. This can be achieved by first considering the worst-case scenario with the maximum number of adversarial nodes, calculating the optimal strategy \( \eta^*_{N,t} \) for the DC, and then evaluating whether an increase in the number of honest nodes, combined with the adversary's best response to the new situation, leads to an increase in the DC's utility.


\section{Proof of Theorem \ref{theorem: equivalence_two_problem}}\label{proof:theorem: equivalence_two_problem}
We begin by observing that (\ref{eq:etastar}) can be reformulated as 
\begin{align}\label{seprating_utilities}
    \eta^*_{N,t} 
        =\underset{\pare \in \Lambda_{\mathsf{DC}}}{\arg\max} ~ \underset{(\beta, \alpha) \in \mathcal{J}^{\eta}_{N,t}}{\min} ~ Q_{\mathsf{DC}} \left(\beta, \alpha\right),
\end{align}
where
\begin{align}\label{J_definition}
    \mathcal{J}^{\eta}_{N,t} \overset{\Delta}{=} \bigl\{\bigl( 
    \MSE\big(g^*(.),  \eta \big), 
    \PA \left( g^*(.), \pare \right)\big)
    \big) \bigl| ~g^*(.) \in \mathcal{B}^{\pare}_{N,t}\bigr\}.
\end{align} 
To prove Theorem \ref{theorem: equivalence_two_problem}, we  we first show the following lemma.
\begin{lemma}
\label{lemmaJSC}
Let 
\begin{align}\label{defining_the_Set_C_eta}
    \mathcal{C}^{\pare}_{N,t} \triangleq \bigg\{ \left(c^{\eta}_{N,t}(\alpha), \alpha\right) ~|~ 0 < \alpha \leq 1 \bigg\},
\end{align}
where $c^{\eta}_{N,t}(\alpha)$ is defined in \eqref{C_definition}.
For any $\pare \in \Lambda_{\mathsf{AD}}$, $\mathcal{J}^{\pare}_{N,t} \subseteq \mathcal{C}^{\pare}_{N,t}$. 
\end{lemma}
\begin{proof}
    Consider an arbitrary element $(\beta, \alpha) \in \mathcal{J}^{\eta}_{N,t}$. We aim to show $(\beta, \alpha) \in \mathcal{C}^{\eta}_{N,t}$. Based on the definition of $\mathcal{C}^{\eta}_{N,t}$ in \eqref{defining_the_Set_C_eta},  we show that for any $g^*(.) \in \mathcal{B}^{\pare}_{N,t}$, and $\alpha \triangleq \PA \left( g^*(.), \pare \right)$, the distribution $g^*(.)$ is one of the optimal solutions of the following maximization problem
\begin{align}\label{maximization_problem_for_C}
    \underset{\gdot \in \Lambda_{\mathsf{AD}}^t}{\max} ~ \underset{\PA \left( \gdot, \pare \right) \geq \alpha}{\MSE\left(\gdot, \pare \right)}.
\end{align}
We prove this claim by contradiction. Assume that this is not the case. Let\footnote{In the case that the optimization problem \eqref{definition_of_g_1} has more than one solution, just pick one of them randomly as the noise function $g_1(.)$.}
    \begin{align}\label{definition_of_g_1}
        g_1(.) = \underset{\gdot \in \Lambda_{\mathsf{AD}}^t}{\arg\max} ~ \underset{\PA \left( \gdot, \pare \right) \geq \alpha}{\MSE\left(\gdot, \pare \right)}.
    \end{align}
    Based on our contradictory hypothesis that $g^*(.)$ is not an optimal solution of \eqref{maximization_problem_for_C}, one can verify that 
    \begin{align}\label{comparing_g_1_and_star}
        \MSE\big(g_1(.), \pare \big) > \MSE\big( g^*(.),  \pare \big)
    \end{align}
    Therefore, we have
    \begin{align}\label{JSC_contradiction}
        \mathsf{U}_{\mathsf{AD}}^N\left( g_1(.), \pare \right) &= Q_{\mathsf{AD}} \big(\MSE\big( g_1(.), \pare\big), \PA\big(g_1(.), \pare\big)\big)\nonumber \\
        &\overset{(a)}{\geq}
        Q_{\mathsf{AD}} \big(\MSE\big( g_1(.), \pare\big), \alpha \big)\nonumber \\
        &\overset{(b)}{>} 
        Q_{\mathsf{AD}} \big(\MSE\big( g^*(.),  \pare \big), \alpha \big) \nonumber \\
        &\overset{(c)}{=} Q_{\mathsf{AD}} \big(\MSE\big( g^*(.),  \pare \big), \PA \left( g^*(.), \pare \right)\big)\nonumber \\
        &=\mathsf{U}_{\mathsf{AD}}^N\left( g^*(.),  \pare \right),
    \end{align}
    where (a) follows from the facts that based on \eqref{definition_of_g_1}, $\PA\big(g_1(.), \pare\big) \geq \alpha$ and also that $Q_{\mathsf{AD}}(., .)$ is a strictly increasing function with respect to its second argument, (b) follows from \eqref{comparing_g_1_and_star} and also that $Q_{\mathsf{AD}}(., .)$ is a strictly increasing function with respect to its first argument, and (c) follows from the definition of $\alpha$.

    However, \eqref{JSC_contradiction} implies that $\mathsf{U}_{\mathsf{AD}}^N\left( g_1(.), \pare \right) > \mathsf{U}_{\mathsf{AD}}^N\left( g^*(.),  \pare \right)$, which is a contradiction, since $g^*(.) \in \mathcal{B}^{\pare}_{N,t}$. Thus, our first contradictory hypothesis is wrong, and the statement of this lemma is valid.
\end{proof}
Now we prove Theorem \ref{theorem: equivalence_two_problem}.
    Recall that 
\begin{align}
    \mathcal{L}^{\eta}_{N,t} =\underset{0 < \alpha \leq 1 }{\arg\max} ~Q_{\mathsf{AD}}(c^{\eta}_{N,t} (\alpha), \alpha).
\end{align}
Let 
\begin{align}\label{L_definition}
    \mathcal{K}^{\eta}_{N,t} = \{ (c^{\eta}_{N,t}(\alpha), \alpha)~|~ \alpha \in \mathcal{L}^{\eta}_{N,t} \}
\end{align}
Based in \eqref{J_definition}, to complete the proof we show
\begin{align}\label{equivalence of k and j}
    \mathcal{J}^{\eta}_{N,t} = \mathcal{K}^{\eta}_{N,t}
\end{align}
 we first show that $\mathcal{J}^{\eta}_{N,t} \subseteq \mathcal{K}^{\eta}_{N,t}$, and then prove $\mathcal{K}^{\eta}_{N,t} \subseteq \mathcal{J}^{\eta}_{N,t}$.

\subsection{\texorpdfstring{Proof of $\mathcal{J}^{\eta}_{N,t} \subseteq \mathcal{K}^{\eta}_{N,t}$}{x}} Consider $(\beta, \alpha) \in \mathcal{J}^{\eta}_{N,t}$.  We claim that $(\beta, \alpha) \in \mathcal{K}^{\eta}_{N,t}$, and show this by contradiction. Assume, as a contradictory hypothesis, that it is not the case. According to Lemma \ref{lemmaJSC} $(\beta, \alpha) \in \mathcal{C}^{\eta}_{N,t}$. Since we assumed that $(\beta, \alpha) \notin \mathcal{K}^{\eta}_{N,t}$, based on the definition in \eqref{L_definition},  there exists $(b,a) \in \mathcal{K}^{\eta}_{N,t}$, such that we have \begin{align}\label{JnsL_contradiction_assuption}
        Q_{\mathsf{AD}}(b,a) > Q_{\mathsf{AD}}(\beta, \alpha).
    \end{align}

    Since $(\beta, \alpha) \in \mathcal{J}^{\eta}_{N,t}$, based on the definition in \eqref{J_definition} there exists a noise distribution $g_{\alpha}(.) \in \mathcal{B}^{\pare}_{N,t}$, where 
    \begin{align}\label{definition of beta alpha}
    \left(\beta, \alpha\right) &= \bigl( 
    \MSE\big(g_{\alpha}(.), \pare \big), \PA \big(g_{\alpha}(.), \pare
    \big)
    \big).
    \end{align}
    
    Similarly for $(b,a) \in \mathcal{K}^{\eta}_{N,t}$,
    based on the definition in \eqref{L_definition}, 
    we have $(b,a) \in \mathcal{C}^{\eta}_{N,t}$. Based on the definition in \eqref{defining_the_Set_C_eta},
    there exists a noise distribution $g_b(.) \in \Lambda_{\mathsf{AD}}$, where 
    \begin{align}\label{a_to_PA_relationship}
        \PA \left( g_b(.), \pare \right) \geq a,
    \end{align}
    and $g_b(.)$ is one of the optimal solutions of the following
    maximization problem
    \begin{align}\label{gb_to_opt}
     \underset{\gdot \in \Lambda_{\mathsf{AD}}^t}{\max} ~ \underset{\PA \left( \gdot, \pare \right) \geq a}{\MSE\left(\gdot, \pare \right)},
    \end{align}
    and we have
    \begin{align}\label{b_to_MMSE_relation}
        b = \MSE\left(g_{b}(.), \pare \right).
    \end{align}
    One can verify that we have
    \begin{align}\label{JnsL_contradiction_point}
        \mathsf{U}_{\mathsf{AD}}^N\left(g_b(.), \pare \right) &= Q_{\mathsf{AD}}\big( 
    \MSE\left(g_{b}(.), \pare \right), \PA \left( g_b(.), \pare \right)
    \big) \nonumber \\
    &\overset{(a)}{=} Q_{\mathsf{AD}}\big( b, \PA \left( g_b(.), \pare \right)
    \big) \nonumber \\
    &\overset{(b)}{\geq} Q_{\mathsf{AD}}\left( b, a \right) \nonumber \\
    &\overset{(c)}{>} Q_{\mathsf{AD}}\left( \beta, \alpha \right) \nonumber \\
    &\overset{(d)}{=} Q_{\mathsf{AD}}\big( 
    \MSE\left(g_{\alpha}(.), \pare \right), \PA \left( g_{\alpha}(.), \pare \right)
    \big) \nonumber \\
    &=\mathsf{U}_{\mathsf{AD}}^N\left(g_{\alpha}(.), \pare \right)
    \end{align}
    where (a) follows from \eqref{b_to_MMSE_relation}, (b) follows from \eqref{a_to_PA_relationship} and  $Q_{\mathsf{AD}}(.,.)$ is a strictly increasing function with respect to its second argument, (c) follows from \eqref{JnsL_contradiction_assuption}, (d) follows from \eqref{definition of beta alpha}. However, \eqref{JnsL_contradiction_point} implies that $\mathsf{U}_{\mathsf{AD}}^N\left(g_b(.), \pare \right) > \mathsf{U}_{\mathsf{AD}}^N\left(g_{\alpha}(.), \pare \right)$, which is a contradiction, since $g_{\alpha}(.) \in \mathcal{B}^{\pare}_{N,t}$. Therefore, our first assumption was wrong, and for each $(\beta, \alpha) \in \mathcal{J}^{\eta}_{N,t}$ we have $(\beta, \alpha) \in \mathcal{K}^{\eta}_{N,t}$.
    
    \subsection{\texorpdfstring{Proof of $\mathcal{K}^{\eta}_{N,t} \subseteq \mathcal{J}^{\eta}_{N,t}$}{x}}
    We prove this by contradiction.  Assume that this is not the case. Consider $(b,a) \in \mathcal{K}^{\eta}_{N,t}$, where $(b,a) \notin \mathcal{J}^{\eta}_{N,t}$, and the noise distribution $g_b(.) \in \Lambda_{\mathsf{AD}}$, where \eqref{a_to_PA_relationship}, \eqref{gb_to_opt}, and \eqref{b_to_MMSE_relation} holds.
    Now consider a $(\beta, \alpha) \in \mathcal{J}^{\eta}_{N,t}$, and the noise distribution $g_{\alpha}(.) \in \mathcal{B}^{\pare}_{N,t}$, where \eqref{definition of beta alpha} holds.
    We note that based on Lemma \ref{lemmaJSC}, $(\beta, \alpha) \in \mathcal{C}^{\eta}_{N,t}$. Since $(b,a) \in \mathcal{K}^{\eta}_{N,t}$, based on the definition \eqref{J_definition} we have
    \begin{align}\label{ba_tobetaalpha_relation}
       Q_{\mathsf{AD}}\left( b,a \right) \geq Q_{\mathsf{AD}}\left( \beta, \alpha \right).
    \end{align}
    Consider the following chain of inequalities: 
    \begin{align}\label{jset_to_lset}
        \mathsf{U}_{\mathsf{AD}}^N\left(g_b(.), \pare \right) &= Q_{\mathsf{AD}}\bigl( 
    \MSE\left(g_{b}(.), \pare \right), \PA \left(g_{b}(.), \pare
    \right)
    \big) \nonumber \\
    &\overset{(a)}{=} Q_{\mathsf{AD}}\big( b, \PA \left(g_{b}(.), \pare
    \right)
    \big) \nonumber \\
    &\overset{(b)}{\geq} Q_{\mathsf{AD}}\left( b, a \right) \nonumber \\
    &\overset{(c)}{\geq} Q_{\mathsf{AD}}\left( \beta, \alpha \right) \nonumber \\
    &\overset{(d)}{=} Q_{\mathsf{AD}}\big( 
    \MSE\left(g_{\alpha}(.), \pare \right), \PA \left( g_{\alpha}(.), \pare \right)
    \big) \nonumber \\
    &=\mathsf{U}_{\mathsf{AD}}^N\left(g_{\alpha}(.), \pare \right),
    \end{align}
    where (a) follows from \eqref{b_to_MMSE_relation}, (b) follows from \eqref{a_to_PA_relationship} and  $Q_{\mathsf{AD}}(.,.)$ is a strictly increasing function with respect to its second argument, (c) follows from \eqref{ba_tobetaalpha_relation}, (d) follows from \eqref{definition of beta alpha}.

    However, \eqref{jset_to_lset} implies that $\mathsf{U}_{\mathsf{AD}}^N\left(g_b(.), \pare \right) \geq \mathsf{U}_{\mathsf{AD}}^N\left(g_{\alpha}(.), \pare \right)$. On the other hand, since  $g_{\alpha}(.) \in \mathcal{B}^{\pare}_{N,t}$, we have $\mathsf{U}_{\mathsf{AD}}^N\left(g_{\alpha}(.), \pare \right) \geq \mathsf{U}_{\mathsf{AD}}^N\left(g_b(.), \pare \right)$. Thus $\mathsf{U}_{\mathsf{AD}}^N\left(g_b(.), \pare \right) = \mathsf{U}_{\mathsf{AD}}^N\left(g_{\alpha}(.), \pare \right)$.
    Consequently, based on \eqref{jset_to_lset}, we have $Q_{\mathsf{AD}}\left( b, \PA \left(g_{b}(.), \pare
    \right)
    \right) = Q_{\mathsf{AD}}\left( b, a \right)$. Since $Q_{\mathsf{AD}}(.,.)$ is a strictly increasing function with respect to its second argument, it implies that $\PA \left(g_{b}(.), \pare \right) = a$. Thus, we have
    \begin{align}
        (b,a) = \bigl( 
    \MSE\left(g_{b}(.), \pare \right), \PA \left(g_{b}(.), \pare
    \right)
    \big),
    \end{align}
    On the other hand, $g_{\alpha}(.) \in \mathcal{B}^{\pare}_{N,t}$ and $\mathsf{U}_{\mathsf{AD}}^N\left(g_b(.), \pare \right) = \mathsf{U}_{\mathsf{AD}}^N\left(g_{\alpha}(.), \pare \right)$, which implies that $g_b(.) \in \mathcal{B}^{\pare}_{N,t}$. Therefore, by definition it implies that $(b,a) \in \mathcal{J}^{\eta}_{N,t}$. This is against our first assumption that $(b,a) \notin \mathcal{J}^{\eta}_{N,t}$, which is a contradiction. Therefore, we have $\mathcal{K}^{\eta}_{N,t} \subseteq \mathcal{J}^{\eta}_{N,t}$. This completes the proof of Theorem \ref{theorem: equivalence_two_problem}.
\section{Proof of Theorem \ref{lemma:C_is_constant_for_N}}\label{proof:lemma:C_is_constant_for_N}
    
Recall that for any $0<\alpha \leq 1$, we have 
\begin{align}\label{optimization_problem}
    c^{\eta}_{N,t} (\alpha) \triangleq \underset{\gdot \in \Lambda_{\mathsf{AD}}^t}{\max} ~ \underset{\PA \left( \gdot, \pare \right) \geq \alpha}{\MSE\left(\gdot, \pare\right)}.
\end{align}
To prove \eqref{C_is_constant_for_N},  for any $0<\alpha\leq 1$, $N\geq 2$, and $t < N$, we first show 
\begin{align}\label{more_adv_is_good}
    c^{\eta}_{N,t} (\alpha) \geq c^{\eta}_{N-t+1,1} (\alpha),
\end{align}
and then
\begin{align}\label{more_adv_is_not_good}
    c^{\eta}_{N,t} (\alpha) \leq c^{\eta}_{N-t+1,1} (\alpha).
\end{align}
\subsection{\texorpdfstring{Proving $c^{\eta}_{N,t} (\alpha) \geq c^{\eta}_{N-t+1,1} (\alpha)$}{x}}
We consider two scenarios. In Scenario $1$, we assume $|\mH| = N-t$, and $|\mT| = 1$. In Scenario $2,$ we assume $\mH=N-t, \mT=t.$ By developing an appropriate coupling between the adversarial noise distributions in both scenarios, we will show that $c^{\eta}_{N,t} (\alpha) \geq c^{\eta}_{N-t+1,1} (\alpha){x}$.
Consider Scenario $1$. Assume that $g^*_1(.) \in \Lambda_{\mathsf{AD}}^1$, is a noise distribution that is a solution to the following optimization problem
\begin{align}\label{def_of_g*1}
    \underset{\gdot \in \Lambda_{\mathsf{AD}}^1}{\max} ~ \underset{\mathsf{PA}_{N-t+1} \left( \gdot, \pare \right) \geq \alpha}{\mathsf{MSE}_{N-t+1}\left(\gdot, \pare\right)}.
\end{align}
 We utilize $g^*_1(.)$ to create $g_t$, a noise distribution for the scenario 2 involving $|\mH| = N-t$, and $|\mT| = t$, as 
 \begin{align}
     g_t(\{n_a\}_{a \in \mT}) = \underset{{a \in \mT}}{\Pi}\delta(z-n_a)g^*_1(z).
 \end{align}
 Specifically, the adversary employs the distribution of $g^*_1(.)$ to sample a noise, denoted by $n$. Subsequently, it applies this noise for all of its $t$ nodes. Therefore, for all $a \in \mT$, we have $n_a = n$.  
 
 Note that by utilizing $g^*_1(.)$, and $g_t(.)$, for scenarios 1 and 2 respectively,  for each realization of the random variables in scenario 1, i.e., $(\bu, \{\bn_h\}_{h \in \mH}, \{\bn_a\}_{a \in \mT})$, where $|\mH| = N-t$, and $|\mT| = t$, there exists a corresponding realization of these random variables in scenario 2, where $|\mH| = N-t$, and $|\mT| = t$, and the value of $\max(\by_i)$ and $\min(\by_i)$ are same in these scenarios, and vice versa. Therefore
\begin{align}
    \PA \left( g_t(.), \pare \right) &= \mathsf{PA}_{N-t+1} \left( g^*_1(.), \pare \right) \geq \alpha, \\
    \MSE\left(g_t(.), \pare\right)&=\mathsf{MSE}_{N-t+1}\left(g^*_1(.), \pare\right). \label{mse_is_equal_2-to_N}
\end{align}

Furthermore, Note that
\begin{align}
    c^{\eta}_{N,t} (\alpha) &=\underset{\gdot \in \Lambda_{\mathsf{AD}}^t}{\max} ~ \underset{\PA \left( \gdot, \pare \right) \geq \alpha}{\MSE\left(\gdot, \pare\right)} \nonumber \\
    &\geq \MSE\left(g_t(.), \pare\right)\nonumber \\
    &\overset{(a)}{=}\mathsf{MSE}_{N-t+1}\left(g^*_1(.), \pare\right) \nonumber \\
    & \overset{(b)}{=}\underset{\gdot \in \Lambda_{\mathsf{AD}}^1}{\max} ~ \underset{\mathsf{PA}_{N-t+1} \left( \gdot, \pare \right) \geq \alpha}{\mathsf{MSE}_{N-t+1}\left(\gdot, \pare\right)}\nonumber \\
    &\overset{(c)}{=}c^{\eta}_{N-t+1,1} (\alpha)
\end{align}
where (a) follows from \eqref{mse_is_equal_2-to_N}, (b) follows from the definition of $g^*_1(.)$, and (c) follows from \eqref{C_definition}. This completes the proof of \eqref{more_adv_is_good}.
\subsection{\texorpdfstring{Proving $c^{\eta}_{N,t} (\alpha) \leq c^{\eta}_{N-t+1,1} (\alpha)$}{x}}
To prove \eqref{more_adv_is_not_good}, first we show the following lemma.
\begin{lemma}\label{lemma:there_is_non_cancelling_noise}
    For any $N \geq 2$, $0 < \alpha \leq 1$, and $t<N$, where $|\mT| = t$, there exist a solution $g(\{n_a\}_{a \in \mT })$ to the optimization problem \eqref{C_definition}, such that
    \begin{align}
        \Pr \left( \exists i,j \in \mT, ~|\bn_i - \bn_j| > \eta \Delta\right) = 0.
    \end{align}
\end{lemma}
\begin{figure}[t]
    \centering
    \includegraphics[width=0.95\linewidth]{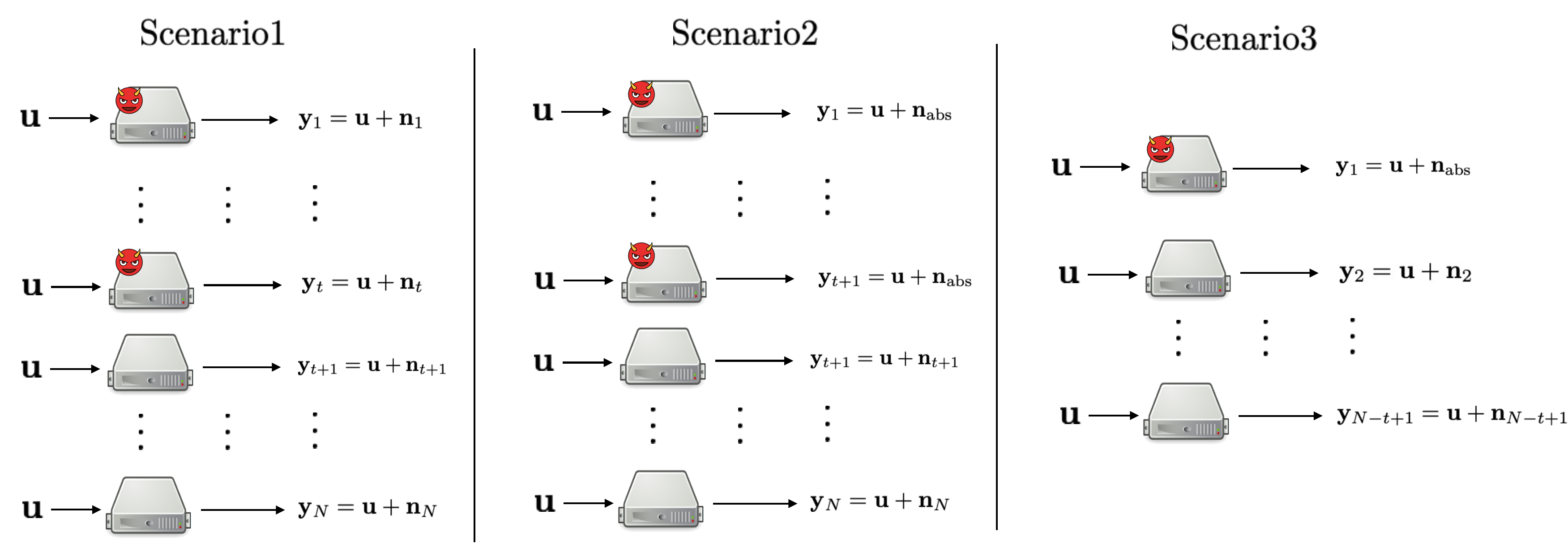}
\caption{ In the first scenario, there are $t$ adversarial nodes and $N-t$ honest nodes. The noise distribution of the adversarial nodes is $g^*_t(.)$, which is a solution to \eqref{C_definition} and satisfies Lemma \ref{lemma:there_is_non_cancelling_noise}. Let $\bn_{\text{abs}} \triangleq \max |\bn_a|$, for $a \in  \{ 1,2,\dots,t\}$. In the second scenario there are again $t$ adversarial nodes and $N-t$ honest nodes. The noise of the adversarial nodes is $\bn_{\text{abs}}$. In the third scenario, there are one adversarial node and $N-t$ honest nodes. The noise of the adversarial node is $\bn_{\text{abs}}$.}
    \label{fig:Different_Scenario}
\end{figure}

The proof of this lemma can be found in Appendix \ref{proof:lemma:there_is_non_cancelling_noise}. Consider two scenarios, where in scenario $1$ (depicted in Fig. \ref{fig:Different_Scenario}), we assume that the adversary is utilizing $g^*_t(.)$, which is a solution to \eqref{C_definition} and satisfies Lemma \ref{lemma:there_is_non_cancelling_noise}. Meaning that, 
\begin{align}\label{non_cancelling
_noise_Assumption}
        \Pr \left( \exists i,j \in \mT, ~|\bn^{(1)}_i - \bn^{(1)}_j| > \eta \Delta; g^*_t(.)\right) = 0.
    \end{align}
    Let $\bn_{\text{abs}} \triangleq \underset{a \in \mT}{\max |\bn^{(1)}_a|}$. In the second scenario (depicted is Fig. \ref{fig:Different_Scenario}), we assume that $\bn^{(2)}_a = \bn_{\text{abs}}$, for all $a \in \mT$.  Essentially, in the second scenario, we presume that all adversarial nodes possess the same noise, which is the noise with the maximum absolute value in scenario $1$.
    To implement this, the adversary can locally simulate $g^*_t(\cdot)$ for its $t$ nodes, calculate the noise with the highest absolute value, and replicate this noise for all $t$ nodes in the second scenario. Let $g_t(.)$ be the corresponding noise distribution of the construction in the second scenario. Consider the following claims.
    \begin{claim}\label{claim:Prob_is_same}
        Any realization of the noise from both adversarial and honest nodes, which results in an acceptance event in scenario $1$, would also lead to an acceptance of the inputs in scenario $2$, and vice versa.
    \end{claim}
    \begin{claim}\label{claim:mse_is_same}
        For any realization of the noise from both adversarial and honest nodes, which results in an acceptance event in scenarios $1$ and $2$, we have
    \begin{align}
        |\bn^{(1)}_{\max} + \bn^{(1)}_{\min}| \leq |\bn^{(2)}_{\max} + \bn^{(2)}_{\min}|.
    \end{align}
    \end{claim}
    The proof of Claim \ref{claim:Prob_is_same} and \ref{claim:mse_is_same} can be found in Appendix \ref{proof:claim_prb_and_mse}. Note that a direct consequence of Claim \ref{claim:Prob_is_same} is 
\begin{align}\label{prob_is_same_more_adv}
        \PA \left( g_t(.), \pare \right) = \PA \left( g^*_t(.), \pare \right) \overset{(a)}{\geq} \alpha,
    \end{align}
    where (a) follows from the fact that $g^*_t(.)$ is a solution to \eqref{C_definition}.
    
    Also, based on Claim \ref{claim:mse_is_same}, we have 
    \begin{align}\label{g_t_is_better_g*}
        \MSE\left(g_t(.), \pare\right) &= \mathbb{E}\left[\big(\mathbf{u} - \frac{\max(\underline{\by}^{(2)}) + \min(\underline{\by}^{(2)})}{2}\big)^2 ; g_t(.)| \acce\right] \nonumber \\
        &= \frac{1}{4} \mathbb{E}\left[|\bn^{(2)}_{\max} + \bn^{(2)}_{\min}|^2 ; g_t(.)| \acce\right] \nonumber \\
        &\overset{(a)}{\geq} \frac{1}{4} \mathbb{E}\left[|\bn^{(1)}_{\max} + \bn^{(1)}_{\min}|^2 ; g^*_t(.)| \acce\right]
        \nonumber \\
        & = \mathbb{E}\left[\big(\mathbf{u} - \frac{\max(\underline{\by}^{(1)}) + \min(\underline{\by}^{(1)})}{2}\big)^2 ; g^*_t(.)| \acce\right] \nonumber \\
        & = 
        \MSE\left(g^*_t(.), \pare\right),
    \end{align}
    where (a) follows from Claim \ref{claim:Prob_is_same} and \ref{claim:mse_is_same}. On the other hand, since $g^*_t(.)$ is a solution to \eqref{C_definition}, and based on \eqref{prob_is_same_more_adv}, 
    we have $\PA \left( g_t(.), \pare \right) \geq \alpha$, it implies that
    \begin{align}\label{g*_is_better_g_t}
        \MSE\left(g^*_t(.), \pare\right) \geq \MSE\left(g_t(.), \pare\right).
    \end{align}
    Therefore, based on \eqref{g_t_is_better_g*} and \eqref{g*_is_better_g_t}, we have
    \begin{align}\label{gt_and_g*_are_same}
        \MSE\left(g_t(.), \pare\right) = \MSE\left(g^*_t(.), \pare\right).
    \end{align}

    Consider the third scenario (depicted is Fig. \ref{fig:Different_Scenario}) in which $|\mH| = N-t$, and $|\mT| = 1$. Assume that the behavior of the adversary is characterized as follows: The adversary instantiates $t$ instances of adversarial nodes locally and utilizes $g_t(.)$ as the noise distribution, outputting $\mathbf{n}_{\text{abs}}$ as its noise  in the third scenario. More precisely, in the third scenario, instead of having $t$ instances of the same noise value, the adversary employs just one of them. Let $g_1(.)$ be the corresponding noise distribution in this construction. 
    
    Note that for each realization we have $\max(\underline{\by}^{(2)}) = \max(\underline{\by}^{(3)})$, and $\min(\underline{\by}^{(2)}) = \min(\underline{\by}^{(3)})$. Therefore, we have \begin{align}\label{nabs_for_two_nodes_pacc}
        \mathsf{PA}_{N-t+1} \left( g_1(.), \pare \right) = \PA \left( g_t(.), \pare \right) \geq \alpha,
    \end{align}
    and 
    \begin{align}\label{nabs_for_two_nodes}
        \mathsf{MSE}_{N-t+1}\left(g_1(.), \pare\right) = \MSE\left(g_t(.), \pare\right).
    \end{align}
    Thus
    \begin{align}
        c^{\eta}_{N,t} (\alpha) &= \underset{\gdot \in \Lambda_{\mathsf{AD}}^t}{\max} ~ \underset{\PA \left( \gdot, \pare \right) \geq \alpha}{\MSE\left(\gdot, \pare\right)} \nonumber \\
        &\overset{(a)}{=} \MSE\left(g^*_t(.), \pare\right) \nonumber \\
        &\overset{(b)}{=} \MSE\left(g_t(.), \pare\right) \nonumber \\
        &\overset{(c)}{=}\mathsf{MSE}_{N-t+1}\left(g_1(.), \pare\right) \nonumber \\
        &\overset{(d)}{\leq} \underset{\gdot \in \Lambda_{\mathsf{AD}}^1}{\max} ~ \underset{\mathsf{PA}_{N-t+1} \left( \gdot, \pare \right) \geq \alpha}{\mathsf{MSE}_{N-t+1}\left(\gdot, \pare\right)} \nonumber \\
        &=c^{\eta}_{N-t+1,1} (\alpha),
    \end{align}
    where (a) follows from the definition of $g^*_t(.)$,  (b) follows from \eqref{gt_and_g*_are_same}, (c) follows from \eqref{nabs_for_two_nodes}, and (d) follows from the fact that based on  \eqref{nabs_for_two_nodes_pacc}, we have $\mathsf{PA}_{N-t+1} \left( g_1(.), \pare \right) \geq \alpha$.
    This completes the proof of \eqref{more_adv_is_not_good}. Thus, by proving \eqref{more_adv_is_good} and \eqref{more_adv_is_not_good}, we proved Lemma \ref{lemma:C_is_constant_for_N}.

\section{Proof of  Theorem \ref{theorem:one_honest_is_enough}} \label{Proofs of main theorems}

Note that based on \eqref{C_is_constant_for_N} and Algorithm \ref{Alg:finding_eta}, 
 for any $N \geq 2$ and $t < N$, we have $\eta^*_{N,t} = \eta^*_{N-t+1,1}$.
To show \eqref{adversary_util_is_same},  recall that based on the definition of $\mathcal{J}^{\eta}_{N,t}$ in \eqref{J_definition}, we have 
\begin{align*}
    \mathcal{J}^{\eta}_{N,t} \overset{\Delta}{=} \bigl\{\big( 
    \MSE\big(g^*(.),  \eta \big), 
    \PA \left( g^*(.), \pare \right)
    \big) \bigl| ~g^*(.) \in \mathcal{B}^{\pare}_{N,t}\bigr\}.
\end{align*}
A direct consequence of \eqref{equivalence of k and j} and \eqref{C_is_constant_for_N} is that 
    that for any $\pare \in \Lambda_{DC}$, $N \geq 2$, we have 
    \begin{align}\label{remark:Jn_equal_J2}
        \mathcal{J}^{\eta}_{N,t} = \mathcal{J}^{\eta}_{N-t+1,1}.
    \end{align}
 In addition, consider $g_t^*(.) \in \mathcal{B}^{\pare_{N,t}^*}_{N,t}$ and $g^*_1(.) \in \mathcal{B}^{\pare_{N-t+1,1}^*}_{N-t+1,1}$.
Let $(b_t,a_t) \triangleq \big( 
    \MSE\big(g^*_t(.),  \eta^*_{N,t} \big), 
    \PA \left( g^*_t(.), \eta^*_{N,t} \right)
    \big)$, and $(b_1,a_1) \triangleq \big( 
    \mathsf{MSE}_{N-t+1}\big(g^*_1(.),  \eta^*_{N-t+1,1} \big), 
    \mathsf{PA}_{N-t+1} \left( g^*_1(.), \eta^*_{N-t+1,1} \right)
    \big)$.
    Note that based on \eqref{J_definition}, $(b_t,a_t) \in \mathcal{J}^{\eta^*_{N,t}}_{N,t}$ and $(b_1,a_1) \in \mathcal{J}^{\eta^*_{N-t+1,1}}_{N-t+1,1}$. Since $\eta^*_{N,t} = \eta^*_{N-t+1,1}$,  \eqref{remark:Jn_equal_J2} implies that $(b_1,a_1) \in \mathcal{J}^{\eta^*_{N,t}}_{N,t}$. Moreover, based on definition \eqref{J_definition}, we know that for all elements of $\mathcal{J}^{\eta^*_{N,t}}_{N,t}$, the value of $Q_{\mathsf{AD}}(.,.)$ is same. Thus,
    \begin{align}\label{bn_an_b2_a2}
        Q_{\mathsf{AD}}(b_t,a_t) =  Q_{\mathsf{AD}}(b_1,a_1)
    \end{align}
    This implies that
    \begin{align}
        \mathsf{U}_{\mathsf{AD}}^N\left(g_t^*(.), \pare_{N,t}^* \right) \overset{(a)}{=} Q_{\mathsf{AD}}(b_t,a_t) 
         \overset{(b)}{=} Q_{\mathsf{AD}}(b_1,a_1) \overset{(c)}{=} \mathsf{U}_{\mathsf{AD}}^{N-t+1}\left(g_1^*(.), \pare_{N-t+1,1}^* \right)
    \end{align}
    where (a) and (c) follows from \eqref{adv-utility}, and (b) follows from  \eqref{bn_an_b2_a2}. This completes the proof of  \eqref{adversary_util_is_same}.
    
    To complete the proof of Theorem \ref{theorem:one_honest_is_enough}, we show \eqref{util_dc_same}. Note that
    \begin{align}
        \underset{\gdot \in \mathcal{B}^{\pare^*_{N,t}}_{N,t}}{\min} ~ {\mathsf{U}}^N_\mathsf{DC}\left(\gdot, \pare \right) &\overset{(a)}{=} \underset{\gdot \in \mathcal{B}^{\eta^*_{N,t}}_{N,t}}{\min} ~ Q_{\mathsf{DC}} \left( \MSE\left(\gdot, \pare\right), \PA \left( \gdot, \pare \right)\right) \nonumber \\
        & \overset{(b)}{=} \underset{(\beta, \alpha) \in \mathcal{J}^{\eta^*_{N,t}}_{N,t}}{\min} ~ Q_{\mathsf{DC}} \left(\beta, \alpha\right) \nonumber \\
        & \overset{(c)}{=} \underset{(\beta, \alpha) \in \mathcal{J}^{\eta^*_{N-t+1,1}}_{N-t+1,1}}{\min} ~ Q_{\mathsf{DC}} \left(\beta, \alpha\right) \nonumber \\
        &\overset{(d)}{=} \underset{\gdot \in \mathcal{B}^{\eta^*_{N-t+1,1}}_{N-t+1,1}}{\min} ~ Q_{\mathsf{DC}} \left( \mathsf{MSE}_{N-t+1}\left(\gdot, \pare\right), \mathsf{PA}_{N-t+1} \left( \gdot, \pare \right)\right) \nonumber \\
        &\overset{(e)}{=} \underset{\gdot \in \mathcal{B}^{\pare^*_{N-t+1,1}}_{N-t+1,1}}{\min} ~ {\mathsf{U}}_\mathsf{DC}^{N-t+1}\left(\gdot, \pare \right)
    \end{align}
    where (a) and (e) follows from \eqref{dc-utility}, (b) and (d) follows from the definition of \eqref{J_definition}, (c) follows from  \eqref{remark:Jn_equal_J2} and the fact that $\eta^*_{N,t} = \eta^*_{N-t+1,1}$. This completes the proof of \eqref{util_dc_same}.


\section{Proof of Theorem \ref{theorem:CofJ_N_ is same}}\label{proof:theorem:CofJ_N_ is same}

Recall that based on Theorem \ref{lemma:C_is_constant_for_N}, for any $0 < \alpha \leq 1$, $N\geq 2$, and $t < N$ we have 
\begin{align}
    c^{\eta}_{N,t} (\alpha) = c^{\eta}_{N-t+1,1} (\alpha).
\end{align}
Therefore, to prove Theorem \ref{theorem:CofJ_N_ is same} we show the following lemma.
\begin{lemma}\label{lemma:characterisng c eta}
    As long as the noise distribution of the adversary is limited to the strong noises, for any $0 < \alpha \leq 1$, $N\geq 1$,  we have 
\begin{align}
    c^{\eta}_{N+1,1} (\alpha) = \upperbound.
\end{align}
\end{lemma}

To prove Lemma \ref{lemma:characterisng c eta}, we establish in Claim \ref{lemma:upper_bound_of_mean} that $c^{\eta}_{N+1,1}(\alpha) \leq  \upperbound$ (converse). To do so, we first show that restricting the search space of $g(\cdot)$ to symmetric distributions does not affect optimality. Next, we follow these steps: (i) We demonstrate that by limiting the domain of the noise distributions to $[-(\eta+1)\Delta, -(\eta-1)\Delta] \cup [(\eta-1)\Delta, (\eta+1)\Delta]$ and assuming the probability of acceptance is exactly $\alpha$, optimality is still preserved. (ii) We then prove that for such noise distributions, the mean squared error is given by $\mathsf{MSE}_{N+1}\left( g(\cdot), \pare \right) = \frac{1}{2\alpha}\int_{(\eta-1)\Delta}^{(\eta+1)\Delta} \nu_{\eta, N}(z)g(z) \, dz$, where $\nu_{\eta, N}(z)$ is defined in Theorem \ref{theorem:CofJ_N_ is same}. (iii) By adjusting the parameters of the integral and applying Jensen's inequality, we derive the upper bound $\upperbound$.

In Claim \ref{lemma:lower_bound_of_mean}, we establish that $c^{\eta}_{N+1,1}(\alpha) \geq  \upperbound$ (achievability). Using Algorithm \ref{Alg:finding_noise}, we find a noise distribution and show that for this distribution, $\mathsf{MSE}_{N+1}(g(\cdot), \pare) = \upperbound$ and $\mathsf{PA}_{N+1}(g(\cdot), \eta) = \alpha$. In the following we prove Claim
         \ref{lemma:upper_bound_of_mean} and \ref{lemma:lower_bound_of_mean} to complete the proof of Lemma \ref{lemma:characterisng c eta}, and Theorem \ref{theorem:CofJ_N_ is same}.

    \subsection{Proof of Claim \ref{lemma:upper_bound_of_mean}}\begin{claim}\label{lemma:upper_bound_of_mean}
        For any $\pare \in \Lambda_{\mathsf{DC}} $, $0 < \alpha \leq 1$, and $N \geq 1$ we have
    \begin{align}\label{H_eta_is_upperbound}
        c^{\eta}_{N+1,1}(\alpha)\leq \upperbound.
    \end{align}
    \end{claim}
    \begin{proof}
        To prove Claim \ref{lemma:upper_bound_of_mean}, we first show the following lemma which establishes that to characterize $c^{\eta}_{N+1,1}(\alpha)$, we can limit the search space for $\gdot$ to the set of symmetric probability density functions. 
\begin{lemma}\label{lemma:making_symmetry} 
    For any $\eta \in \Lambda_{\mathsf{DC}}$, $\gdot \in \Lambda_{\mathsf{AD}}$, and $N \geq 1$,
    we have $\mathsf{MSE}_{N+1}\left( g_{\textrm{sym}}(.), \pare\right) = \mathsf{MSE}_{N+1}\left( \gdot, \pare\right)$ and  $\mathsf{PA}_{N+1} (g_{\textrm{sym}}(.), \pare) = \mathsf{PA}_{N+1}(g(.), \pare)$,  where
      $g_{\textrm{sym}}(z) \triangleq \frac{g(z) + g(-z)}{2}$. 
\end{lemma}
\begin{proof}
    The formal proof is in Appendix \ref{proof:lemma:making_symmetry}. Informally the proof is as follows. Let $g_{\textrm{ref}}(z) \triangleq g(-z)$. Note that the acceptance rule is of the form $\max(\underline{\by}) - \min(\underline{\by}) \leq \eta \Delta$, or equivalently $\max(\underline{\bn}) - \min(\underline{\bn}) \leq \eta \Delta$, where $\underline{\bn} = (\bn_1, \dots, \bn_{N+1})$. The statistics of $\max(\underline{\bn}) - \min(\underline{\bn})$ are the same for both $\gdot, g_{\textrm{ref}}(.)$ and $g_{\textrm{sym}}(.)$. Thus, the probability of acceptance is the same for all these noise distributions. In addition, the MSE  given acceptance, is $|\bu - \frac{\max(\underline{\by}) + \min(\underline{\by})}{2}|^2 = \frac{(\max(\underline{\bn}) + \min(\underline{\bn}))^2}{4}$, whose statistics also remain the same for all these noise distributions.
\end{proof}
\begin{corollary}
    Let  $\Bar{\Lambda}_{\mathsf{AD}} = \big\{ g(z) ~|  ~g(z) \in \Lambda_{\mathsf{AD}} ~\& ~g(z) = g(-z) \big\}$, which is simply the set of all symmetric noise distributions. A direct conclusion of Lemma \ref{lemma:making_symmetry} and this definition is that, for any $0<\alpha \leq 1$, $\eta \in \Lambda_{\DC}$, and $N \geq 1$, we have
\begin{align}\label{relation_between_best_seymmetric_best_general}
    \underset{\gdot \in \Bar{\Lambda}_{\mathsf{AD}}}{\max} ~ \underset{\mathsf{PA}_{N+1} \left( \gdot, \pare \right) \geq \alpha}{\mathsf{MSE}_{N+1}\big( \gdot, \pare \big)} = \underset{\gdot \in \Lambda_{\mathsf{AD}}}{\max} ~ \underset{\mathsf{PA}_{N+1} \left( \gdot, \pare \right) \geq \alpha}{\mathsf{MSE}_{N+1}\big( \gdot, \pare \big)}.
\end{align}
\end{corollary}  
Note that based on Lemma~\ref{lemma:making_symmetry}, we limit the search space of the noise distributions to symmetric ones.  Therefore, hence after we assume that the noise distribution of the adversary is symmetric. 
Also, recall that $\gdot$ is a strong noise distribution, and thus $\Pr (|\bn_a| < \Delta; \gdot) = 0$. 

To prove Claim \ref{lemma:upper_bound_of_mean}, we next show Lemmas~\ref{general_format_symmetric}, \ref{lemma:bounded_noise_existence}, and \ref{lemma:exact_acc_noise_existence}, stated as follows: 
    \begin{lemma}\label{general_format_symmetric}
    For any strong symmetric $\gdot \in \Lambda_{\mathsf{AD}}$, and 
    $N \geq 1$,
    we have
    \begin{align}
    \mathsf{PA}_{N+1} \left( \gdot, \pare \right) &= 2\int_{\Delta}^{(\eta-1)\Delta} g(z) \,dz +  2\int_{(\eta-1)\Delta}^{(\eta+1)\Delta} \left( 
  \int_{z-\eta\Delta}^{\Delta} f_{\bn_{\min}}(x) \,dx\right) g(z) \,dz,\label{general_symmetric_acc}\\
    \mathsf{MSE}_{N+1}\big( \gdot, \pare \big) &=
        \frac{1}{2\mathsf{PA}_{N+1} \left( \gdot, \pare \right)}\int_{\Delta}^{(\eta-1)\Delta} \left(\int_{-\Delta}^{\Delta} (x+z)^2 f_{\bn_{\min}}(x)\,dx \right)g(z) \,dz
        \nonumber \\
        &+\frac{1}{2\mathsf{PA}_{N+1} \left( \gdot, \pare \right)}\int_{(\eta-1)\Delta}^{(\eta+1)\Delta} \left(\int_{z-\eta\Delta}^{\Delta} (x+z)^2 f_{\bn_{\min}}(x)\,dx \right)g(z) \,dz.\label{general_symmetric_mse}
\end{align}
where 
\begin{align}
    f_{\bn_{\min}}(x) = N f_{\bn_{h}}(x) (1-F_{\bn_h}(x))^{(N-1)},
\end{align}
is the PDF of $\bn_{\min} = \min \bn_h$, for $h \in \mH$, and $f_{\bn_h}(.)$ is the PDF of  the noise of the honest nodes, and $F_{\bn_h}(.)$ is the CDF of  the noise of the honest nodes.
    \end{lemma}
    \begin{proof}
        Proof is in Appendix \ref{proof:general_format_symmetric}
    \end{proof}
    

\begin{lemma}\label{lemma:bounded_noise_existence}
    For any strong symmetric noise distribution $g_1(.)$, where $\PA(g_1(.), \eta) > 0$, there exist a strong symmetric noise distribution $g_2(.)$, such that for the case of 
     $|z| < (\eta-1)\Delta$,   
    $|z| > (\eta+1)\Delta$, $g_2(z)=0$ and
    \begin{align}
        \mathsf{PA}_{N+1}(g_2(.), \eta) &\geq \mathsf{PA}_{N+1}(g_1(.), \eta) \label{betternoise_acc} \\
        \mathsf{MSE}_{N+1}(g_2(.), \eta) &\geq \mathsf{MSE}_{N+1}( g_1(.), \eta) \label{betternoise_mse}.
    \end{align}
\end{lemma}
    
\begin{proof}
    Proof is in Appendix \ref{proof:lemma:bounded_noise_existence}
\end{proof}

For any symmetric noise distribution $g(.)$, such that for the case of $|z| < (\eta-1)\Delta$ and $|z| > (\eta-1)\Delta$, we have $g(z) = 0$, we call $g(.)$ a satisfying noise of Lemma \ref{lemma:bounded_noise_existence}.

\begin{lemma}\label{lemma:exact_acc_noise_existence}
    Let $g_1(.)$ be a strong satisfying noise of Lemma \ref{lemma:bounded_noise_existence}. For each $0< \alpha < \mathsf{PA}_{N+1}(g_1(.), \eta)$, there exists a satisfying noise of Lemma \ref{lemma:bounded_noise_existence} , $g_2(.)$, such that  $\mathsf{PA}_{N+1} \left( g_2(.), \pare \right) = \alpha$, and $\mathsf{MSE}_{N+1}( g_2(.), \eta) = \mathsf{MSE}_{N+1}( g_1(.), \eta)$.
\end{lemma}

\begin{proof}
    Proof is in Appendix \ref{proof:lemma:exact_acc_noise_existence}.
\end{proof}

    \end{proof}
   Using the above lemmas, we prove Claim \ref{lemma:upper_bound_of_mean}, as follows. 
    \begin{proof}
       Without loss of generality, henceforth, we assume that all of the adversarial noise meets the condition of Lemmas \ref{lemma:bounded_noise_existence} and \ref{lemma:exact_acc_noise_existence}. More precisely, $\mathsf{PA}_{N+1} \left( g(.), \pare \right) = \alpha$, and also  $g(z) = 0$ for $|z| < (\eta-1)\Delta$ and
    $|z| > (\eta+1)\Delta$. Then from Lemma  \ref{general_format_symmetric}, we have
\begin{align}
    \mathsf{PA}_{N+1} \left( \gdot, \pare \right) &= 2\int_{(\eta-1)\Delta}^{(\eta+1)\Delta} \left( 
  \int_{z-\eta\Delta}^{\Delta} f_{\bn_{\min}}(x) \,dx\right) g(z) \,dz, \label{acc_neww_general}\\
    \mathsf{MSE}_{N+1}\big(\gdot, \pare \big) & =\frac{1}{2\alpha}\int_{(\eta-1)\Delta}^{(\eta+1)\Delta} \left(\int_{z-\eta\Delta}^{\Delta} (x+z)^2f_{\bn_{\min}}(x)\,dx \right)g(z) \,dz. \label{mse_neww_general}
\end{align}
Let $k_{\eta, N}(z) \triangleq \int_{z-\eta\Delta}^{\Delta} f_{\bn_{\min}}(x)\,dx$, for $(\eta-1)\Delta \leq z \leq  (\eta+1)\Delta$.  Based on \eqref{acc_neww_general} we have
\begin{align}\label{prob_acc_to_noise_general}
    \mathsf{PA}_{N+1} \left( \gdot, \pare \right) = 2\int_{(\eta-1)\Delta}^{(\eta+1)\Delta} \left( 
  \int_{z-\eta\Delta}^{\Delta} f_{\bn_{\min}}(x) \,dx\right) g(z) \,dz  
    =2\int_{(\eta-1)\Delta}^{(\eta+1)\Delta}  k_{\eta, N}(z) g(z) \,dz.
\end{align}
    Since we assume that $g(z)$ is symmetric, we have
\begin{align}\label{pdf_positive_general}
    \int_{(\eta-1)\Delta}^{(\eta+1)\Delta} g(z) \,dz = \frac{1}{2}.
\end{align}
 By changing the parameters of  integrals for $q = k_{\eta, N}(z)$, we can rewrite \eqref{pdf_positive_general} as
\begin{align}\label{rewrite:pdf_positive_general}
\int_{0}^{1} w_{\eta, N}(q) \,dq = \frac{1}{2},
\end{align}
where in the above equation, $w_{\eta, N}(q) \triangleq \frac{-g(k_{\eta, N}^{-1}(q))}{k_{\eta, N}^{\prime}( k_{\eta, N}^{-1} (q))}$, for $0 \leq q \leq 1$. Note that since $k_{\eta, N}(z)$ is a strictly decreasing function, then $w_{\eta, N}(q)$ is a non-negative function.  Similarly, we can rewrite \eqref{prob_acc_to_noise_general} as
\begin{align}\label{rewrite:prob_acc_to_noise_general}
  \int_{0}^{1}  q w_{\eta, N}(q) \,dq = \frac{\alpha}{2}
\end{align}
Let $\nu_{\eta, N}(z) \triangleq \int_{z-\eta\Delta}^{\Delta} (x+z)^2f_{\bn_{\min}}(x)\,dx$. Based on \eqref{mse_neww_general}, we have
\begin{align}\label{mse_general_for_bounded_noise}
    \mathsf{MSE}_{N+1}\big(\gdot, \pare \big) = &  \frac{1}{2\alpha}\int_{(\eta-1)\Delta}^{(\eta+1)\Delta} \left(\int_{z-\eta\Delta}^{\Delta} (x+z)^2f_{\bn_{\min}}(x)\,dx \right)g(z) \,dz \nonumber \\
     & =  \frac{1}{2\alpha}\int_{(\eta-1)\Delta}^{(\eta+1)\Delta} \nu_{\eta, N}(z)g(z) \,dz \nonumber \\
     &\overset{(a)}{=} \frac{1}{2\alpha}\int_{0}^{1}  h_{\eta, N}(q) w_{\eta, N}(q) \,dq,
\end{align}
where in (a) $q = k_{\eta, N}(z)$ and  $ h_{\eta, N}(q) \triangleq \nu_{\eta, N}(k_{\eta, N}^{-1} (q))$. Let $h^*_{\eta, N}(q)$ be the upper concave envelop of $h_{\eta, N}(q)$, for $0 \leq q \leq 1$. Then
\begin{align}
    \frac{1}{2\alpha}\int_{0}^{1}  h_{\eta, N}(q) w_{\eta, N}(q) \,dq &\leq \frac{1}{2\alpha}\int_{0}^{1}  h^*_{\eta, N}(q) w_{\eta, N}(q) \,dq \nonumber \\
    & =  \frac{1}{4\alpha}\int_{0}^{1}  h^*_{\eta, N}(q) 2w_{\eta, N}(q) \,dq \nonumber \\
    &\overset{(a)}{\leq} \frac{1}{4\alpha} h^*_{\eta, N} \left( \int_{0}^{1}  2q w_{\eta, N}(q) \,dq \right) \nonumber \\
    & \overset{(b)}{=} \frac{h^*_{\eta, N}(\alpha)}{4\alpha},
\end{align}
where (a) follows from  Jensen's inequality and the fact that based on \eqref{rewrite:pdf_positive_general}, $\int_{0}^{1} 2w_{\eta, N}(q) \,dq = 1$, and (b) follows from \eqref{rewrite:prob_acc_to_noise_general}. This completes the proof of Claim \ref{lemma:upper_bound_of_mean}.
\end{proof}

    \subsection{Proof of Claim \ref{lemma:lower_bound_of_mean}}\begin{claim}\label{lemma:lower_bound_of_mean}
        For any $\pare \in \Lambda_{\mathsf{DC}} $, $0 < \alpha \leq 1$, and $N \geq 1$ we have
    \begin{align}\label{H_eta_is_lowerbound}
        c^{\eta}_{N+1,1}(\alpha)\geq \upperbound.
    \end{align}
    \end{claim}
    \begin{proof}
To prove this claim, we introduce a noise distribution satisfying 
$\mathsf{PA}_{N+1} \left( \gdot, \pare \right) \geq \alpha$ and  $\mathsf{MSE}_{N+1}\big( \gdot, \pare \big) = \frac{h^*_{\eta, N}(\alpha)}{4\alpha}$. Note that for $(\eta-1)\Delta \leq z \leq  (\eta+1)\Delta$, we have
\begin{align}
    k_{\eta, N}(z) = \int_{z-\eta\Delta}^{\Delta} f_{\bn_{\min}}(x)\,dx =F_{\bn_{\min}}(\Delta) - F_{\bn_{\min}}(z-\eta\Delta) \overset{(a)}{=}1 - F_{\bn_{\min}}(z-\eta\Delta),
\end{align}
where (a) follows from the fact 
that $\Pr (|\bn_{\min}| > \Delta) = 0$, which implies that $F_{\bn_{\min}}(\Delta) = 1$.
Note that
\begin{align}
   F_{\bn_{\min}}(x) = 1 - (1 - F_{\bn_{h}}(x))^N.
\end{align}
Recall that we assume that $F_{\bn_{h}}$ is a strictly increasing function in $[-\Delta, \Delta]$, i.e., for all $-\Delta \leq a < b \leq \Delta$, we have 
$F_{\bn_{\min}}(a) < F_{\bn_{\min}}(b)$. This implies that $F_{\bn_{\min}}$ is a strictly increasing function in $[-\Delta, \Delta]$, and thus $k_{\eta, N}(z)$ is an invertible function for $(\eta-1)\Delta \leq z \leq  (\eta+1)\Delta$. Therefore, we have
\begin{align}\label{inverse_of_CDF}
    k^{-1}_{\eta, N}(q) = \eta \Delta + F^{-1}_{\bn_{\min}}(1-q),
\end{align}
for any $0 \leq q \leq 1$. Consider following cases:
\begin{enumerate}
    \item $h^*_{\eta, N}(\alpha) = h_{\eta, N}(\alpha)$: In this case, let $z_1 = k^{-1}_{\eta, N}(\alpha)$ and $g(z) = \frac{1}{2}\delta(z+z_1) + \frac{1}{2}\delta(z-z_1)$. Note that we can calculate $z_1$ based on \eqref{inverse_of_CDF}.
    Based on \eqref{prob_acc_to_noise_general}, we have
    \begin{align}
       \mathsf{PA}_{N+1} \left( g(.), \pare \right) &=  2\int_{(\eta-1)\Delta}^{(\eta+1)\Delta}  k_{\eta, N}(z) g(z) \,dz \nonumber \\
       & = k_{\eta, N}(z_1) \nonumber \\
       & = k_{\eta, N}(k^{-1}_{\eta, N}(\alpha)) \nonumber \\
       & = \alpha.
    \end{align}
    Additionally, based on \eqref{mse_general_for_bounded_noise}
    \begin{align}
        \mathsf{MSE}_{N+1}\big( \gdot, \pare \big) & =  \frac{1}{2\alpha}\int_{(\eta-1)\Delta}^{(\eta+1)\Delta} \nu_{\eta, N}(z)g(z) \,dz  \nonumber \\
        & = \frac{1}{4\alpha} \nu_{\eta, N}(z_1) \nonumber \\
        & \overset{(a)}{=} \frac{1}{4\alpha} \nu_{\eta, N}(k^{-1}_{\eta, N}(\alpha)) \nonumber \\
        & \overset{(b)}{=} \frac{h_{\eta, N}(\alpha)}{4\alpha} \nonumber \\
        & \overset{(c)}{=} \frac{h^*_{\eta, N}(\alpha)}{4\alpha},
    \end{align}
    where (a) follows from the definition of $z_1$ and (b) follows from the definition of $h_{\eta, N}(.)$, and (c) follows from the fact that in this case we assume that we have $h^*_{\eta, N}(\alpha) = h_{\eta, N}(\alpha)$.
    \item $h^*_{\eta, N}(\alpha) \neq h_{\eta, N}(\alpha)$: In this case, since $h^*_{\eta, N}(.)$ is the upper-concave envelop of $h_{\eta, N}(.)$, there exist $q_1$ and $q_2$, such that $0 \leq q_1 < \alpha < q_2 \leq 1$ and 
    \begin{align}\label{concave_envelop_boundries}
        h^*_{\eta, N}(q_1) = h_{\eta, N}(q_1), \nonumber \\
        h^*_{\eta, N}(q_2) = h_{\eta, N}(q_2).
    \end{align}
    In addition,  for all $q_1 \leq q \leq q_2$, we have
    \begin{align}\label{linear_concave}
        h^*_{\eta, N}(q) = \frac{h_{\eta, N}(q_2) - h_{\eta, N}(q_1)}{q_2 - q_1} (q - q_1) + h_{\eta, N}(q_1).
    \end{align}
    Let $z_1 = k^{-1}_{\eta, N}(q_1)$, $z_2 = k^{-1}_{\eta, N}(q_2)$, $\beta_1 = \frac{q_2 -\alpha }{2(q_2 - q_1)}$, $\beta_2 = \frac{\alpha - q_1}{2(q_2 - q_1)}$, and 
    \begin{align}
        g(z) = \beta_1 \delta(z+z_1) +\beta_2 \delta(z+z_2) +\beta_1 \delta(z-z_1) +\beta_2 \delta(z-z_2).
    \end{align}
    Note that we can calculate $z_1$ and $z_2$ based on \eqref{inverse_of_CDF}. One can verify that
    \begin{align}
        2\beta_1 + 2\beta_2 &= 1,  \\
        2\beta_1 q_1 + 2\beta_2 q_2 &= \alpha \label{q_plus_beta}.
    \end{align}
    Based on \eqref{prob_acc_to_noise_general}, we have
    \begin{align}
       \mathsf{PA}_{N+1} \left( g(.), \pare \right) &=  2\int_{(\eta-1)\Delta}^{(\eta+1)\Delta}  k_{\eta, N}(z) g(z) \,dz \nonumber \\
       & = 2\beta_1k_{\eta, N}(z_1) + 2\beta_2 k_{\eta, N}(z_2) \nonumber \\
       & = 2\beta_1 q_1 + 2\beta_2 q_2 \nonumber \\
       & \overset{(a)}{=} \alpha,
    \end{align}
    where (a) follows from \ref{q_plus_beta}.
    Additionally, based on \eqref{mse_general_for_bounded_noise}
    \begin{align}
        \mathsf{MSE}_{N+1}\big( \gdot, \pare \big) & =  \frac{1}{2\alpha}\int_{(\eta-1)\Delta}^{(\eta+1)\Delta} \nu_{\eta, N}(z)g(z) \,dz  \nonumber \\
        & = \frac{1}{2\alpha} \left( \beta_1\nu_{\eta, N}(z_1) + \beta_2\nu_{\eta, N}(z_2)\right)\nonumber \\
        &\overset{(a)}{=} \frac{1}{2\alpha} \left( \beta_1\nu_{\eta, N}(k^{-1}_{\eta, N}(q_1)) + \beta_2\nu_{\eta, N}(k^{-1}_{\eta, N}(q_2))\right)\nonumber\\
        & = \frac{1}{4\alpha} \left( 2\beta_1 h_{\eta, N}(q_1) + 2\beta_2 h_{\eta, N}(q_2)\right) \nonumber \\
        & \overset{(b)}{=} \frac{1}{4\alpha} \left( 2\beta_1 h^*_{\eta, N}(q_1) + 2\beta_2 h^*_{\eta, N}(q_2)\right) \nonumber \\
        &\overset{(c)}{=} \frac{1}{4\alpha} h^*_{\eta, N}\left( 2\beta_1 q_1 + 2\beta_2 q_2\right) \nonumber \\
        & \overset{(d)}{=} \frac{h^*_{\eta, N}(\alpha)}{4\alpha},
    \end{align}
    where (a) follows from the definition of $h_{\eta, N}(.)$ , and (b) follows from \eqref{concave_envelop_boundries}, (c) follows from the fact that based on \ref{linear_concave}, $h^*(.)$ is a linear function in $[q_1, q_2]$ and $2\beta_1 + 2\beta_2 = 1$, and (d) follows from \eqref{q_plus_beta}.This completes the proof of Claim \ref{lemma:lower_bound_of_mean}.
\end{enumerate}
\end{proof}

\appendices

\section{Proof of Lemma \ref{lemma:there_is_non_cancelling_noise}}\label{proof:lemma:there_is_non_cancelling_noise}
    Consider a solution $g^*_t(.)$ of the optimization problem \eqref{C_definition}.  Let $\mathcal{D}$ be the event of for all $i,j \in \mT$, we have $|\bn_i - \bn_j| \leq \eta \Delta$. Consider the new noise of $[\bn_a|\mathcal{D}]_{a \in \mT}$ and let $g_t(.)$ be the corresponding distribution of this construction. Noteworthy that $g_t(.)$ is basically the distribution of $g^*_t(.)$, condition on $\mathcal{D}$. This implies that when the adversarial noise distribution follows $g^*_t(.)$, given the occurrence of event $\mathcal{D}$, the noise distribution becomes $g_t(.)$, resulting in 
    \begin{align}\label{acceptance_cond_D}
        Pr \big( \acce; g^*_t(.)|\mathcal{D}\big) = Pr \big( \acce; g_t(.)|\mathcal{D}\big).
    \end{align}
    We show that $g_t(.)$ is also a solution to \eqref{C_definition}, and it satisfies the properties stated in the Lemma \ref{lemma:there_is_non_cancelling_noise}. Note that
    \begin{align}\label{acc_g^*_t}
        \PA \left( g^*_t(.), \pare \right) &= \Pr \big( \acce; g^*_t(.)\big) \nonumber \\
        &=\Pr \big( \acce; g^*_t(.)|\mathcal{D}\big) \Pr (\mathcal{D}; g^*_t(.)) + \Pr \big( \acce; g^*_t(.)|\mathcal{D}^c\big) \Pr (\mathcal{D}^c; g^*_t(.)) \nonumber \\
        & \overset{(a)}{=} \Pr \big( \acce; g^*_t(.)|\mathcal{D}\big) \Pr (\mathcal{D}; g^*_t(.)) \nonumber \\
        & \overset{(b)}{=} \Pr \big( \acce; g_t(.)|\mathcal{D}\big) \Pr (\mathcal{D}; g^*_t(.))
    \end{align}
    where (a) follows from the fact that based on definition of $\mathcal{D}$, $\Pr \big( \acce; g^*_t(.)|\mathcal{D}^c\big) = 0$,  and (b) follows from \eqref{acceptance_cond_D}.

    Furthermore, we have
    \begin{align}\label{acc_g_t}
        \PA \left( g_t(.), \pare \right) &= \Pr \big( \acce; g_t(.)\big) \nonumber \\
        &=\Pr \big( \acce; g_t(.)|\mathcal{D}\big) \Pr (\mathcal{D}; g_t(.)) + \Pr \big( \acce; g_t(.)|\mathcal{D}^c\big) \Pr (\mathcal{D}^c; g_t(.)) \nonumber \\
        & \overset{(a)}{=} \Pr \big( \acce; g_t(.)|\mathcal{D}\big) \Pr (\mathcal{D}; g_t(.)) \nonumber \\
        &\overset{(b)}{=} \Pr \big( \acce; g_t(.)|\mathcal{D}\big)
    \end{align}
    where (a) follows from the fact that based on definition of $\mathcal{D}$, $\Pr \big( \acce; g^*_t(.)|\mathcal{D}^c\big) = 0$, and (b) follows from the fact that based on definition of $g_t(.)$, we have $Pr (\mathcal{D}; g_t(.)) = 1$.

    Based on \eqref{acc_g^*_t} and \eqref{acc_g_t}, we have 
    \begin{align}\label{new_noise_bigger_Acc}
        \PA \left( g_t(.), \pare \right) = \Pr \big( \acce; g_t(.)|\mathcal{D}\big) \geq \Pr \big( \acce; g_t(.)|\mathcal{D}\big) \Pr (\mathcal{D}; g^*_t(.)) = \PA \left( g^*_t(.), \pare \right) \geq \alpha.
    \end{align}
    Let $\bn_{\min} \triangleq \underset{i \in [N]}{\min \bn_i}$, and $\bn_{\max} \triangleq \underset{i \in [N]}{\max \bn_i}$.
    Note that
    \begin{align}\label{new_noise_same_mse}
        &\MSE\left(g^*_t(.), \pare\right) = \mathbb{E}\left[\big(\mathbf{u} - \frac{\max(\underline{\by}) + \min(\underline{\by})}{2}\big)^2 ; g^*_t(.)| \acce\right] = \frac{1}{4} \mathbb{E}\left[\big(\bn_{\min} + \bn_{\max}\big)^2 ; g^*_t(.)| \acce\right] \nonumber \\
        & = \frac{1}{4} \mathbb{E}\left[\big(\bn_{\min} + \bn_{\max}\big)^2 ; g^*_t(.)| \acce, \mathcal{D}\right] \Pr \big(\mathcal{D}; g^*_t(.)| \acce \big) + \frac{1}{4} \mathbb{E}\left[\big(\bn_{\min} + \bn_{\max}\big)^2 ; g^*_t(.)| \acce, \mathcal{D}^c\right] \Pr \big(\mathcal{D}^c; g^*_t(.)| \acce \big) \nonumber \\
        & \overset{(a)}{=} \frac{1}{4} \mathbb{E}\left[\big(\bn_{\min} + \bn_{\max}\big)^2 ; g^*_t(.)| \acce, \mathcal{D}\right]  \nonumber
         \\
         & \overset{(b)}{=} \frac{1}{4} \mathbb{E}\left[\big(\bn_{\min} + \bn_{\max}\big)^2 ; g_t(.)| \acce, \mathcal{D}\right] \nonumber \\
         &\overset{(c)}{=}\frac{1}{4} \mathbb{E}\left[\big(\bn_{\min} + \bn_{\max}\big)^2 ; g_t(.)| \acce, \mathcal{D}\right] \Pr \big(\mathcal{D}; g_t(.)| \acce \big) + \frac{1}{4} \mathbb{E}\left[\big(\bn_{\min} + \bn_{\max}\big)^2 ; g_t(.)| \acce, \mathcal{D}^c\right] \Pr \big(\mathcal{D}^c; g_t(.)| \acce \big) \nonumber \\
         & = \frac{1}{4} \mathbb{E}\left[\big(\bn_{\min} + \bn_{\max}\big)^2 ; g_t(.)| \acce\right] \nonumber \\
         & = \mathbb{E}\left[\big(\mathbf{u} - \frac{\max(\underline{\by}) + \min(\underline{\by})}{2}\big)^2 ; g_t(.)| \acce\right] \nonumber \\
         & = \MSE\left(g_t(.), \pare\right),
    \end{align}
    where (a) follows from the fact that based on definition of $\mathcal{D}$, we have $\Pr \big(\mathcal{D}^c; g^*_t(.)| \acce \big) = 0$ and $\Pr \big(\mathcal{D}; g^*_t(.)| \acce \big) = 1$, (b) follows from the fact that when the adversarial noise distribution follows $g^*_t(.)$, given the occurrence of event $\mathcal{D}$, the noise distribution becomes $g_t(.)$, (c) follows from the fact that based on definition of $\mathcal{D}$, we have $\Pr \big(\mathcal{D}^c; g_t(.)| \acce \big) = 0$ and $\Pr \big(\mathcal{D}; g_t(.)| \acce \big) = 1$. Combining \eqref{new_noise_bigger_Acc} and \eqref{new_noise_same_mse} implies that $g_t(.)$ is also a solution to \eqref{C_definition}. This completes the proof of Lemma \ref{lemma:there_is_non_cancelling_noise}.
\section{Proof of Claims \ref{claim:Prob_is_same} and \ref{claim:mse_is_same}}\label{proof:claim_prb_and_mse}
Let $\underline{\by}^{(1)}$ be the set of observed data for scenario $1$, and $\underline{\by}^{(2)}$ be the set of observed data for scenario $2$. Note that $\underline{\by}^{(2)} \subseteq \underline{\by}^{(1)}$. Therefore,
    \begin{align*}
        \max(\underline{\by}^{(2)})  &\leq \max(\underline{\by}^{(1)})
         \\
        \min(\underline{\by}^{(2)}) &\geq \min(\underline{\by}^{(1)}).
    \end{align*}
    Thus, 
    \begin{align}
       \max(\underline{\by}^{(2)}) - \min(\underline{\by}^{(2)}) \leq \max(\underline{\by}^{(1)}) - \min(\underline{\by}^{(1)}). 
    \end{align}
    This implies that any realization of the noise from both adversarial and honest nodes, which results in an acceptance event in scenario $1$, would also lead to an acceptance of the inputs in scenario $2$. We prove that it is also true vice versa. In addition we show that for each of these realizations we have 
    \begin{align}\label{scenario2_has_bigger_mse}
        |\bn^{(1)}_{\max} + \bn^{(1)}_{\min}| \leq |\bn^{(2)}_{\max} + \bn^{(2)}_{\min}|.
    \end{align}
    Assume that the inputs has been accepted in scenario $2$. 
    Consider the following cases
    \begin{enumerate}
        \item $\bn_{\text{abs}} \geq 0$: We consider two cases:
        \begin{itemize}
            \item $\bn^{(1)}_{\max} \notin \{\bn_h\}_{h \in \mH}$: In this case, one can verify that $\bn^{(1)}_{\max} = \bn^{(2)}_{\max} = \bn_{\text{abs}}$, and for some $h \in \mH$, we have $\bn^{(2)}_{\min} = \bn_h$. Since the inputs is accepted in scenario $2$, we have
            \begin{align}\label{accpetnace_w1_nh_isnot_max}
                \bn_{\text{abs}} - \bn_h = \bn^{(2)}_{\max} - \bn^{(2)}_{\min} \leq \eta\Delta.
            \end{align}
            If $\bn^{(1)}_{\min} = \bn_h$, we have
            \begin{align}
                \bn^{(1)}_{\max} - \bn^{(1)}_{\min} = \bn_{\text{abs}} - \bn_h = \bn^{(2)}_{\max} - \bn^{(2)}_{\min} =  \overset{(a)}{\leq} \eta\Delta,
            \end{align}
            where (a) follows from \eqref{accpetnace_w1_nh_isnot_max}. This implies that the inputs is accepted is scenario $1$. Furthermore, we have
            \begin{align}
                |\bn^{(1)}_{\max} + \bn^{(1)}_{\min}| = |\bn_{\text{abs}} + \bn_h| = |\bn^{(2)}_{\max} + \bn^{(2)}_{\min}|.
            \end{align}
            
            However, if $\bn^{(1)}_{\min} \neq \bn_h$, it implies that $\bn^{(1)}_{\min} = \bn^{(1)}_i$, for some $i \in \mT$. Based on \eqref{non_cancelling
_noise_Assumption}, this implies that the inputs in scenario $1$ is also  accepted.
On the other hand, since $\bn_{\text{abs}} = \underset{a \in \mT}{\max |\bn^{(1)}_a|}$ and $\bn_{\text{abs}} \geq 0$, it implies that $|\bn^{(1)}_i| \leq |\bn_{\text{abs}}| = \bn_{\text{abs}}$, which leads to
\begin{align}\label{nmin_nabs_positive_case}
    |\bn_{\text{abs}} + \bn^{(1)}_i| = \bn_{\text{abs}} + \bn^{(1)}_i.
\end{align}
Moreover, we have
\begin{align}\label{nmin_nh_positive_case}
    -\bn_{\text{abs}} \leq \bn^{(1)}_i \overset{(a)}{\leq} \bn_h \overset{(b)}{\leq} \bn_{\text{abs}},
\end{align}
where (a) follows from the fact that $\bn^{(1)}_{\min} = \bn^{(1)}_i$, (b) follows from the fact  that  $\bn^{(1)}_{\max} = \bn_{\text{abs}}$. \eqref{nmin_nh_positive_case} leads to 
\begin{align}\label{abs_nmin_nh_positive_case}
    |\bn_{\text{abs}} + \bn_h| = \bn_{\text{abs}} + \bn_h .
\end{align}
Therefore, we have
\begin{align}
    |\bn^{(1)}_{\max} + \bn^{(1)}_{\min}| = |\bn_{\text{abs}} + \bn^{(1)}_i| \overset{(a)}{=} \bn_{\text{abs}} + \bn^{(1)}_i \overset{(b)}{\leq} \bn_{\text{abs}} + \bn_h \overset{(c)}{=} |\bn_{\text{abs}} + \bn_h| = |\bn^{(2)}_{\max} + \bn^{(2)}_{\min}|,
\end{align}
where (a) follows from \eqref{nmin_nabs_positive_case}, (b) follows from the fact that 
$\bn^{(1)}_{\min} = \bn^{(1)}_i$, (c) follows from \eqref{abs_nmin_nh_positive_case}.
            \item $\bn^{(1)}_{\max} \in \{\bn_h\}_{h \in \mH}$: Let $\bn^{(1)}_{\max} = \bn_h $, for some $h \in \mH$. In this case, one can verify that 
            \begin{align}\label{same_max_2_scenarios}
                \bn^{(2)}_{\max} = \bn^{(1)}_{\max} = \bn_h.
            \end{align}
             We have $0  \leq |\bn_{\text{abs}}| = \bn_{\text{abs}} \leq \bn_h \leq \Delta$. This implies that for all $i \in \mT$, we have $|\bn^{(1)}_i| \leq |\bn_{\text{abs}}| = \bn_{\text{abs}} \leq \Delta$. Consequently, for all $i,j \in [N]$, we have
            \begin{align}
                |\by^{(1)}_i - \by^{(1)}_j| = |\bn^{(1)}_i - \bn^{(1)}_j| \leq 2\Delta \leq \eta\Delta,
            \end{align}
            meaning that the inputs in scenario $1$ is accepted. 
            Consider the following cases:
            \begin{itemize}
                \item $\bn^{(1)}_{\min} \in \{\bn_h\}_{h \in \mH}$: In this case one can verify that 
                \begin{align}
                    \bn^{(1)}_{\min} = \bn^{(2)}_{\min}
                \end{align}
                 Thus, based on \eqref{same_max_2_scenarios}, we have 
                \begin{align}
                    |\bn^{(2)}_{\max} + \bn^{(2)}_{\min}| = |\bn^{(1)}_{\max} + \bn^{(1)}_{\min}|.
                \end{align}
                \item $\bn^{(1)}_{\min} \notin \{\bn_h\}_{h \in \mH}$: Let $\bn^{(1)}_{\min} = \bn^{(1)}_j$, for some $j \in \mT$. Note that
                \begin{align}\label{max_and_mins_of_two_scenarios}
                    -\bn_h \leq -\bn_{\text{abs}} \leq \bn^{(1)}_j = \bn^{(1)}_{\min} \overset{(a)}{\leq} \bn^{(2)}_{\min} \leq \bn^{(2)}_{\max} \overset{(b)}{=}\bn_h,
                \end{align}
                where (a) follows from the fact that $\underline{\by}^{(2)} \subseteq \underline{\by}^{(1)}$, and (b) follows from \eqref{same_max_2_scenarios}.
                This implies that
            \begin{align}
                |\bn^{(1)}_{\max} + \bn^{(1)}_{\min}| = |\bn^{(1)}_j + \bn_h| \overset{(a)}{=} \bn^{(1)}_j + \bn_h \overset{(b)}{\leq} \bn^{(2)}_{\min} + \bn_h \overset{(c)}{=} |\bn^{(2)}_{\min} + \bn_h| = |\bn^{(2)}_{\max} + \bn^{(2)}_{\min}|,
            \end{align}
            where (a) and (c) follows from \eqref{max_and_mins_of_two_scenarios}, and (b) follows from $\bn^{(1)}_j = \bn^{(1)}_{\min} \leq  \bn^{(2)}_{\min}$.
            \end{itemize}

        \end{itemize}
        Thus, for the case of $\bn_{\text{abs}} \geq 0$, any realization of the noises that leads to acceptance in scenario $2$ would also result in an acceptance event in scenario $1$, and vice versa, and also \eqref{scenario2_has_bigger_mse} is valid in this case as well.
        \item $\bn_{\text{abs}} < 0$: 
        We consider two cases:
        \begin{itemize}
            \item $\bn^{(1)}_{\min} \notin \{\bn_h\}_{h \in \mH}$: In this case, one can verify that $\bn^{(1)}_{\min} = \bn^{(2)}_{\min} = \bn_{\text{abs}}$, and for some $h \in \mH$, we have $\bn^{(2)}_{\max} = \bn_h$. On the other hand, since the inputs is accepted in scenario $2$, we have
            \begin{align}\label{accpetnace_w1_nh_isnot_min}
               \bn_h - \bn_{\text{abs}}   = \bn^{(2)}_{\max} - \bn^{(2)}_{\min} \leq \eta\Delta.
            \end{align}
            If $\bn^{(1)}_{\max} = \bn_h$, we have
            \begin{align}
                \bn^{(1)}_{\max} - \bn^{(1)}_{\min} = \bn_h -\bn_{\text{abs}}   = \bn^{(2)}_{\max} - \bn^{(2)}_{\min} =  \overset{(a)}{\leq} \eta\Delta,
            \end{align}
            where (a) follows from \eqref{accpetnace_w1_nh_isnot_min}. This implies that the inputs is accepted is scenario $1$. Furthermore, we have
            \begin{align}
                |\bn^{(1)}_{\max} + \bn^{(1)}_{\min}| = |\bn_{\text{abs}} + \bn_h| = |\bn^{(2)}_{\max} + \bn^{(2)}_{\min}|.
            \end{align}
            
            However, if $\bn^{(1)}_{\max} \neq \bn_h$, it implies that $\bn^{(1)}_{\max} = \bn^{(1)}_i$, for some $i \in \mT$. Based on \eqref{non_cancelling
_noise_Assumption}, this implies that the inputs in scenario $1$ is also  accepted.
On the other hand, since $\bn_{\text{abs}} = \underset{a \in \mT}{\max |\bn^{(1)}_a|}$ and $\bn_{\text{abs}} < 0$, it implies that $|\bn^{(1)}_i| \leq |\bn_{\text{abs}}| = -\bn_{\text{abs}}$, which leads to
\begin{align}\label{nmin_nabs_negative_case}
    |\bn_{\text{abs}} + \bn^{(1)}_i| = -\bn_{\text{abs}} - \bn^{(1)}_i.
\end{align}
Moreover, we have
\begin{align}\label{nmax_nh_negative_case}
    \bn_{\text{abs}} \overset{(a)}{\leq} \bn_h \overset{(b)}{\leq} \bn^{(1)}_i  \leq -\bn_{\text{abs}},
\end{align}
where (a) follows from the fact that $\bn^{(1)}_{\min} = \bn_{\text{abs}}$, (b) follows from the fact  that  $\bn^{(1)}_{\max} = \bn^{(1)}_i$. \eqref{nmax_nh_negative_case} leads to 
\begin{align}\label{abs_nmax_nh_negative_case}
    |\bn_{\text{abs}} + \bn_h| = -\bn_{\text{abs}} - \bn_h .
\end{align}
Therefore, we have
\begin{align}
    |\bn^{(1)}_{\max} + \bn^{(1)}_{\min}| = |\bn_{\text{abs}} + \bn^{(1)}_i| \overset{(a)}{=} - \bn_{\text{abs}} - \bn^{(1)}_i \overset{(b)}{\leq} -\bn_{\text{abs}} - \bn_h \overset{(c)}{=} |\bn_{\text{abs}} + \bn_h| = |\bn^{(2)}_{\max} + \bn^{(2)}_{\min}|,
\end{align}
where (a) follows from \eqref{nmin_nabs_negative_case}, (b) follows from the fact that 
$\bn^{(1)}_{\max} = \bn^{(1)}_i$, (c) follows from \eqref{abs_nmax_nh_negative_case}.
            \item $\bn^{(1)}_{\min} \in \{\bn_h\}_{h \in \mH}$: 
            Let $\bn^{(1)}_{\min} = \bn_h $, for some $h \in \mH$. In this case, one can verify that \begin{align}\label{same_min_2_scenarios}
                \bn^{(2)}_{\min} = \bn^{(1)}_{\min} = \bn_h.
            \end{align}
            We have $0  > -|\bn_{\text{abs}}| = \bn_{\text{abs}} \geq \bn_h \geq -\Delta$. This implies that for all $i \in \mT$, we have $|\bn^{(1)}_i| \leq |\bn_{\text{abs}}| = -\bn_{\text{abs}} \leq \Delta$. Consequently, for all $i,j \in [N]$, we have
            \begin{align}
                |\by^{(1)}_i - \by^{(1)}_j| = |\bn^{(1)}_i - \bn^{(1)}_j| \leq 2\Delta \leq \eta\Delta,
            \end{align}
            meaning that the inputs in scenario $1$ is accepted. 
            Consider the following cases:
            \begin{itemize}
               \item $\bn^{(1)}_{\max} \in \{\bn_h\}_{h \in \mH}$: In this case one can verify that 
                \begin{align}
                    \bn^{(1)}_{\max} = \bn^{(2)}_{\max}
                \end{align}
                 Thus, based on \eqref{same_min_2_scenarios}, we have 
                \begin{align}
                    |\bn^{(2)}_{\max} + \bn^{(2)}_{\min}| = |\bn^{(1)}_{\max} + \bn^{(1)}_{\min}|.
                \end{align}
                \item $\bn^{(1)}_{\max} \notin \{\bn_h\}_{h \in \mH}$: Let $\bn^{(1)}_{\max} = \bn^{(1)}_j$, for some $j \in \mT$. Note that
                \begin{align}\label{max_and_mins_of_two_scenarios_v2}
                    \bn_h \overset{(a)}{=}  \bn^{(2)}_{\min} 
                     \leq \bn^{(2)}_{\max} \overset{(b)}{\leq }\bn^{(1)}_{\max} =  \bn^{(1)}_j \leq - \bn_{\text{abs}} \leq -\bn_h,
                \end{align}
                where (a) follows from \eqref{same_min_2_scenarios}, and (b) follows from the fact that $\underline{\by}^{(2)} \subseteq \underline{\by}^{(1)}$. This implies that
            \begin{align}
                |\bn^{(1)}_{\max} + \bn^{(1)}_{\min}| = |\bn^{(1)}_j + \bn_h| \overset{(a)}{=} -\bn^{(1)}_j - \bn_h \overset{(b)}{\leq} -\bn^{(2)}_{\max} - \bn_h \overset{(c)}{=} |\bn^{(2)}_{\max} + \bn_h| = |\bn^{(2)}_{\max} + \bn^{(2)}_{\min}|,
            \end{align}
            where (a) and (c) follows from \eqref{max_and_mins_of_two_scenarios_v2}, and (b) follows from $\bn^{(1)}_j = \bn^{(1)}_{\max} \geq \bn^{(2)}_{\max} $.
            \end{itemize}
        \end{itemize}
        Consequently, for the case of $\bn_{\text{abs}} < 0$, any realization of the noises that leads to acceptance in scenario $2$ would also result in an acceptance event in scenario $1$, and vice versa, and also \eqref{scenario2_has_bigger_mse} is valid in this case as well.
    \end{enumerate}
    We proved Claim \ref{claim:Prob_is_same} and \ref{claim:mse_is_same} for all of the above cases, and thus the proof is complete. 

\section{Proof of Lemma \ref{lemma:making_symmetry}}\label{proof:lemma:making_symmetry}

For the formal proof, first consider the following lemma.
\begin{lemma}\label{lemma:flipping_noise}
    For any $\pare \in \Lambda_{\mathsf{DC}} $, $g(z) \in \Lambda_{\mathsf{AD}}$, and $N \geq 1$, we have
    \begin{align*}
        \mathsf{PA}_{N+1} (g_{\textrm{ref}}(.), \pare) &= \mathsf{PA}_{N+1} (g(.), \pare),  \\
        \mathsf{MSE}_{N+1}\big( g_{\textrm{ref}}(.), \pare\big) &= \mathsf{MSE}_{N+1}\big( \gdot, \pare\big),
    \end{align*}
    where $g_{\textrm{ref}}(z) \triangleq g(-z)$.
\end{lemma}
\begin{proof}
Let $\underline{\bn} = (\bn_1, \dots, \bn_{N+1})$, and $\underline{\bn}_h = (\bn_i)_{i \in \mH}$. Consider two scenarios. In scenario $1$ the adversary employs $\gdot$ as the noise distribution, whereas 
in scenario 2 the adversary employs $g_{\textrm{ref}}(.)$ as the noise distribution. We note that that for the joint probability density functions  $f^{(1)}_{\mathbf{n}_a, \underline{\bn}_h}$ and $f^{(2)}_{\mathbf{n}_a, \underline{\bn}_h}$, for scenarios $1$ and $2$ respectively, we have
\begin{align}\label{scenario1_and_scenario2}
    f^{(1)}_{\mathbf{n}_a, \underline{\bn}_h} (n_a,\underline{n}_h) = f^{(2)}_{\mathbf{n}_a, \underline{\bn}_h} (-n_a, -\underline{n}_h).
\end{align}

Recall that the acceptance rule is a function of $\max(\underline{\by}) - \min(\underline{\by}) = \max(\underline{\bn}) - \min(\underline{\bn})$.
This implies that for any realization of the noises in scenario $1$, i.e., $(n_a, \underline{n}_h)$  where the computation is accepted, the corresponding realization $(-n_a, -\underline{n}_h)$ in scenario $2$ is also accepted, and vice versa. This is because the absolute value of the difference between noises, is the same for both of these realizations. Thus, we have
\begin{align}\label{relation_between_gref_g_acc}
    \Pr \big(\mathcal{A}_{\pare}; g_{\textrm{ref}}(.)~|~ \bn_a=-n_a,\underline{\bn}_h=-\underline{n}_h\big) = \Pr \big(\mathcal{A}_{\pare}; \gdot~|~\bn_a=n_a,\underline{\bn}_h=\underline{n}_h\big).
\end{align}

This implies that
\begin{align}\label{flipped_noise_Acc}
     \mathsf{PA}_{N+1} (g_{\textrm{ref}}(.), \pare) &= \int_{-\infty}^{\infty}\int_{[-\Delta, \Delta]^N} \Pr \big(\mathcal{A}_{\pare}; g_{\textrm{ref}}(.)~|~\bn_a=n_a,\underline{\bn}_h=\underline{n}_h\big)
    f^{(2)}_{\mathbf{n}_a, \underline{\bn}_h} (n_a, \underline{n}_h) \,d\underline{n}_h \,dn_a\nonumber \\
    &\overset{(a)}{=} \int_{-\infty}^{\infty}\int_{[-\Delta, \Delta]^N} \Pr \big(\mathcal{A}_{\pare}; g_{\textrm{ref}}(.)~|~\bn_a=-n^{\prime}_a,\underline{\bn}_h= -\underline{n}^{\prime}_h\big)
    f^{(2)}_{\mathbf{n}_a, \underline{\bn}_h} (-n^{\prime}_a, -\underline{n}^{\prime}_h) \,d\underline{n}^{\prime}_h \,dn^{\prime}_a\nonumber \\
    &\overset{(b)}{=}
    \int_{-\infty}^{\infty}\int_{[-\Delta, \Delta]^N} \Pr \big(\mathcal{A}_{\pare}; \gdot~|~\bn_a=n^{\prime}_a,\underline{\bn}_h= \underline{n}^{\prime}_h\big)
    f^{(1)}_{\mathbf{n}_a, \underline{\bn}_h} (n^{\prime}_a, \underline{n}^{\prime}_h) \,d\underline{n}^{\prime}_h \,dn^{\prime}_a\nonumber \\
    &=\mathsf{PA}_{N+1} (g(.), \pare),
\end{align}
where (a) follows from changing the parameters of the integral by $\underline{n}^{\prime}_h = -\underline{n}_h$, $n^{\prime}_a = -n_a$, (b) follows from \eqref{scenario1_and_scenario2} and \eqref{relation_between_gref_g_acc}.

Let $\underline{\by}^{(1)}$ denotes the observed inputs in scenario 1, and $\underline{\by}^{(2)}$ denotes the observed inputs in scenario 2. We have
\begin{align}\label{flipped_noise_Mse}
    \mathsf{MSE}_{N+1}&\big( g_{\textrm{ref}}(.), \pare\big) \nonumber \\
    &= \mathbb{E}\big[\big(\mathbf{u} - \frac{\max(\underline{\by}^{(2)}) + \min(\underline{\by}^{(2)})}{2}\big)^2 | \mathcal{A}_{\pare}; g_{\textrm{ref}}(.)\big] \nonumber\\
    &=\int_{-\infty}^{\infty}\int_{[-\Delta, \Delta]^N} \mathbb{E}\big[\big(\mathbf{u} - \frac{\max(\underline{\by}^{(2)}) + \min(\underline{\by}^{(2)})}{2}\big)^2 | \mathcal{A}_{\pare}, n_a, \underline{n}_h; g_{\textrm{ref}}(.)\big]
    f^{(2)}_{\mathbf{n}_a, \underline{\bn}_h|\mathcal{A}_{\pare}} (n_a, \underline{n}_h|\mathcal{A}_{\pare} ; g_{\textrm{ref}}(.)) \,d\underline{n}_h \,dn_a
    \nonumber\\
    &\overset{(a)}{=}\int_{-\infty}^{\infty}\int_{[-\Delta, \Delta]^N}  \frac{(n_{\text{max}} + n_{\text{min}})^2}{4}\times 
    \frac{\Pr\big( \mathcal{A}_{\pare}; g_{\textrm{ref}}(.)|n_a,\underline{n}_h\big)f^{(2)}_{\mathbf{n}_a, \underline{\bn}_h} (n_a, \underline{n}_h)}{\Pr\big( \mathcal{A}_{\pare}; g_{\textrm{ref}}(.)\big)} \,d\underline{n}_h \,dn_a \nonumber \\
    &\overset{(b)}{=}\int_{-\infty}^{\infty}\int_{[-\Delta, \Delta]^N} \frac{(-n^{\prime}_{\text{min}} - n^{\prime}_{\text{max}})^2}{4}\times 
    \frac{\Pr\big( \mathcal{A}_{\pare}; g_{\textrm{ref}}(.)|-n^{\prime}_a,-\underline{n}^{\prime}_h\big)f^{(2)}_{\mathbf{n}_a, \underline{\bn}_h} (-n^{\prime}_a, -\underline{n}^{\prime}_h)}{\Pr\big( \mathcal{A}_{\pare}; g_{\textrm{ref}}(.)\big)} \,d\underline{n}^{\prime}_h \,dn^{\prime}_a \nonumber \\
    &\overset{(c)}{=}\int_{-\infty}^{\infty}\int_{[-\Delta, \Delta]^N} \frac{(n^{\prime}_{\text{min}} + n^{\prime}_{\text{max}})^2}{4}\times 
    \frac{\Pr\big( \mathcal{A}_{\pare}; \gdot|n^{\prime}_a,\underline{n}^{\prime}_h\big)f^{(1)}_{\mathbf{n}_a, \underline{\bn}_h} (n^{\prime}_a, \underline{n}^{\prime}_h)}{\Pr\big( \mathcal{A}_{\pare}; \gdot\big)} \,d\underline{n}^{\prime}_h \,dn^{\prime}_a \nonumber \\
    &=\int_{-\infty}^{\infty}\int_{[-\Delta, \Delta]^N} \mathbb{E}\big[\big(\mathbf{u} - \frac{\max(\underline{\by}^{(1)}) + \min(\underline{\by}^{(1)})}{2}\big)^2 | \mathcal{A}_{\pare}, n^{\prime}_a, \underline{n}^{\prime}_h; \gdot\big]
    f^{(1)}_{\mathbf{n}_a, \underline{\bn}_h|\mathcal{A}_{\pare}} (n^{\prime}_a, \underline{n}^{\prime}_h|\mathcal{A}_{\pare}; \gdot) \,d\underline{n}^{\prime}_h \,dn^{\prime}_a
    \nonumber \\
    &=\mathbb{E}\big[\big(\mathbf{u} - \frac{\max(\underline{\by}^{(1)}) + \min(\underline{\by}^{(1)})}{2}\big)^2 | \mathcal{A}_{\pare}; \gdot\big] 
    \nonumber\\
    &=\mathsf{MSE}_{N+1}\big( \gdot, \pare\big) 
\end{align}
where in (a) we have $n_{\text{max}} = \max\{ n_a, \underline{n}_h\}$, and $n_{\text{min}} = \min\{ n_a, \underline{n}_h\}$, (b) follows from changing the parameters of the integral by $\underline{n}^{\prime}_h = -\underline{n}_h$, $n^{\prime}_a = -n_a$, $n^{\prime}_{\text{max}} = \max\{ n^{\prime}_a, \underline{n}^{\prime}_h\}$, and $n^{\prime}_{\text{min}} = \min\{ n^{\prime}_a, \underline{n}^{\prime}_h\}$, (c) follows from \eqref{scenario1_and_scenario2}, \eqref{relation_between_gref_g_acc}, and \eqref{flipped_noise_Acc}.

\end{proof}
We now proceed to prove Lemma \ref{lemma:making_symmetry}.  
Consider a situation where the adversary utilizes a $\mathsf{Bernoulli}(\frac{1}{2})$ random variable $\mathbf{c}$ to determine its noise distribution. Specifically, when $\mathbf{c}=0$, the adversary selects $\gdot$ as the noise distribution, and when $\mathbf{c}=1$, it chooses $g_{\textrm{ref}}(.)$. The random variable $\mathbf{c}$ is independent of all other variables in the system. Let $g_c(.)$ represent the adversary’s noise distribution under this setup. It follows straightforwardly that $g_c(z) = g_{\textrm{sym}}(z)$.
Therefore, we have
\begin{align}\label{bernoli_noise_
acc}
    \Pr \left(\mathcal{A}_{\pare}; g_{\textrm{sym}}(.)\right) &=\Pr \left(\mathcal{A}_{\pare}; g_c(.)\right) \nonumber \\
    &= \Pr \left(\mathcal{A}_{\pare}; g_c(.) | \mathbf{c} = 0\right)\Pr(\mathbf{c} = 0) + 
    \Pr \left(\mathcal{A}_{\pare}; g_c(.) | \mathbf{c} = 1\right)\Pr(\mathbf{c} = 1)
    \nonumber\\
    &=\Pr \left(\mathcal{A}_{\pare}; \gdot\right)\frac{1}{2} + \Pr \left(\mathcal{A}_{\pare}; g_{\textrm{ref}}(.)\right)\frac{1}{2} \nonumber \\
    &=\frac{\Pr \left(\mathcal{A}_{\pare}; \gdot\right) + \Pr \left(\mathcal{A}_{\pare}; g_{\textrm{ref}}(.)\right)}{2} \nonumber \\
    &\overset{(a)}{=} \Pr \left(\mathcal{A}_{\pare}; \gdot\right),
\end{align}
where (a) follows from Lemma \ref{lemma:flipping_noise}. 

Additionally, we have
\begin{align}
    \mathsf{MSE}_{N+1}\big( g_{\textrm{sym}}(.), \pare\big) &= \mathsf{MSE}_{N+1}\big( g_c(.), \pare\big) \\ \nonumber 
    &=  \mathbb{E}\big[\big(\mathbf{u} - \frac{\max(\underline{\by}) + \min(\underline{\by})}{2}\big)^2 | \mathcal{A}_{\pare}; g_c(.)\big] \nonumber\\
    &=  \mathbb{E}\big[\big(\mathbf{u} - \frac{\max(\underline{\by}) + \min(\underline{\by})}{2}\big)^2 | \mathcal{A}_{\pare}, \mathbf{c}=0; g_c(.)\big]\Pr \big( \mathbf{c}=0|\mathcal{A}_{\pare}; g_c(.) \big) \nonumber \\
    &+ 
    \mathbb{E}\big[\big(\mathbf{u} - \frac{\max(\underline{\by}) + \min(\underline{\by})}{2}\big)^2 | \mathcal{A}_{\pare}, \mathbf{c}=1; g_c(.)\big]\Pr \big( \mathbf{c}=1|\mathcal{A}_{\pare}; g_c(.) \big) \nonumber \\
    &= \mathbb{E}\big[\big(\mathbf{u} - \frac{\max(\underline{\by}) + \min(\underline{\by})}{2}\big)^2 | \mathcal{A}_{\pare}; \gdot\big]
    \frac{\Pr \big( \mathcal{A}_{\pare}; g_c(.)|\mathbf{c} = 0 \big)\Pr(\mathbf{c}=0)}{\Pr \big( \mathcal{A}_{\pare}; g_c(.) \big)} \nonumber\\
    &+\mathbb{E}\big[\big(\mathbf{u} - \frac{\max(\underline{\by}) + \min(\underline{\by})}{2}\big)^2 | \mathcal{A}_{\pare}; g_{\textrm{ref}}(.)\big]
    \frac{\Pr \big( \mathcal{A}_{\pare}; g_c(.)|\mathbf{c} = 1 \big)\Pr(\mathbf{c}=1)}{\Pr \big( \mathcal{A}_{\pare}; g_c(.) \big)}\nonumber \\
    &= \mathsf{MSE}_{N+1}\big( \gdot, \pare\big)
    \frac{\frac{1}{2}\Pr \big( \mathcal{A}_{\pare}; \gdot \big)}{\Pr \big( \mathcal{A}_{\pare}; g_c(.) \big)} \nonumber \\
    &+ \mathsf{MSE}_{N+1}\big( g_{\textrm{ref}}(.), \pare\big)
    \frac{\frac{1}{2}\Pr \big( \mathcal{A}_{\pare}; g_{\textrm{ref}}(.) \big)}{\Pr \big( \mathcal{A}_{\pare}; g_c(.) \big)} \nonumber \\
     &\overset{(a)}{=}  \frac{1}{2}\mathsf{MSE}_{N+1}\big( \gdot, \pare\big) + \frac{1}{2}\mathsf{MSE}_{N+1}\big( g_{\textrm{ref}}(.), \pare\big)
    \nonumber \\
     &\overset{(b)}{=}  \mathsf{MSE}_{N+1}\big( \gdot, \pare\big),
\end{align}
where (a) follows from \eqref{bernoli_noise_
acc}, and (b) follows from Lemma \ref{lemma:flipping_noise}.
This completes the proof of Lemma \ref{lemma:making_symmetry}.
\section{Proof of Lemma \ref{general_format_symmetric}}\label{proof:general_format_symmetric}

    Let $\underline{\by} = (\by_1,\dots,\by_{N+1})$ and $\underline{\bn} = (\bn_1,\dots,\bn_{N+1})$. To prove Lemma \ref{general_format_symmetric}, we first show \eqref{general_symmetric_acc}. For any $\gdot \in \Lambda_{\mathsf{AD}}$, and $N \geq 1$, note that
\begin{align}\label{acc_v1}
    \mathsf{PA}_{N+1} \left( \gdot, \pare \right) = \int_{-\infty}^{\infty}\Pr \big( \acce| \bn_a = z \big) g(z) \,dz.
\end{align}
Recall that the acceptance rule is $\max(\underline{\by}) - \min(\underline{\by}) \leq \eta \Delta $, or equivalently $\max(\underline{\bn}) - \min(\underline{\bn}) \leq \eta\Delta$. Also, for any $h \in \mH$, we assume that $\Pr (|\bn_h| > \Delta) =0$. Thus, for the case of $|z| > (\eta+1)\Delta$, $\Pr \big( \acce| \bn_a = z \big) = 0$, and for the case of $|z| \leq (\eta-1)\Delta$, $\Pr \big( \acce| \bn_a = z \big) = 1$. 
Let $\bn_{\text{min}} = \min \bn_h$ , and $\bn_{\text{max}} = \max \bn_h$ for $h \in \mH$. One can easily verify that
\begin{align}
    f_{\bn_{\text{min}}}(x) &= N f_{\bn_{h}}(x) (1-F_{\bn_h})^{(N-1)}, \label{pdf of min}\\
    f_{\bn_{\max}}(x) &=  f_{\bn_{\min}}(-x),\label{pdf_max_min}
\end{align}
where $f_{\bn_{\text{min}}}(.)$ is the PDF of $\bn_{\text{min}}$, and $f_{\bn_{\max}}(.)$ is the PDF of $\bn_{\max}$.

Note that  for the case of $(\eta-1)\Delta \leq z \leq (\eta+1)\Delta$, we have $\max(\underline{\bn}) = z$, and $\min(\underline{\bn}) = \bn_{\text{min}}$. Thus
\begin{align}\label{conditional_acc_positive_bounded}
    \Pr \big( \acce| \bn_a = z \big) = 
    \Pr \big( \bn_{\text{min}} \geq z-\eta\Delta \big)
    =\underset{h\in\mH}{\Pi}\Pr \big(\bn_h \geq z-\eta\Delta \big)
    =
    \big(\int_{z-\eta\Delta}^{\Delta} f_{\bn_h}(x) \,dx\big)^N.
\end{align}
Also, in the case of 
$-(\eta+1)\Delta \leq z \leq -(\eta-1)\Delta$, we have $\min(\underline{\bn}) = z$,  and $\max(\underline{\bn}) = \bn_{\text{max}}$. Thus
\begin{align}
    \Pr \big( \acce| \bn_a = z \big) = 
    \Pr \big( \bn_{\text{max}} \leq z+\eta\Delta \big)
    =
    \underset{h\in\mH}{\Pi}\Pr \big(\bn_h \leq z+\eta\Delta \big)
    \nonumber \\
    =
    \big( \int_{-\Delta}^{z+\eta\Delta} f_{\bn_h}(x) \,dx\big)^N  = 
    \big(\int_{-z-\eta\Delta}^{\Delta} f_{\bn_h}(x) \,dx\big)^N= \Pr \big( \acce| \bn_a = -z \big).
\end{align}
This implies that $\Pr \big( \acce| \bn_a = z \big)$ is a symmetric function with respect to $z$. Therefore, we can rewrite \eqref{acc_v1} as 
\begin{align}
    \mathsf{PA}_{N+1} \left( \gdot, \pare \right) &= \int_{-\infty}^{\infty}\Pr \big( \acce| \bn_a = z \big) g(z) \,dz \nonumber \\
    &= 2\int_{0}^{\infty}\Pr \big( \acce| \bn_a = z \big) g(z) \,dz \nonumber \\
    & \overset{(a)}{=} 2\int_{0}^{(\eta+1)\Delta}\Pr \big( \acce| \bn_a = z \big) g(z) \,dz
    \label{general_acc_for_symmetric} \\
    & \overset{(b)}{=}  2\int_{\Delta}^{(\eta-1)\Delta}\Pr \big( \acce| \bn_a = z \big) g(z) \,dz  + 2\int_{(\eta-1)\Delta}^{(\eta+1)\Delta}\Pr \big( \acce| \bn_a = z \big) g(z) \,dz, \label{semi_general_acc_symmetric} \\
    & \overset{(c)}{=} 2\int_{\Delta}^{(\eta-1)\Delta} g(z) \,dz + 2\int_{(\eta-1)\Delta}^{(\eta+1)\Delta}\Pr \big(\bn_{\text{min}} \in [z-\eta\Delta, \Delta] \big) g(z) \,dz
    \nonumber \\
    & = 2\int_{\Delta}^{(\eta-1)\Delta} g(z) \,dz + 2\int_{(\eta-1)\Delta}^{(\eta+1)\Delta} \left( 
  \int_{z-\eta\Delta}^{\Delta} f_{\bn_{\text{min}}}(x) \,dx\right) g(z) \,dz,\label{final_PA}
\end{align}
where (a) follows from the fact that for the case of $|z| > (\eta+1)\Delta$, $\Pr \big( \acce| \bn_a = z \big) = 0$, and (b) follows from the fact that 
$\Pr (|\bn_a| < \Delta) = 0$,
(c) follows from the fact that for the case of $|z| \leq (\eta-1)\Delta$, $\Pr \big( \acce| \bn_a = z \big) = 1$. Combining \eqref{final_PA} and \eqref{pdf of min} completes the proof of \eqref{general_symmetric_acc}. 

Now we prove \eqref{general_symmetric_mse}. For any $\gdot \in \Lambda_{\mathsf{AD}}$, note that
\begin{align}\label{costV1}
    \mathsf{MSE}_{N+1}\big( \gdot, \pare \big) &=\mathbb{E}\big[\big(\mathbf{u} - \frac{\max(\underline{\by}) + \min(\underline{\by})}{2}\big)^2 | \mathcal{A}_{\pare}\big] =\frac{1}{4}\mathbb{E}\big[\big(\max(\underline{\bn}) + \min(\underline{\bn})\big)^2 | \mathcal{A}_{\pare}\big]\nonumber \\
    &=\frac{1}{4}\int_{-\infty}^{\infty} \mathbb{E}[(\max(\underline{\bn}) + \min(\underline{\bn}))^2 \mid \mathcal{A}_{\pare}, \mathbf{n}_a = z]f_{\mathbf{n}_a|\mathcal{A}_{\pare}}(z|\acce) \,dz \nonumber \\
    &=\frac{1}{4}\int_{-\infty}^{\infty} \mathbb{E}[(\max(\underline{\bn}) + \min(\underline{\bn}))^2 \mid \mathcal{A}_{\pare}, \mathbf{n}_a = z]\frac{\Pr(\mathcal{A}_{\pare}|\mathbf{n}_a =z)g(z)}{\mathsf{PA}_{N+1} \left( \gdot, \pare \right)} \,dz \nonumber \\
    &=\frac{1}{4\mathsf{PA}_{N+1} \left( \gdot, \pare \right)}\int_{-\infty}^{\infty} \mathbb{E}[(\max(\underline{\bn}) + \min(\underline{\bn}))^2 \mid \mathcal{A}_{\pare}, \mathbf{n}_a = z]\Pr(\mathcal{A}_{\pare}|\mathbf{n}_a =z)g(z) \,dz \nonumber \\
    &\overset{(a)}{=}\frac{1}{4\mathsf{PA}_{N+1} \left( \gdot, \pare \right)}\int_{-(\eta+1)\Delta}^{(\eta+1)\Delta} \mathbb{E}[(\max(\underline{\bn}) + \min(\underline{\bn}))^2 \mid \mathcal{A}_{\pare}, \mathbf{n}_a = z]\Pr(\mathcal{A}_{\pare}|\mathbf{n}_a =z)g(z) \,dz \nonumber \\
    &\overset{(b)}{=}\frac{1}{4\mathsf{PA}_{N+1} \left( \gdot, \pare \right)}\int_{-(\eta+1)\Delta}^{-(\eta-1)\Delta} \mathbb{E}[(\max(\underline{\bn}) + \min(\underline{\bn}))^2 \mid \mathcal{A}_{\pare}, \mathbf{n}_a = z]\Pr(\mathcal{A}_{\pare}|\mathbf{n}_a =z)g(z) \,dz\nonumber \\
    &+\frac{1}{4\mathsf{PA}_{N+1} \left( \gdot, \pare \right)}\int_{-(\eta-1)\Delta}^{-\Delta} \mathbb{E}[(\max(\underline{\bn}) + \min(\underline{\bn}))^2 \mid \mathcal{A}_{\pare}, \mathbf{n}_a = z]\Pr(\mathcal{A}_{\pare}|\mathbf{n}_a =z)g(z) \,dz\nonumber \\
    &+ \frac{1}{4\mathsf{PA}_{N+1} \left( \gdot, \pare \right)}\int_{\Delta}^{(\eta-1)\Delta} \mathbb{E}[(\max(\underline{\bn}) + \min(\underline{\bn}))^2 \mid \mathcal{A}_{\pare}, \mathbf{n}_a = z]\Pr(\mathcal{A}_{\pare}|\mathbf{n}_a =z)g(z) \,dz \nonumber \\
    &+\frac{1}{4\mathsf{PA}_{N+1} \left( \gdot, \pare \right)}\int_{(\eta-1)\Delta}^{(\eta+1)\Delta} \mathbb{E}[(\max(\underline{\bn}) + \min(\underline{\bn}))^2 \mid \mathcal{A}_{\pare}, \mathbf{n}_a = z]\Pr(\mathcal{A}_{\pare}|\mathbf{n}_a =z)g(z) \,dz,
    \end{align}
    where (a) follows from the fact that for $|z| > (\eta+1)\Delta$, $\Pr \big( \acce| \bn_a = z \big) = 0$, and (b) follows from the fact that $\Pr (|\bn_a| < \Delta) = 0$.
    
    Note that for the case of $z > \Delta$, we have $\max(\underline{\bn}) = z$, and $\min(\underline{\bn}) = \bn_{\text{min}}$. Thus
    \begin{align}\label{exp_pos_z}
        \mathbb{E}[(\max(\underline{\bn}) + \min(\underline{\bn}))^2 &\mid \mathcal{A}_{\pare}, \mathbf{n}_a = z]\Pr \big( \acce| \bn_a = z \big) \nonumber \\
        &= \int_{-\Delta }^{\Delta}(x+ z)^2f_{\bn_{\min}|\bn_a,\acce}(x|z,\acce) \Pr \big( \acce| \bn_a = z \big) \,dx  \ \nonumber \\
        &= \int_{-\Delta }^{\Delta}(x+ z)^2\frac{\Pr \big( \acce| \bn_a = z , \bn_{\min} = x\big)f_{\bn_a|\bn_{\min}}(z|x)f_{\bn_{\min}}(x)}{\Pr \big( \acce| \bn_a = z \big)g(z)} \Pr \big( \acce| \bn_a = z \big) \,dx \nonumber \\
        &= \int_{-\Delta }^{\Delta}(x+ z)^2\frac{\Pr \big( \acce| \bn_a = z , \bn_{\min} = x\big)f_{\bn_a|\bn_{\min}}(z|x)f_{\bn_{\min}}(x)}{g(z)}  \,dx \nonumber \\
        & \overset{(a)}{=} \int_{-\Delta }^{\Delta}(x+ z)^2\Pr \big( \acce| \bn_a = z , \bn_{\min} = x\big)f_{\bn_{\min}}(x) \,dx,
    \end{align}
    where (a) follows from the fact that $\bn_a$ and $\bn_{\min}$ are independent and thus $f_{\bn_a|\bn_{\min}}(z|x)  = g(z)$.
    Thus, based on \eqref{exp_pos_z}, for the case of $\Delta \leq z < (\eta-1)\Delta$, since $\Pr \big( \acce| \bn_a = z , \bn_{\min} = x\big) = 1$,
    we have
    \begin{align}\label{exp_pos_z_v1}
        \mathbb{E}[(\max(\underline{\bn}) + \min(\underline{\bn}))^2 &\mid \mathcal{A}_{\pare}, \mathbf{n}_a = z]\Pr \big( \acce| \bn_a = z \big) =  \int_{-\Delta }^{\Delta}(x+ z)^2 f_{\bn_{\min}}(x) \,dx.
    \end{align}
    Additionally, for the case $(\eta-1)\Delta \leq z \leq (\eta+1)\Delta$, based on \eqref{exp_pos_z}, we have
    \begin{align}\label{exp_pos_z_v2}
        \mathbb{E}[(\max(\underline{\bn}) + \min(\underline{\bn}))^2 &\mid \mathcal{A}_{\pare}, \mathbf{n}_a = z]\Pr \big( \acce| \bn_a = z \big) =  \int_{z-\eta\Delta }^{\Delta}(x+ z)^2 f_{\bn_{\min}}(x) \,dx.
    \end{align}
    Similarly, for the case of $z <-\Delta$, we have $\min(\underline{\bn}) = z$, and $\max(\underline{\bn}) = \bn_{\text{max}}$. Thus
    \begin{align}\label{exp_neg_z}
        \mathbb{E}[(\max(\underline{\bn}) + \min(\underline{\bn}))^2 &\mid \mathcal{A}_{\pare}, \mathbf{n}_a = z]\Pr \big( \acce| \bn_a = z \big) \nonumber \\
        &= \int_{-\Delta }^{\Delta}(x+ z)^2f_{\bn_{\max}|\bn_a,\acce}(x|z,\acce) \Pr \big( \acce| \bn_a = z \big) \,dx  \ \nonumber \\
        &= \int_{-\Delta }^{\Delta}(x+ z)^2\frac{\Pr \big( \acce| \bn_a = z , \bn_{\max} = x\big)f_{\bn_a|\bn_{\max}}(z|x)f_{\bn_{\max}}(x)}{\Pr \big( \acce| \bn_a = z \big)g(z)} \Pr \big( \acce| \bn_a = z \big) \,dx \nonumber \\
        &= \int_{-\Delta }^{\Delta}(x+ z)^2\frac{\Pr \big( \acce| \bn_a = z , \bn_{\max} = x\big)f_{\bn_a|\bn_{\max}}(z|x)f_{\bn_{\max}}(x)}{g(z)}  \,dx \nonumber \\
        & \overset{(a)}{=} \int_{-\Delta }^{\Delta}(x+ z)^2\Pr \big( \acce| \bn_a = z , \bn_{\max} = x\big)f_{\bn_{\max}}(x) \,dx,
    \end{align}
    where (a) follows from the fact that $\bn_a$ and $\bn_{\max}$ are independent and thus $f_{\bn_a|\bn_{\max}}(z|x)  = g(z)$.
    Thus, based on \eqref{exp_neg_z}, for the case of $-(\eta-1)\Delta < z \leq  -\Delta$, since $\Pr \big( \acce| \bn_a = z , \bn_{\max} = x\big) = 1$, we have
    \begin{align}\label{exp_neg_z_v1}
        \mathbb{E}[(\max(\underline{\bn}) + \min(\underline{\bn}))^2 &\mid \mathcal{A}_{\pare}, \mathbf{n}_a = z]\Pr \big( \acce| \bn_a = z \big) \nonumber \\
        & = \int_{-\Delta }^{\Delta}(x+ z)^2f_{\bn_{\max}}(x) \,dx \nonumber \\
        & \overset{(a)}{=} \int_{-\Delta }^{\Delta}(x+ z)^2f_{\bn_{\min}}(-x) \,dx \nonumber \\
        & \overset{(b)}{=} \int_{-\Delta }^{\Delta}(-x^{\prime}+ z)^2f_{\bn_{\min}}(x^{\prime}) \,dx^{\prime} \nonumber \\
        & \overset{(c)}{=} \int_{-\Delta }^{\Delta}(x^{\prime}+ z)^2f_{\bn_{\min}}(x^{\prime}) \,dx^{\prime}
    \end{align}
    where (a) follows from \eqref{pdf_max_min}, (b) follows changing the parameter of the integral by $x^{\prime} = -x$, and in (c), we have $\Delta \leq z < (\eta-1)\Delta$.
    Additionally, for the case of  $-(\eta+1)\Delta \leq z \leq  -(\eta-1)\Delta$, based on \eqref{exp_neg_z}, we have
    \begin{align}\label{exp_neg_z_v2}
        \mathbb{E}[(\max(\underline{\bn}) + \min(\underline{\bn}))^2 &\mid \mathcal{A}_{\pare}, \mathbf{n}_a = z]\Pr \big( \acce| \bn_a = z \big) \nonumber \\
        & = \int_{-\Delta }^{z+\eta\Delta}(x+ z)^2f_{\bn_{\max}}(x) \,dx \nonumber \\
        & \overset{(a)}{=} \int_{-\Delta }^{z+\eta\Delta}(x+ z)^2f_{\bn_{\min}}(-x) \,dx \nonumber \\
        & \overset{(b)}{=} \int_{-z-\eta\Delta}^{\Delta}(-x^{\prime}+ z)^2f_{\bn_{\min}}(x^{\prime}) \,dx^{\prime} \nonumber \\
        & \overset{(c)}{=} \int_{z-\eta\Delta }^{\Delta}(x^{\prime}+ z)^2f_{\bn_{\min}}(x^{\prime}) \,dx^{\prime}
    \end{align}
    where (a) follows from \eqref{pdf_max_min}, (b) follows changing the parameter of the integral by $x^{\prime} = -x$, and in (c), we have $(\eta-1)\Delta \leq z \leq (\eta+1)\Delta$.
    
    By combining \eqref{costV1}, \eqref{exp_pos_z_v1}, \eqref{exp_pos_z_v2},  \eqref{exp_neg_z_v1}, and 
    \eqref{exp_neg_z_v2}, we have
    \begin{align}
        \mathsf{MSE}_{N+1}\big( \gdot, \pare \big) &=\frac{1}{2\mathsf{PA}_{N+1}\left( g_2(.), \pare \right)}\int_{\Delta}^{(\eta+1)\Delta} \left(\int_{-\Delta}^{\Delta} (x+z)^2\Pr \big( \acce| \bn_a = z , \bn_{\min} = x)f_{\bn_{\min}}(x)\,dx \right)g_2(z) \,dz \label{mse_mean_general_proxy}
        \\
        &=\frac{1}{2\mathsf{PA}_{N+1} \left( \gdot, \pare \right)}\int_{\Delta}^{(\eta-1)\Delta} \left(\int_{-\Delta}^{\Delta} (x+z)^2f_{\bn_{\min}}(x)\,dx \right)g(z) \,dz
        \nonumber \\
        &+\frac{1}{2\mathsf{PA}_{N+1} \left( \gdot, \pare \right)}\int_{(\eta-1)\Delta}^{(\eta+1)\Delta} \left(\int_{z-\eta\Delta}^{\Delta} (x+z)^2f_{\bn_{\min}}(x)\,dx \right)g(z) \,dz.
    \end{align}
     This completes the proof of \eqref{general_symmetric_mse}.

\section{Proof of Lemma \ref{lemma:bounded_noise_existence}}\label{proof:lemma:bounded_noise_existence}

    To prove this lemma we first show that there exist a strong symmetric noise distribution $g_0(.)$, such that $g_0(.)=0$ for $|z| > (\eta+1)\Delta$, and 
    \begin{align*}
        \mathsf{PA}_{N+1}(g_0(.), \eta) &\geq \mathsf{PA}_{N+1}(g_1(.), \eta),  \\
        \mathsf{MSE}_{N+1}( g_0(.), \eta) &= \mathsf{MSE}_{N+1}( g_1(.), \eta) .
    \end{align*}
    To prove this, note that
    since $g_1(.)$ is symmetric, we have $\int_{0}^{\infty} g_1(z) \,dz = \frac{1}{2}$.
    Let  $b \triangleq \int_{(\eta+1)\Delta}^{\infty} g_1(z) \,dz$. Since $\mathsf{PA}_{N+1}(g_1(.), \eta) > 0$, based on \eqref{general_symmetric_acc}, it implies that $b < \frac{1}{2}$. Define  a symmetric noise distribution $g_0(.)$, where $g_0(z) = \frac{g_1(z)}{1 - 2b}$, for all $|z| \leq (\eta+1)\Delta$, and otherwise $g_0(z) = 0$.
Based on \eqref{general_symmetric_acc}, we have
\begin{align}\label{new_noise_bounded_acc}
    \mathsf{PA}_{N+1}\left( g_0(.), \pare \right) = \frac{\mathsf{PA}_{N+1}\left( g_1(.), \pare \right)}{1-2b} \geq \mathsf{PA}_{N+1}\left( g_1(.), \pare \right).
\end{align}
On the other hand, Based on \eqref{general_symmetric_mse}, we have
\begin{align}\label{new_noise_bounded_mse}
    \mathsf{MSE}_{N+1}\big( g_0(.), \pare \big) = \mathsf{MSE}_{N+1}\big( g_1(.), \pare \big).
\end{align}

     For any $x \in \mathbb{R}$ define the function $u(z,x): \mathbb{R} \to \mathbb{R}$ such that $u(z,x) = 1$, if and only if $z\geq x$, and otherwise we have   $u(z,x) = 0$. Define  a symmetric noise distribution $g_2(.)$, where for $z \geq 0$, we have
     \begin{align}\label{def_of_g_0}
         g_2(z) = \delta(z - (\eta-1)\Delta) \times \int_{\Delta}^{(\eta-1)\Delta} g_0(z) \,dz  + g_0(z)  \big(u(z,(\eta-1)\Delta) - u(z,(\eta+1)\Delta)\big),
     \end{align}
     where $\delta(z)$ is a Dirac's delta function. Note that $g_2(.)$ is a symmetric noise distribution and $g_2(.)=0$ for
     $|z| < (\eta-1)\Delta$,   
    $|z| > (\eta+1)\Delta$. We show that it satisfies \eqref{betternoise_acc} and \eqref{betternoise_mse}.
    Based on \eqref{general_acc_for_symmetric}, we have \begin{align}
        \mathsf{PA}_{N+1}\left( g_2(.), \pare \right) &= 2\int_{0}^{(\eta+1)\Delta}\Pr \big( \acce| \bn_a = z \big) g_2(z) \,dz, \nonumber \\
        &\overset{(a)}{=}2\int_{0}^{(\eta+1)\Delta}\Pr \big( \acce| \bn_a = z \big) \bigg( \delta(z - (\eta-1)\Delta) \times \int_{\Delta}^{(\eta-1)\Delta} g_0(z) \,dz\bigg) \,dz \nonumber\\
        &+2\int_{0}^{(\eta+1)\Delta}\Pr \big( \acce| \bn_a = z \big) \bigg( g_0(z)  \big(u(z,(\eta-1)\Delta) - u(z,(\eta+1)\Delta)\big)\bigg) \,dz
        \nonumber\\
        &= 2 \Pr \big( \acce| \bn_a = (\eta-1)\Delta \big) \int_{\Delta}^{(\eta-1)\Delta} g_0(z) \,dz + 2\int_{(\eta-1)\Delta}^{(\eta+1)\Delta}\Pr \big( \acce| \bn_a = z \big) g_0(z) \,dz
        \nonumber \\
        &\overset{(b)}{=} 2 \int_{\Delta}^{(\eta-1)\Delta} g_0(z) \,dz + 2\int_{(\eta-1)\Delta}^{(\eta+1)\Delta}\Pr \big( \acce| \bn_a = z \big) g_0(z) \,dz \nonumber \\
        & \overset{(c)}{=} \mathsf{PA}_{N+1}\left( g_0(.), \pare \right)\label{g_0_relation_to_g_2} \\
        &\overset{(d)}{\geq} \mathsf{PA}_{N+1}\left( g_1(.), \pare \right) \label{new_noise_same_acc}
    \end{align}
    where (a) follows from  \eqref{def_of_g_0}, and (b) follows from the fact that $\Pr \big( \acce| \bn_a = (\eta-1)\Delta \big) = 1$, and (c)
    follows from \eqref{semi_general_acc_symmetric}, and (d) follows from \eqref{new_noise_bounded_acc}.
    
    On the other hand, based on \eqref{mse_mean_general_proxy}, since $g_2(.)$ is a strong noise, we have
    \begin{align}\label{new_noise_same_mse_v2}
        &\mathsf{MSE}_{N+1}\big( g_2(.), \pare \big) \nonumber\\
        &=\frac{1}{2\mathsf{PA}_{N+1}\left( g_2(.), \pare \right)}\int_{0}^{(\eta+1)\Delta} \left(\int_{-\Delta}^{\Delta} (x+z)^2\Pr \big( \acce| \bn_a = z , \bn_{\min}= x)f_{\bn_{\min}}(x)\,dx \right)g_2(z) \,dz 
         \nonumber\\
         &\overset{(a)}{=} \frac{1}{2\mathsf{PA}_{N+1}\left( g_2(.), \pare \right)}\int_{0}^{(\eta+1)\Delta} \left(\int_{-\Delta}^{\Delta} (x+z)^2\Pr \big( \acce| \bn_a = z , \bn_{\min}= x)f_{\bn_{\min}}(x)\,dx \right)\bigg( \delta(z - (\eta-1)\Delta) \times \int_{\Delta}^{(\eta-1)\Delta} g_0(z) \,dz\bigg) \,dz
         \nonumber \\
         &+\frac{1}{2\mathsf{PA}_{N+1}\left( g_2(.), \pare \right)}\int_{0}^{(\eta+1)\Delta} \left(\int_{-\Delta}^{\Delta} (x+z)^2\Pr \big( \acce| \bn_a = z , \bn_{\min}= x)f_{\bn_{\min}}(x)\,dx \right)\bigg( g_0(z)  \big(u(z,(\eta-1)\Delta) - u(z,(\eta+1)\Delta)\big)\bigg) \,dz
         \nonumber\\
         &\overset{(b)}{=}\frac{1}{2\mathsf{PA}_{N+1}\left( g_2(.), \pare \right)}\int_{-\Delta}^{\Delta} (x+(\eta-1)\Delta)^2\Pr \big( \acce| \bn_a = (\eta-1)\Delta , \bn_{\min}= x)f_{\bn_{\min}}(x)\,dx \times \bigg(\int_{\Delta}^{(\eta-1)\Delta} g_0(z) \,dz\bigg)\nonumber \\
        &+\frac{1}{2\mathsf{PA}_{N+1}\left( g_2(.), \pare \right)}\int_{(\eta-1)\Delta}^{(\eta+1)\Delta} \left(\int_{z-\eta\Delta}^{\Delta} (x+z)^2f_{\bn_{\min}}(x)\,dx \right)g_0(z) \,dz \nonumber \\
         &\overset{(c)}{=}\frac{1}{2\mathsf{PA}_{N+1}\left( g_2(.), \pare \right)}\int_{-\Delta}^{\Delta} (x+(\eta-1)\Delta)^2f_{\bn_{\min}}(x)\,dx \int_{\Delta}^{(\eta-1)\Delta} g_0(z) \,dz \nonumber \\
        &+\frac{1}{2\mathsf{PA}_{N+1}\left( g_2(.), \pare \right)}\int_{(\eta-1)\Delta}^{(\eta+1)\Delta} \left(\int_{z-\eta\Delta}^{\Delta} (x+z)^2f_{\bn_{\min}}(x)\,dx \right)g_0(z) \,dz \nonumber \\
         &\geq \frac{1}{2\mathsf{PA}_{N+1}\left( g_2(.), \pare \right)}\int_{\Delta}^{(\eta-1)\Delta} \left(\int_{-\Delta}^{\Delta} (x+z)^2f_{\bn_{\min}}(x)\,dx \right)g_0(z) \,dz \nonumber \\
        &+\frac{1}{2\mathsf{PA}_{N+1}\left( g_2(.), \pare \right)}\int_{(\eta-1)\Delta}^{(\eta+1)\Delta} \left(\int_{z-\eta\Delta}^{\Delta} (x+z)^2f_{\bn_{\min}}(x)\,dx \right)g_0(z) \,dz 
        \nonumber \\
    &\overset{(d)}{=}\frac{1}{2\mathsf{PA}_{N+1}\left( g_2(.), \pare \right)}\int_{\Delta}^{(\eta-1)\Delta} \left(\int_{-\Delta}^{\Delta} (x+z)^2f_{\bn_{\min}}(x)\,dx \right)g_0(z) \,dz \nonumber \\
        &+\frac{1}{2\mathsf{PA}_{N+1}\left( g_0(.), \pare \right)}\int_{(\eta-1)\Delta}^{(\eta+1)\Delta} \left(\int_{z-\eta\Delta}^{\Delta} (x+z)^2f_{\bn_{\min}}(x)\,dx \right)g_0(z) \,dz 
        \nonumber \\
         &\overset{(e)}{=}  \mathsf{MSE}_{N+1}\big( g_0(.), \pare \big), \nonumber \\
         & \overset{(f)}{=}  \mathsf{MSE}_{N+1}\big( g_1(.), \pare \big)
    \end{align}
    where (a) follows from \eqref{def_of_g_0}, (b) follows from \eqref{exp_pos_z_v2}, (c) follows from the fact that for the case of $-\Delta \leq x \leq \Delta$, we have $\Pr \big( \acce| \bn_a = (\eta-1)\Delta , \bn_{\min}= x) = 1$, (d) follows from \eqref{g_0_relation_to_g_2}, (e) follows from \eqref{general_symmetric_mse}, (f) follows from \eqref{new_noise_bounded_mse}. 
    
Combining \eqref{new_noise_same_acc} and \eqref{new_noise_same_mse_v2} completes the proof of Lemma \ref{lemma:bounded_noise_existence}.

\section{Proof of Lemma \ref{lemma:exact_acc_noise_existence}}\label{proof:lemma:exact_acc_noise_existence}

Let $g_1(.)$ be a satisfying noise of Lemma \ref{lemma:bounded_noise_existence} and $\mathsf{PA}_{N+1} \left( g_1(.), \pare \right) = \alpha_1 > \alpha$.
    Define   satisfying noise of Lemma \ref{lemma:bounded_noise_existence} $g_2(.)$, where $g_2(z) = \frac{\alpha}{\alpha_1}g_1(z)$, for all $|z| \leq (\eta+1)\Delta$, $g_2((\eta+2)\Delta) = \frac{\alpha_1-\alpha}{2\alpha_1}\delta(z - (\eta+2)\Delta)$, and otherwise, $g_2(z) = 0$. Based on \eqref{general_symmetric_acc}, one can verify that
    \begin{align}\label{new_noise_exactprob_acc}
        \mathsf{PA}_{N+1} \left( g_2(.), \pare \right) = \frac{\alpha}{\alpha_1}\mathsf{PA}_{N+1} \left( g_1(.), \pare \right) = \alpha.
    \end{align}
    Additionally,
    based on \eqref{general_symmetric_mse}, we have
\begin{align}\label{new_noise_exactprob_mse}
    \mathsf{MSE}_{N+1}\big( g_2(.), \pare \big) = \mathsf{MSE}_{N+1}\big( g_1(.), \pare \big).
\end{align}
    Combining \eqref{new_noise_exactprob_acc} and \eqref{new_noise_exactprob_mse} completes the proof of Lemma \ref{lemma:exact_acc_noise_existence}.
\section{\texorpdfstring{Characterizing $h_{\eta, \ell}(z)$: An special Case}{X}}

In this appendix we characterize $h_{\eta, \ell}(z)$ and its concave envelop, for the specific case of $\mathbf{n}_h \sim \text{unif}[-\Delta, \Delta]$. Recall that $k_{\eta, \ell}(z) \triangleq \int_{z-\eta\Delta}^{\Delta} f_{\bn_{\min}}(x)  \,dx$, where $z \in [(\eta-1)\Delta, (\eta+1)\Delta]$.
Note that
\begin{align}
  k_{\eta, \ell}(z) &= F_{\bn_{\min}} (\Delta) - F_{\bn_{\min}} (z - \eta \Delta)  = 1 -  F_{\bn_{\min}} (z - \eta \Delta) \nonumber \\
  & = (1-F_{\bn_h}(z-\eta \Delta))^{\ell} = \left(\frac{(\eta+1)\Delta - z}{2\Delta}\right)^{\ell}
\end{align}

Thus, for $0 \leq q \leq 1$ we have
\begin{align}
    k^{-1}_{\eta, \ell}(q) = \Delta(\eta+1-2q^{\frac{1}{\ell}}).
\end{align}
On the other hand, we have
\begin{align}
    \nu_{\eta, \ell}(z) &= \int_{z-\eta\Delta}^{\Delta} (x+z)^2f_{\bn_{\min}}(x)\,dx = \int_{z-\eta\Delta}^{\Delta} (x+z)^2 \ell f_{\bn_{h}}(x) (1 - F_{\bn_h}(x))^{\ell - 1}\,dx \nonumber \\
    &= \frac{\ell}{(2\Delta)^{\ell}} \int_{z-\eta\Delta}^{\Delta} (x+z)^2  (\Delta-x)^{\ell - 1}\,dx \nonumber \\
    &= \frac{-\ell}{(2\Delta)^{\ell}} ((\eta+1)\Delta - z)^\ell \left( \frac{2((\eta+1)\Delta -z)(z+\Delta)}{\ell + 1} - \frac{((\eta+1)\Delta-z)^2}{\ell + 2} - \frac{(z+\Delta)^2}{\ell}\right)
\end{align}
Therefore, we have
\begin{align}
    h_{\eta, \ell}(q) &= \nu_{\eta, \ell}(k_{\eta, \ell}^{-1} (q)) = \Delta^2\ell q \left( 
    \frac{-4q^{\frac{1}{\ell}}(\eta+2 - 2q^{\frac{1}{\ell}})}{\ell+1} + \frac{4q^{\frac{2}{\ell}}}{\ell + 2} + \frac{ (\eta+2-2q^{\frac{1}{\ell}})^2}{\ell}
    \right).
\end{align}
This implies that
\begin{align}\label{second_derivitive}
    \frac{d^2}{dq^2}h_{\eta, \ell}(q) = \frac{\Delta^2 q^{\frac{-\ell +1}{\ell}} }{\ell^2 (1 + \ell)} \left(16 \left(1 + 2\ell(2 + \ell)\right) q^{\frac{1}{\ell}} - 4 (1 + \ell)(1 + 2\ell)(2 + \eta)\right).
\end{align}
Based on \eqref{second_derivitive}, one can verify that for the case $\eta \geq 2 + \frac{4}{\ell+1} - \frac{4}{2\ell+1}$, and $0 \leq q \leq 1$, the function $h_{\eta, \ell} (q)$ is concave, leading to $h^*_{\eta, \ell} (q) = h_{\eta, \ell} (q)$. On the other hand, for the case of $2 \leq \eta < 2 + \frac{4}{\ell+1} - \frac{4}{2\ell+1}$, based on \eqref{second_derivitive}, one can easily verify that for the case of $\left( (\eta+2)\frac{(1+\ell)(1+2\ell)}{4(2\ell^2+4\ell+1)}\right)^{\ell} < q \leq 1$,  $h_{\eta, \ell}(q)$ is convex, and for the case of $0 \leq q \leq \left( (\eta+2)\frac{(1+\ell)(1+2\ell)}{4(2\ell^2+4\ell+1)}\right)^{\ell}$, it is concave. In order to find the concave envelop, i.e., $h_{\eta, \ell}^*(q)$ for this case, we draw a line passing through the point $\big(1,  h_{\eta, \ell}(1) \big) $  and find the point of tangency with $h_{\eta, \ell}(q)$. More precisely, we need to solve the following equation for $q$:
\begin{align}\label{tengency_equation}
  h_{\eta, \ell}(q) - h_{\eta, \ell}(1) = \big(\frac{d}{dq}h_{\eta, \ell}(q)\big)  (q-1).
\end{align}
Using this, one can easily calculate the concave envelop using the details of the proof of Lemma \ref{lemma:lower_bound_of_mean}.

\section{Proof of Remark \ref{remark:special_case_uniform}}\label{proof:remark:special_case_uniform}
To prove Remark~\ref{remark:special_case_uniform}, we first note that Lemma~\ref{lemma:bounded_noise_existence} implies that if, for the case \( |z| < \Delta \), the adversary’s noise distribution satisfies \( g(z) = 0 \), then Theorem~\ref{theorem:CofJ_N_ is same} follows. In order to justify setting \( g(z) = 0 \) for \( |z| < \Delta \), we must show that for any symmetric noise distribution \( g(z) \) that places nonzero mass within the interval \((- \Delta, \Delta)\), if the adversary instead opts for a symmetric noise distribution that allocates equal mass to the endpoints \( \pm \Delta \), then the mean squared error increases. Consequently, it is optimal to assume that for \( |z| < \Delta \), we have \( g(z) = 0 \). Once this claim is established, Theorem~\ref{theorem:CofJ_N_ is same} follows directly from Lemma~\ref{lemma:bounded_noise_existence}. To prove this claim, we demonstrate a more general statement as follows.

Let $\bn_1, \bn_2, \dots, \bn_k$ be independent and identically distributed (i.i.d.) random variables drawn from the uniform distribution on the interval $[-\Delta, \Delta]$. Additionally, let $\bn_a$ be a constant, i.e., $\bn_a = z$ with probability $1$, where $0 \leq z \leq \Delta$. Define 
\[
\be \triangleq \frac{\max \{ \bn_1, \bn_2, \dots, \bn_k, \bn_a \} + \min \{ \bn_1, \bn_2, \dots, \bn_k, \bn_a \}}{2}.
\]
We demonstrate that $\mathbb{E}[\be^2]$ attains its maximum when $z = \Delta$. 

Let $\br = \max\{ |\bn_1|, |\bn_2|, \dots, |\bn_k| \}$.
Using the law of total expectation, we write
\[
\mathbb{E}[\be^2] = \mathbb{E}_{\br}\left[\mathbb{E}[\be^2 \mid r]\right].
\]
To achieve this, we must determine both the probability density function $f_{\br}(r)$ and the conditional expectation $\mathbb{E}[\be^2 \mid r]$. We begin by characterizing $f_{\br}(r)$. Note that each $\bn_i$ is uniformly distributed over $[-\Delta,\Delta]$. Consequently, for any $0 \leq r \leq \Delta$,
\[
\Pr\bigl(|\bn_i|\le r\bigr) \;=\; \frac{2r}{2\Delta} \;=\; \frac{r}{\Delta}.
\]
Since the $\bn_i$ are i.i.d., the CDF of $\br$ is
\[
F_{\br}(r) \;=\; \Pr\bigl(\br \le r\bigr)
\;=\; \Pr\bigl(|\bn_1|\le r,\dots,|\bn_k|\le r\bigr)
\;=\;\Bigl(\frac{r}{\Delta}\Bigr)^k,
\quad 0 \le r \le \Delta.
\]
For $r<0$, $F_{\br}(r) = 0$, and for $r>\Delta$, $F_{\br}(r) = 1$. Differentiating the CDF with respect to $r$ yields the PDF:
\begin{align} \label{pdf_of_r}
    f_{\br}(r) \;=\; \frac{d}{dr}\,F_{\br}(r)
\;=\;\frac{k\,r^{k-1}}{\Delta^k},
\quad 0 \le r \le \Delta,
\end{align}
and $f_{\br}(r) = 0$ otherwise. 

Next we begin to characterize the conditional expectation $\mathbb{E}[\be^2 \mid r]$. Let
\[
B_i \;=\; \bigl\{\text{the event that }|\bn_i|\text{ is the maximum among }|\bn_1|,\dots,|\bn_k|\bigr\}.
\]
Because $\bn_1, \bn_2, \dots, \bn_k$ are i.i.d.\ and symmetrically distributed around zero, each event $B_i$ occurs with equal probability:
\[
\Pr\bigl(B_i\bigr) \;=\; \frac{1}{k}.
\]
Furthermore, by symmetry,
\[
\mathbb{E}\bigl[\be^2 \,\big\vert\, r, B_1\bigr]
\;=\;\mathbb{E}\bigl[\be^2 \,\big\vert\, r, B_2\bigr]
\;=\;\dots
\;=\;\mathbb{E}\bigl[\be^2 \,\big\vert\, r, B_k\bigr].
\]
Hence,
\[
\mathbb{E}[\be^2 \mid r]
\;=\;\sum_{i=1}^k \Pr(B_i \mid r)\,\mathbb{E}[\be^2 \mid r, B_i].
\]
Since $\Pr(B_i \mid r) = 1/k$, we obtain
\[
\mathbb{E}[\be^2 \mid r]
\;=\;\frac{1}{k}\,\sum_{i=1}^k \mathbb{E}[\be^2 \mid r, B_i]
\;=\;\mathbb{E}[\be^2 \mid r, B_1].
\]

\noindent
\textbf{Without loss of generality, assume that $B_1$ holds.} 
In other words, we take $\bn_1$ to be the variable with the maximum absolute value, i.e.\ $|\bn_1| = r$ and $|\bn_1| \geq |\bn_j|$ for all $j \neq 1$. Under this assumption, we can now proceed to explicitly calculate $\mathbb{E}[\be^2 \mid r, B_1]$. We claim that under the event $\{B_1,\,|\bn_1| = r\}$, each $\bn_j$ for $j \neq 1$ is uniformly distributed on $[-r,r]$. 
To prove this claim, we start by computing the conditional cumulative distribution function (CDF) of $\bn_j$ for $j\neq 1$, given that $|\bn_j| \le r$. Note that
\[
f_{\bn_j}(x) = \frac{1}{2\Delta}, \quad x \in [-\Delta, \Delta],
\]
and the probability that $|\bn_j|\le r$ is 
\[
\Pr\bigl(|\bn_j|\le r\bigr) = \frac{2r}{2\Delta} = \frac{r}{\Delta}.
\]
For any $x\in[-r,r]$, the conditional CDF of $\bn_j$ is given by
\[
F_{\bn_j \mid |\bn_j|\le r}(x) = \frac{\Pr\bigl(-r\le \bn_j\le x\bigr)}{\Pr\bigl(|\bn_j|\le r\bigr)}.
\]
Since $\bn_j$ is uniform on $[-\Delta, \Delta]$, we have
\[
\Pr\bigl(-r\le \bn_j\le x\bigr) = \frac{x - (-r)}{2\Delta} = \frac{x + r}{2\Delta}.
\]
Thus, the conditional CDF becomes
\[
F_{\bn_j \mid |\bn_j|\le r}(x) = \frac{\frac{x+r}{2\Delta}}{\frac{r}{\Delta}} = \frac{x+r}{2r}, \quad \text{for } -r\le x\le r.
\]
Differentiating this CDF with respect to $x$, we obtain the conditional PDF:
\[
f_{\bn_j \mid |\bn_j|\le r}(x) = \frac{d}{dx}\left(\frac{x+r}{2r}\right) = \frac{1}{2r}, \quad -r\le x\le r.
\]
This is precisely the density function of a uniform distribution on $[-r, r]$. Hence, under the event $\{B_1,\,|\bn_1|=r\}$, each $\bn_j$ for $j\neq 1$ is uniformly distributed on $[-r, r]$. 

Next, we define two additional events to capture the sign of $\bn_1$:
\[
A^+ \;=\;\{\bn_1 = r\}
\quad\text{and}\quad
A^- \;=\;\{\bn_1 = -r\}.
\]
Due to symmetry, we have
\[
\Pr\bigl(A^+ \mid r, B_1\bigr) 
\;=\;\Pr\bigl(A^- \mid r, B_1\bigr) 
\;=\;\tfrac{1}{2}.
\]
Hence,
\begin{align}\label{main_eq_expectation}
    \mathbb{E}\bigl[\be^2 \,\big\vert\, r, B_1\bigr]
\;=\;\frac12\,\mathbb{E}\bigl[\be^2 \,\big\vert\, r, B_1, A^+\bigr]
\;+\;\frac12\,\mathbb{E}\bigl[\be^2 \,\big\vert\, r, B_1, A^-\bigr].
\end{align}
Intuitively, $A^+$ and $A^-$ distinguish whether $\bn_1$ attains its value at $+r$ or $-r$. In either case, the distribution of the other noise variables $\bn_j$ (for $j \neq 1$) remains uniform on $[-r,r]$, as established previously. To fully characterize $\mathbb{E}[\be^2]$, we must evaluate the conditional expectations $\mathbb{E}[\be^2 \mid r, B_1, A^+]$ and $\mathbb{E}[\be^2 \mid r, B_1, A^-]$ for the different configurations of $\bn_a$ and the remaining $\bn_j$'s. We now proceed to analyze these cases in detail. Recall that the adversarial input is $\bn_a = z$ with $z > r$, and that under the event $B_1$, the remaining noise variables $\bn_2, \dots, \bn_k$ are uniformly distributed on $[-r, r]$. We now consider the following case.

\medskip

\noindent
\textbf{Case 1:} $z > r$.

\begin{enumerate}
    \item \textbf{Subcase 1A:} Suppose $A^-$ holds, i.e., $\bn_1 = -r$. Then, by definition, since $\bn_a = z > r$ and $-r < r$, we have:
    \[
    \max\{\bn_1,\bn_2,\dots,\bn_k,\bn_a\} = z \quad \text{and} \quad \min\{\bn_1,\bn_2,\dots,\bn_k,\bn_a\} = -r.
    \]
    Thus, the estimator is given by
    \[
    \be = \frac{z + (-r)}{2} = \frac{z - r}{2},
    \]
    and therefore,
    \begin{align}\label{case1:z>r}
        \mathbb{E}\bigl[\be^2 \,\big\vert\, r, B_1, A^-, z>r\bigr] = \left(\frac{z - r}{2}\right)^2.
    \end{align}
    
    \item \textbf{Subcase 1B:} Now, suppose $A^+$ holds, i.e., $\bn_1 = r$. In this situation, the maximum is still $z$ (since $z > r$), but the minimum is now determined by the other noise variables:
    \[
    \bm = \min\{\bn_2,\bn_3,\dots,\bn_k\}.
    \]
    Under the event $B_1$, as established earlier, each $\bn_j$ for $j\neq 1$ is uniformly distributed on $[-r, r]$. We first characterize the distribution of $\bm$. For $x \in [-r, r]$, the CDF of any $\bn_j$ is
    \[
    F_{\bn_j}(x) = \frac{x + r}{2r}.
    \]
    Since $\bm$ is the minimum of $(k-1)$ i.i.d.\ random variables, its CDF is
    \[
    F_\bm(x) = 1 - \Pr(\bm > x) = 1 - \left(1 - F_{\bn_j}(x)\right)^{k-1} = 1 - \left(\frac{r - x}{2r}\right)^{k-1},
    \]
    and differentiating with respect to $x$ gives the PDF:
    \[
    f_\bm(x) = \frac{d}{dx}F_\bm(x) = (k-1) \cdot \frac{1}{2r} \left(\frac{r - x}{2r}\right)^{k-2}, \quad x \in [-r, r].
    \]
    In Subcase 1B, the estimator becomes
    \[
    \be = \frac{\max\{r,\, \dots, \bn_k,\, z\} + \min\{\bm, z\}}{2} = \frac{z + \bm}{2},
    \]
    since $z > r$ implies that the maximum is $z$, and the minimum is $m$. Therefore, 
    \[
    \be^2 = \left(\frac{z + \bm}{2}\right)^2.
    \]
    The conditional expectation in this subcase is then obtained by integrating over the possible values of $\bm$:
    \begin{align}\label{case2:z>r}
        \mathbb{E}\left[\be^2 \,\big|\, r, B_1, A^+, z>r\right] &= \int_{-r}^{r} \left(\frac{z + x}{2}\right)^2 f_\bm(x) \, dx \nonumber \\
        &=\int_{-r}^{r} \left(\frac{z + x}{2}\right)^2 (k-1) \cdot \frac{1}{2r} \left(\frac{r - x}{2r}\right)^{k-2} \, dx \nonumber \\
        & = \frac{1}{4}(r + z)^2 + \frac{(k-1)r^2}{k+1} + \frac{(k-1)r(r+z)}{k}
    \end{align}
\end{enumerate}
  Finally, based on \eqref{main_eq_expectation}, \eqref{case1:z>r}, and \eqref{case2:z>r}, we have
    \begin{align}\label{z>r}
        \mathbb{E}\bigl[\be^2 \,\big\vert\, r,\,z>r\bigr] = \mathbb{E}\left[\be^2 \,\big|\, r, B_1, z>r\right] = \frac{r^2+z^2}{4}+ \frac{(k-1)r^2}{2(k+1)} + \frac{(k-1)r(r+z)}{2k}.
    \end{align}

\noindent
\textbf{Case 2:} $z \leq r.$

\begin{enumerate}
    \item \textbf{Subcase 2A:} $A^+$ holds, i.e., $\bn_1 = r$. Since $z \le r$, the maximum among $\{\bn_1,\dots,\bn_k,\bn_a\}$ is $r$, and the minimum is $\min\{\bm,z\}$, where $\bm \;=\;\min\{\bn_2,\bn_3,\dots,\bn_k\}$. Therefore,
    \[
    \be 
    \;=\;\frac{\max\{r,z\} + \min\{\bm,z\}}{2}
    \;=\;\frac{r + \min\{\bm,z\}}{2}.
    \]
    To compute 
    \[
    \mathbb{E}\bigl[\be^2 \,\big\vert\, r,B_1,A^+, z \le r\bigr],
    \]
    we split the integral over $\bm$ at $m=z$. Specifically,
    \[
    \mathbb{E}\bigl[\be^2 \,\big\vert\, r,B_1,A^+, z \le r\bigr]
    \;=\;
    \int_{-r}^{z} 
    \left(\frac{r + x}{2}\right)^2 f_\bm(x)\,dx
    \;+\;
    \int_{z}^{r}
    \left(\frac{r + z}{2}\right)^2 f_\bm(x)\,dx.
    \]
    In the first integral, $m = x < z$, so $\min\{\bm,z\} = x$. In the second integral, $m = x \ge z$, so $\min\{\bm,z\} = z$. By substituting the PDF $f_\bm(x) = (k-1) \cdot \frac{1}{2r} \left(\frac{r - x}{2r}\right)^{k-2}$ into these integrals and evaluating them, we have
    \begin{align}\label{case1:z<r}
        \mathbb{E}\bigl[\be^2 \,\big\vert\, r,B_1,A^+, z \le r\bigr] &= \int_{-r}^{z} 
    \left(\frac{r + x}{2}\right)^2 f_\bm(x)\,dx +
    \int_{z}^{r}
    \left(\frac{r + z}{2}\right)^2 f_\bm(x)\,dx \nonumber \\
    &=\int_{-r}^{z} 
    \left(\frac{r + x}{2}\right)^2 (k-1) \frac{1}{2r} \left(\frac{r - x}{2r}\right)^{k-2}\,dx +
    \int_{z}^{r}
    \left(\frac{r + z}{2}\right)^2 (k-1)  \frac{1}{2r} \left(\frac{r - x}{2r}\right)^{k-2}\,dx \nonumber \\
    & = W(r,z), 
    \end{align}
    where we have
    \begin{align}\label{definition_of_w}
       W(r,z) =  \frac{k-1}{4}\cdot \frac{1}{(2r)^{k-1}}\bigg[& 
    \frac{4r^2}{k-1}[(2r^{k-1}) - (r-z)^{k-1}] + \frac{1}{k+1}[(2r^{k+1}) - (r-z)^{k+1}] - \nonumber \\
    &\frac{4}{k}[(2r^{k}) - (r-z)^{k}] + \frac{(r+z)^2}{k-1}(r-z)^{k-1} \bigg].
    \end{align}

    \item \textbf{Subcase 2B:} $A^-$ holds, i.e., $\bn_1 = -r$. Since $|\bn_1|=r$ is the largest absolute value under $B_1$, each $\bn_j$ for $j\neq 1$ is uniformly distributed on $[-r,r]$. Here, the minimum of all variables $\{\bn_1,\bn_2,\dots,\bn_k,\bn_a\}$ is clearly $-r$, and we must identify the maximum. Define
\[
\bc \;=\;\max\{\bn_2,\bn_3,\dots,\bn_k\}.
\]
Because each $\bn_j$ for $j\neq 1$ is uniform on $[-r,r]$, the CDF of $\bc$ for $x \in [-r,r]$ is
\[
F_{\bc}(x)
\;=\;\Pr\bigl(\bc \le x\bigr)
\;=\;\Pr\bigl(\bn_2 \le x,\dots,\bn_k \le x\bigr)
\;=\;\bigl(\Pr(\bn_j \le x)\bigr)^{k-1}.
\]
Since
\[
\Pr(\bn_j \le x) 
\;=\;\frac{x - (-r)}{2r} 
\;=\;\frac{x + r}{2r},
\quad \text{for } x \in [-r,r],
\]
we obtain
\[
F_{\bc}(x)
\;=\;\left(\frac{x + r}{2r}\right)^{k-1},
\quad x \in [-r,r].
\]
Differentiating gives the PDF:
\[
f_{\bc}(x)
\;=\;\frac{d}{dx}F_{\bc}(x)
\;=\;(k-1)\,\frac{1}{2r}\,\left(\frac{x + r}{2r}\right)^{k-2},
\quad x \in [-r,r].
\]

\medskip

In this subcase, the estimator $\be$ is
\[
\be 
\;=\;\frac{\min\{\bn_1,\dots,\bn_k,\bn_a\} + \max\{\bn_1,\dots,\bn_k,\bn_a\}}{2}
\;=\;\frac{-r + \max\{\bc,\,z\}}{2},
\]
because $\bn_1 = -r$ is the minimum, and the maximum is $\max\{\bc,z\}$. To compute 
\[
\mathbb{E}\bigl[\be^2 \,\big\vert\, r,B_1,A^-, z \le r\bigr],
\]
we split the integration over $\bc$ at $\bc = z$. Concretely,
\begin{align}\label{case2:z<r}
    \mathbb{E}\bigl[\be^2 \,\big\vert\, r,B_1,A^-, z \le r\bigr] &\overset{(a)}{=} \int_{-r}^{z}
\left(\frac{-r + z}{2}\right)^2 f_{\bc}(x)\,dx
 + 
\int_{z}^{r}
\left(\frac{-r + x}{2}\right)^2 f_{\bc}(x)\,dx \nonumber \\
& = \int_{-r}^{z}
\left(\frac{-r + z}{2}\right)^2 \frac{k-1}{2r}\left(\frac{x + r}{2r}\right)^{k-2}\,dx
 + 
\int_{z}^{r}
\left(\frac{-r + x}{2}\right)^2 \frac{k-1}{2r}\left(\frac{x + r}{2r}\right)^{k-2}\,dx \nonumber \\
& \overset{(b)}{=} W(r,-z), 
\end{align}
where (a) follows from the fact that for $x < z$, we have  $\max\{\bc,z\} = z$, and  for $x \ge z$, we have  $\max\{\bc,z\} = x$, and (b) follows from the definition of  $W(.,.)$ in \eqref{definition_of_w}.
\end{enumerate}
Combining \eqref{case1:z<r} and \eqref{case2:z<r}, we have
\begin{align}\label{z<r}
    \mathbb{E}\bigl[\be^2 \,\big\vert\, r,\,z \leq r\bigr] = \mathbb{E}\left[\be^2 \,\big|\, r, B_1, z\leq r\right] &= \frac{1}{2}\mathbb{E}\bigl[\be^2 \,\big\vert\, r,B_1,A^+, z \le r\bigr] + \frac{1}{2}\mathbb{E}\bigl[\be^2 \,\big\vert\, r,B_1,A^-, z \le r\bigr] \nonumber \\
    &= \frac{W(r,z) + W(r,-z)}{2},
\end{align}
where $W(.,.)$ defined in \eqref{definition_of_w}.

Now we proceed to calculate $\mathbb{E}[\be^2] $. 
Note that 
\begin{align}\label{exp[e]}
    \mathbb{E}[\be^2] & = \mathbb{E}_{\br}\bigl[\mathbb{E}[\be^2 \,\big\vert\, \br]\bigr]\nonumber \\
    &=\int_{0}^{\Delta} \mathbb{E}\bigl[\be^2 \,\big\vert\, r\bigr]\; f_{\br}(r)\;dr \nonumber \\
    & = \int_{0}^{z} \mathbb{E}\bigl[\be^2 \,\big\vert\, r,\,r < z\bigr]\; f_{\br}(r)\;dr + \int_{z}^{\Delta} \mathbb{E}\bigl[\be^2 \,\big\vert\, r,\,r \ge z\bigr]\; f_{\br}(r)\;dr \nonumber \\
    & \overset{(a)}{=} \int_{0}^{z} \mathbb{E}\bigl[\be^2 \,\big\vert\, r,\,r < z\bigr]\; \frac{k\,r^{k-1}}{\Delta^k}\;dr + \int_{z}^{\Delta} \mathbb{E}\bigl[\be^2 \,\big\vert\, r,\,r \ge z\bigr]\; \frac{k\,r^{k-1}}{\Delta^k}\;dr \nonumber \\
    &\overset{(b)}{=} \int_{0}^{z} \left(  \frac{r^2+z^2}{4}+ \frac{(k-1)r^2}{2(k+1)} + \frac{(k-1)r(r+z)}{2k}\right) \frac{k\,r^{k-1}}{\Delta^k}\;dr \nonumber \\
    &+ \int_{z}^{\Delta} \left( \frac{W(r,z) + W(r,-z)}{2}\right) \frac{k\,r^{k-1}}{\Delta^k}\;dr \nonumber \\
    & = \frac{2\Delta^2}{(k+1)(k+2)} - \frac{(\Delta^2 - z^2)\left( (\Delta+z)^k + (\Delta-z)^k \right)}{\Delta^k (k+1) 2^{k+1}},
\end{align}
where (a) follows from \eqref{pdf_of_r}, (b) follows from \eqref{z>r} and \eqref{z<r}. Based on \eqref{exp[e]}, one can easily verify that the maximum amount of $\mathbb{E}[\be^2]$ is when $z = \Delta$.


\begin{thebibliography}{10}

\bibitem{SudanBook}
V.~Guruswami, A.~Rudra, and M.~Sudan, {\em Essential Coding Theory}.
\newblock Draft is Available, 2022.

\bibitem{yu2019lagrange}
Q.~Yu, S.~Li, N.~Raviv, S.~M.~M. Kalan, M.~Soltanolkotabi, and S.~A. Avestimehr, ``Lagrange coded computing: Optimal design for resiliency, security, and privacy,'' in {\em The 22nd International Conference on Artificial Intelligence and Statistics}, pp.~1215--1225, PMLR, 2019.

\bibitem{roth2020analog}
R.~M. Roth, ``Analog error-correcting codes,'' {\em IEEE Transactions on Information Theory}, vol.~66, no.~7, pp.~4075--4088, 2020.

\bibitem{buterin2013ethereum}
V.~Buterin {\em et~al.}, ``Ethereum white paper,'' {\em GitHub repository}, vol.~1, pp.~22--23, 2013.

\bibitem{ruoti2019sok}
S.~Ruoti, B.~Kaiser, A.~Yerukhimovich, J.~Clark, and R.~Cunningham, ``{SoK}: Blockchain technology and its potential use cases,'' {\em arXiv preprint arXiv:1909.12454}, 2019.

\bibitem{croman2016scaling}
K.~Croman, C.~Decker, I.~Eyal, A.~E. Gencer, A.~Juels, A.~Kosba, A.~Miller, P.~Saxena, E.~Shi, E.~G{\"u}n~Sirer, {\em et~al.}, ``On scaling decentralized blockchains: (a position paper),'' in {\em International conference on financial cryptography and data security}, pp.~106--125, Springer, 2016.

\bibitem{zhao2021veriml}
L.~Zhao, Q.~Wang, C.~Wang, Q.~Li, C.~Shen, and B.~Feng, ``{VeriML}: Enabling integrity assurances and fair payments for machine learning as a service,'' {\em IEEE Transactions on Parallel and Distributed Systems}, vol.~32, no.~10, pp.~2524--2540, 2021.

\bibitem{thaler2022proofs}
J.~Thaler {\em et~al.}, ``Proofs, arguments, and zero-knowledge,'' {\em Foundations and Trends{\textregistered} in Privacy and Security}, vol.~4, no.~2--4, pp.~117--660, 2022.

\bibitem{liu2021zkcnn}
T.~Liu, X.~Xie, and Y.~Zhang, ``{ZkCNN}: Zero knowledge proofs for convolutional neural network predictions and accuracy,'' in {\em Proceedings of the 2021 ACM SIGSAC Conference on Computer and Communications Security}, pp.~2968--2985, 2021.

\bibitem{xing2023zero}
Z.~Xing, Z.~Zhang, J.~Liu, Z.~Zhang, M.~Li, L.~Zhu, and G.~Russello, ``Zero-knowledge proof meets machine learning in verifiability: A survey,'' {\em arXiv preprint arXiv:2310.14848}, 2023.

\bibitem{mohassel2017secureml}
P.~Mohassel and Y.~Zhang, ``{SecureML}: A system for scalable privacy-preserving machine learning,'' in {\em 2017 IEEE symposium on security and privacy (SP)}, pp.~19--38, IEEE, 2017.

\bibitem{lee2024vcnn}
S.~Lee, H.~Ko, J.~Kim, and H.~Oh, ``{vCNN}: Verifiable convolutional neural network based on zk-snarks,'' {\em IEEE Transactions on Dependable and Secure Computing}, 2024.

\bibitem{garg2023experimenting}
S.~Garg, A.~Goel, S.~Jha, S.~Mahloujifar, M.~Mahmoody, G.-V. Policharla, and M.~Wang, ``Experimenting with zero-knowledge proofs of training,'' {\em Cryptology ePrint Archive}, 2023.

\bibitem{weng2021mystique}
C.~Weng, K.~Yang, X.~Xie, J.~Katz, and X.~Wang, ``Mystique: Efficient conversions for {Zero-Knowledge} proofs with applications to machine learning,'' in {\em 30th USENIX Security Symposium (USENIX Security 21)}, pp.~501--518, 2021.

\bibitem{chen2022interactive}
S.~Chen, J.~H. Cheon, D.~Kim, and D.~Park, ``Interactive proofs for rounding arithmetic,'' {\em IEEE Access}, vol.~10, pp.~122706--122725, 2022.

\bibitem{garg2022succinct}
S.~Garg, A.~Jain, Z.~Jin, and Y.~Zhang, ``Succinct zero knowledge for floating point computations,'' in {\em Proceedings of the 2022 ACM SIGSAC Conference on Computer and Communications Security}, pp.~1203--1216, 2022.

\bibitem{setty2012taking}
S.~Setty, V.~Vu, N.~Panpalia, B.~Braun, A.~J. Blumberg, and M.~Walfish, ``Taking proof-based verified computation a few steps closer to practicality,'' in {\em 21st USENIX Security Symposium (USENIX Security 12)}, pp.~253--268, 2012.

\bibitem{nodehi2024game}
H.~A. Nodehi, V.~R. Cadambe, and M.~A. Maddah-Ali, ``Game of coding: Beyond trusted majorities,'' {\em arXiv preprint arXiv:2401.16643}, 2024.

\bibitem{eskandari2021sok}
S.~Eskandari, M.~Salehi, W.~C. Gu, and J.~Clark, ``{SoK}: Oracles from the ground truth to market manipulation,'' in {\em Proceedings of the 3rd ACM Conference on Advances in Financial Technologies}, pp.~127--141, 2021.

\bibitem{breidenbach2021chainlink}
L.~Breidenbach, C.~Cachin, B.~Chan, A.~Coventry, S.~Ellis, A.~Juels, F.~Koushanfar, A.~Miller, B.~Magauran, D.~Moroz, {\em et~al.}, ``Chainlink 2.0: Next steps in the evolution of decentralized oracle networks,'' {\em Chainlink Labs}, vol.~1, pp.~1--136, 2021.

\bibitem{benligiray2020decentralized}
B.~Benligiray, S.~Milic, and H.~V{\"a}nttinen, ``Decentralized {API}s for web 3.0,'' {\em API3 Foundation Whitepaper}, 2020.

\end{thebibliography}
\end{document}